\documentclass{scrartcl}
\usepackage{amsfonts}
\usepackage{amssymb,amsmath,amsthm}
\usepackage{mathtools}
\usepackage{enumerate}
\usepackage{enumitem}
\usepackage{graphicx}
\usepackage{nicefrac}
\usepackage{color}
\usepackage[all]{xy}
\usepackage{hyperref}
\usepackage{bbm}
\usepackage{tikz}
\usepackage{tikz-cd}
\usetikzlibrary{tikzmark,calc,decorations.pathreplacing}
\usepackage{cancel}
\usepackage{xcolor}
\usepackage{booktabs}
\usepackage{multirow}
\usepackage{subcaption}
\usetikzlibrary{positioning}
\usetikzlibrary{arrows}
\usepackage{wasysym}
\usepackage{pifont}
\usepackage[normalem]{ulem}
\usepackage{comment}

\definecolor{db}{RGB}{66, 126, 147}

\newtheorem{theorem}{Theorem}
\newtheorem{proposition}{Proposition}
\newtheorem{corollary}{Corollary}
\newtheorem{lemma}{Lemma}
\newtheorem{definition}{Definition}
\newtheorem{problem}{Problem}
\newtheorem*{problem*}{Problem}
\newtheorem*{question*}{Question}
\newtheorem{example}{Example}

\newcommand{\xmark}{\ding{55}}

\newcommand{\zz}{\mathbb Z}

\newcommand{\R}{\mathbb R}

\newcommand{\C}{\mathbb C}

\newcommand{\tr}{\mathrm{tr}}

\newcommand{\one}{\mathbbm{1}}
\newcommand{\One}{\mathbb{I}}

\newcommand{\Sym}{\mathrm{Sym}}

\newcommand{\Val}{\mathrm{Val}}

\newcommand\mc[1]{\mathcal{#1}}

\newcommand{\Ad}{\mathrm{Ad}}

\newcommand{\spec}{\mathrm{sp}}

\newcommand{\im}{{\mathrm{im}}}
\newcommand{\sa}{{\mathrm{sa}}}
\newcommand{\supp}{{\mathrm{supp}}}

\newcommand{\fl}{\mathfrak{l}}

\newcommand{\cG}{\mc{G}}
\newcommand{\cI}{\mc{I}}
\newcommand{\cM}{\mc{M}}
\newcommand{\cP}{\mc{P}}
\newcommand{\cS}{\mc{S}}

\newcommand{\cX}{\mc{X}}

\newcommand{\rC}{\mathrm{C}}
\newcommand{\trC}{{\widetilde{\rC}}}
\newcommand{\rw}{\mathrm{w}}

\newcommand{\GO}{{G(\cO)}}

\newcommand{\cA}{{\mc{A}}}
\newcommand{\cAsa}{{\cA_\mathrm{sa}}}
\newcommand{\cB}{{\mc{B}}}
\newcommand{\BC}{{\cB(C)}}
\newcommand{\BO}{{\cB(\cO)}}
\newcommand{\cC}{{\mc{C}}}
\newcommand{\tC}{{\widetilde{C}}}
\newcommand{\ttC}{{\widetilde{\widetilde{C}}}}
\newcommand{\tG}{{\widetilde{G}}}
\newcommand{\cO}{{\mc{O}}}
\newcommand{\rcO}{{\cO^\circ}}
\newcommand{\cOdown}{{\cO_\cap}}
\newcommand{\cOup}{{\cO_\cup}}
\newcommand{\tO}{{\widetilde{O}}}
\newcommand{\tcO}{{\widetilde{\cO}}}
\newcommand{\cOr}{{\cO_\mathrm{cg}}}
\newcommand{\CO}{{\cC(\cO)}}
\newcommand{\sCO}{{\cC(\cO^*)}}
\newcommand{\tCO}{{\cC(\tcO)}}
\newcommand{\rCO}{{\cC(\cO^\circ)}}
\newcommand{\COdown}{{\cC_\cap(\cO)}}
\newcommand{\COup}{{\cC_\cup(\cO)}}
\newcommand{\COm}{{\cC_\mathrm{max}(\cO)}}
\newcommand{\COmin}{{\cC_\mathrm{min}(\cO^\circ)}}
\newcommand{\rCOm}{{\cC_\mathrm{max}(\cO^\circ)}}
\newcommand{\sCOm}{{\cC_\mathrm{max}(\cO^*)}}

\newcommand{\CH}{{\cC(\cH)}}
\newcommand{\CHm}{{\cC_\mathrm{max}(\cH)}}

\newcommand{\mC}{{\mathrm{C}}}

\newcommand{\tp}{{\widetilde{p}}}
\newcommand{\tv}{{\widetilde{v}}}
\newcommand{\tI}{{\widetilde{I}}}
\newcommand{\trho}{{\widetilde{\rho}}}
\newcommand{\tH}{{\widetilde{\cH}}}
\newcommand{\tSH}{{\cS(\tH)}}
\newcommand{\tg}{{\widetilde{\gamma}}}

\newcommand{\oG}{{\overline{G(\pi)}}}

\newcommand{\CG}{\overline{G(d)}}
\newcommand{\ov}{{\overline{v}}}
\newcommand{\op}{{\overline{p}}}

\newcommand{\cH}{\mc{H}}

\newcommand{\cN}{\mc{N}}
\newcommand{\cNsa}{\mc{N}_\sa}
\newcommand{\LH}{\mc{L}(\cH)}
\newcommand{\LHsa}{\LH_\sa}
\newcommand{\BH}{\mc{B}(\cH)}

\newcommand{\UH}{{\mc{U}(\cH)}}
\newcommand{\SH}{\mc{S}(\cH)}

\newcommand{\PH}{{\mc{P}(\cH)}}
\newcommand{\PHone}{{\mc{P}_1(\cH)}}

\newcommand{\PO}{{\mc{P}(\cO)}}

\newcommand{\POone}{{\mc{P}_1(\cO)}}
\newcommand{\POm}{{\mc{P}_\mathrm{min}(\cO)}}

\newcommand{\PC}{{\mc{P}(C)}}

\newcommand{\VO}{{\mc{V}(\cO)}}

\newcommand{\cL}{\mc{L}}

\newcommand\tU{{\tilde{U}}}
\newcommand\tpi{{\widetilde{\pi}}}

\newcommand\tV{{\widetilde{V}}}
\newcommand\tE{{\widetilde{E}}}

\newcommand\Stab{{\mathrm{Stab}}}
\newcommand\QStab{{\mathrm{QStab}}}
\newcommand\Th{{\mathrm{Th}}}

\newcommand\id{{\mathrm{id}}}

\newcommand\ra{\rightarrow}

\newcommand\Llra{\Longleftrightarrow}

\newcommand\hra{\hookrightarrow}

\begin{document}
\title{Coming full circle}
\subtitle{\Large{A unified framework for Kochen-Specker contextuality}}
\author{Markus Frembs\\\small{Institut f\"ur Theoretische Physik, Leibniz Universit\"at Hannover,}\\[-0.2cm]\small{Appelstraße 2, 30167 Hannover, Germany}\\\small{Okinawa Institute of Science and Technology Graduate University,}\\[-0.2cm]\small{Onna, Okinawa 904 0495, Japan}\\
\small{Centre for Quantum Dynamics, Griffith University,}\\[-0.2cm]\small{Yugambeh Country, Gold Coast, QLD 4222, Australia}}
\date{}
\maketitle

\vspace{-1.2cm}

\begin{abstract}
    Contextuality is a key distinguishing feature between classical and quantum physics. It expresses a fundamental obstruction to describing quantum theory using classical concepts. In turn, when understood as a resource for quantum computation, it is expected to hold the key to quantum advantage. Yet, despite its long recognised importance in quantum foundations and, more recently, in quantum computation, the mathematics of contextuality has remained somewhat elusive - different frameworks address different aspects of the phenomenon, yet their precise relationship often is unclear. In fact, there is a glaring discrepancy already between the original notion of contextuality introduced by Kochen and Specker on the one side [J. Math. Mech., 17, 59, (1967)], and the modern approach of studying contextual correlations on the other [Rev. Mod. Phys., 94, 045007 (2022)].

    In a companion paper [arXiv:2408.16764], we introduce the conceptually new tool called ``context connections'', which allows to cast and analyse Kochen-Specker (KS) contextuality in new form. Here, we generalise this notion, and based on it prove a complete characterisation of KS contextuality for finite-dimensional systems. To this end, we develop the framework of ``observable algebras". We show in detail how this framework subsumes the marginal and graph-theoretic approaches to contextuality, and thus that it offers a unified perspective on KS contextuality. In particular, we establish the precise relationships between the various notions of ``contextuality" used in the respective settings, and in doing so, generalise a number of results on the characterisation of the respective notions in the literature.
    
    Given the central role contextuality plays in quantum foundations and quantum computation, the unified framework presented here will be widely applicable. In particular, we expect the new tools provided by it to play a key role in quantifying contextuality as a resource, which remains an important open problem towards the efficient use of quantum computers.
\end{abstract}

\newpage

\tableofcontents

\section{Introduction}\label{sec: motivation}

Planck's derivation of the spectrum of black body radiation, Einstein's explanation of the photo-electric effect and Bohr's phenomenological model for the hydrogen atom all employ the idea of quantised units of energy. The respective phenomena seemed otherwise hard to reconcile with classical physics, and soon ushered in the development of quantum mechanics, whose mathematical form was set in stone by von Neumann \cite{vonNeumann1932}. Yet, while those and many more theoretical breakthroughs - paired with overwhelming experimental confirmation of quantum-mechanical predictions - provide very strong evidence that nature is quantum at its core, they do not rule out the existence of a similarly ingenious description within the realms of classical physics: indeed, a description in terms of classical ``hidden variables" is always possible in principle \cite{Bohm1952,BeltramettiBugajski1995}.

Still, a decisive blow to the hidden-variable programme \cite{EPR1935} was dealt by Bell's theorem \cite{Bell1964}, and the much more recent loophole-free verification of stronger than classical correlations in Bell experiments \cite{CHSH1969,AspectEtAl1981,ZeilingerEtAl2015,ShalmEtAl2015}, that is, of correlations described in terms of ``classical states" - in the form of probability measures - on a ``classical state space" - in the form of a measurable space, under the additional assumption of relativistic causality \cite{WisemanCavalcanti2017}. The study of Bell inequality violations \cite{BrunnerEtAl2014} and the theory of quantum entanglement \cite{Horodeckisz2009} form the cornerstone of quantum information science, and lie behind the promise of a (second) technological revolution based on quantum effects. Indeed, entanglement is at the heart of quantum cryptography \cite{Ekert1991,BennettEtAl1993,Renner2008,BennettBrassard2014}, quantum information processing \cite{Dieks1982,WoottersZurek1982,BennettWiesner1992,BennettEtAl1993,Barrett2007} and has been shown to be a necessary resource for quantum computation \cite{Deutsch1985,JozaaLinden2003,RaussendorfBriegel2001,BriegelRaussendorf2001,RaussendorfBrowneBriegel2003}.

An even stronger candidate for the resource underlying quantum computers is \emph{contextuality} \cite{Raussendorf2013,HowardEtAl2014,FrembsRobertsBartlett2018,FrembsRobertsCampbellBartlett2023}, a concept closely related to Bell nonlocality. Indeed, a common slogan has that contextuality is a generalisation of Bell nonlocality. This statement refers to contextuality as understood in the ``marginal approach" \cite{BudroniEtAl2022}, which - similar to and motivated by Bell inequalities - studies noncontextuality inequalities \cite{KlyachkoEtAl2008,Cabello2008b,BadziagBengtssonCabelloPitowsky2009,YuOh2012,KleinmannEtAl2012,AraujoEtAl2013,CabelloSeveriniWinter2014,XuCabello2019}, and seeks to show their experimental violation by quantum theory \cite{BartosikEtAl2009,KirchmairEtAl2009,AmselemEtAl2009,GuehneEtAl2010,MoussaEtAl2010,LapkiewiczEtAl2011,Winter2014}. Yet, this approach is only partially related with the Kochen-Specker theorem \cite{YuOh2012,YuTong2014,YuGuoTong2015,Frembs2024a}; to make matters worse, there exists by now a patchwork of different approaches, each studying its own notion of ``contextuality" \cite{IshamButterfieldI,IshamButterfieldII,IshamButterfieldIII,IshamButterfieldIV,Spekkens2005,IshamDoeringI,IshamDoeringII,IshamDoeringIII,IshamDoeringIV,HLS2009,HLS2010,AbramskyBrandenburger2011,AbramskyEtAl2015,AbramskyEtAl2017,ChavesFritz2012,FritzChaves2013,CabelloSeveriniWinter2014,AcinFritzLeverrierSainz2015,BartlettRaussendorf2016,DzhafarovKujalaCervantes2020,OkayTyhurstRaussendorf2018,OkayRoberts2017}, whose precise relationship often remains unclear. This not only makes cross-talk difficult, but also hinders progress towards an in-depth, quantitative analysis of contextuality as a resource in quantum computation. In this work, we (begin to) set this record straight. 

To this end, it is instructive to compare the original notion of Kochen-Specker (KS) noncontextuality with Bell locality. Both concepts represent minimal, operationally motivated assumptions, under which a separation between quantum and classical physics can be obtained. Yet, while Bell's theorem shows that certain (entangled) quantum \emph{states} are incompatible with a classical state space description, the Kochen-Specker theorem \cite{KochenSpecker1967} is concerned with the problem whether quantum \emph{observables} admit a classical representation (regardless of whether composite or not). Both results demonstrate that quantum theory does not admit an interpretation in terms of classical (local or noncontextual) hidden variables, yet the respective conclusions differ: Bell's theorem shows that the state space of quantum theory is not (fully) classical, while the Kochen-Specker theorem asserts that there can be no \emph{classical embedding} $\epsilon:\cNsa\ra L_\infty(\Lambda)$ of quantum observables into the space of measurable functions on some measurable space $\Lambda$ (the ``state" or ``phase space" of a physical theory), such that functional relations between compatible observables are preserved (see Eq.~(\ref{eq: KSNC})). As we will see in detail in Sec.~\ref{sec: algebraic vs marginal approach}, for quantum (sub)systems the latter is a strictly stronger result.

Proving the Kochen-Specker theorem therefore is to show that no such embedding exists. The original proof of the Kochen-Specker theorem considers triples of squared spin-$1$ observables $S^2_x,S^2_y,S^2_z$ in three orthogonal directions in $\R^3$, which commute $[S^2_x,S^2_y] =[S^2_y,S^2_z] = [S^2_z,S^2_x] = 0$, share the same spectrum $\spec(S^2_x) = \spec(S^2_y) = \spec(S^2_z) = \{0,1\}$ and satisfy the algebraic constraint $S^2_x + S^2_y + S^2_z = 2\mathbbm{1}$, where $\mathbbm{1}$ denotes the identity in the space of linear operators $\cL(\C^3)$. Under the assumption that a classical embedding for spin-$1$ observables exists, these properties imply that for any three orthogonal directions in space, one spin component must attain spin value 0, while the other two attain spin value 1. Such an assignment is also known as a \emph{Kochen-Specker (KS)-colouring} (see Sec.~\ref{sec: KS-colouring vs colouring}). Ref.~\cite{KochenSpecker1967} constructs a finite set of spin-$1$ observables (defined via their spectral projections, represented by 117 rays in $\R^3$) that does not admit a KS-colouring. Subsequent proofs have reduced this number: the current record in three dimensions stands at 31 rays (and corresponding spin-$1$ observables) \cite{ConwayKochen2006}, and for more than three dimensions at 18 rays \cite{CabelloEstebanzGarcia-Alcaine1997}; the latter is known to be optimal, the former is not \cite{Zhen-PengChenGuehne2020}.

The choice of spin-$1$ observables allows to reduce the problem of finding (an obstruction to the existence of) a classical embedding to a KS-colouring. However, the existence of a KS-colouring is only \emph{necessary but not sufficient} for the existence of a classical embedding \cite{Frembs2024a}. This discrepancy between KS-colourings and proofs of the Kochen-Specker theorem has been thrown into sharp relief by the search for experimental tests of contextuality within the marginal approach. More precisely, it was found that there exist noncontextuality inequalities which are violated by quantum theory, yet which do admit a KS-colouring \cite{YuOh2012}. This observation has been met with some surprise, which we take to express the fact that the precise relationship between the constraints on classical embeddings and the restriction to KS-colourings, as well as other similar notions of ``contextuality" in the literature remain poorly understood \cite{BudroniEtAl2022}. In this work, we provide a careful analysis of and resolve various of these differences.\\

\textbf{Contributions.} In Ref.~\cite{Frembs2024a}, we obtain a reformulation of Kochen-Specker contextuality (Eq.~(\ref{eq: KSNC})), based on the notion of ``context connections'' (see Def.~\ref{def: context connection}), and constraints encoded by these along so-called ``context cycles" (see Def.~\ref{def: context cycle}) on the partial order of contexts corresponding to an observable algebra (see Sec.~\ref{sec: OAs}). There, we focus on its conceptual motivation and, for comparison with Ref.~\cite{KochenSpecker1967}, discuss the case of spin-$1$ systems only. Here, we generalise the formalism introduced in Ref.~\cite{Frembs2024a} and show that it applies universally, in particular, to any (subsystem of a) finite-dimensional quantum system. In short, our first main contribution in Sec.~\ref{sec: algebraic KSNC}, is to prove a
\begin{itemize}
    \item[\textbf{(i)}] \textbf{new, complete characterisation of Kochen-Specker (KS) contextuality.}
\end{itemize}
To this end, we introduce the framework of ``observable algebras", which is motivated both by the partial algebras in Kochen and Specker's original work \cite{KochenSpecker1967,Kochen2015,AbramskyBarbosa2020}, as well as the more recent topos-theoretic formalisation in Ref.~\cite{IshamButterfieldI,IshamDoeringI,IshamDoeringII,IshamDoeringIII,IshamDoeringIV,HLS2009,HLS2010,DoeringFrembs2019a}. This generality enables us to compare KS contextuality with various notions of contextuality in the marginal approach (Sec.~\ref{sec: algebraic vs marginal approach}), and the closely related (hyper)graph-theoretic approaches \cite{ChavesFritz2012,FritzChaves2013,CabelloSeveriniWinter2014,AcinFritzLeverrierSainz2015,AmaralCunha2018} (Sec.~\ref{sec: algebraic vs graph-theoretic approach}); our second key contribution then consists in establishing
\begin{itemize}
    \item[\textbf{(ii)}] \textbf{observable algebras as a unified framework for KS contextuality.}
\end{itemize}

In particular, we show in detail how existing frameworks embed within that of observable algebras and we precisely differentiate the respective notions of contextuality used in those. A summary of this comparison is given in Tab.~\ref{tab: summary}. This part both complements and substantially extends the recent review on the subject \cite{BudroniEtAl2022}.

Finally, by combining our characterisation of KS contextuality with the comparison between different notions of contextuality, in particular, with the notion of state-independent contextuality (SI-C) in the graph-theoretic approach, we identify a
\begin{itemize}
    \item[\textbf{(iii)}] \textbf{necessary criterion for an orthogonality graph to be state-independently contextual (see Def.~\ref{def: SI-C set} and Def.~\ref{def: SI-C graph}), in terms of its chromatic number.}
\end{itemize}
This complements existing partial results in Ref.~\cite{Cabello2012,CabelloKleinmannBudroni2015} and resolves an open conjecture in Ref.~\cite{RamanathanHorodecki2014}. The main results of this paper are summarised in the following list.

\vspace{-0.1cm}

\begin{tikzpicture}[remember picture, overlay]
    \node (rightenum) at (.9\textwidth,0) {};
    \draw [decorate, decoration={brace}, thick] ($({pic cs:top1} -| rightenum) + (0, 1em)$) -- ({pic cs:bot1} -| rightenum) node [midway, right=0.2cm] {\textbf{Sec.}~\ref{sec: algebraic KSNC}};
    \draw [decorate, decoration={brace},thick] ($({pic cs:top2} -| rightenum) + (0, 1em)$) -- ({pic cs:bot2} -| rightenum) node [midway, right=0.2cm] {\textbf{Sec.}~\ref{sec: algebraic vs marginal approach}};
    \draw [decorate, decoration={brace},thick] ($({pic cs:top3} -| rightenum) + (0, 1em)$) -- ({pic cs:bot3} -| rightenum) node [midway, right=0.2cm] {\textbf{Sec.}~\ref{sec: algebraic vs graph-theoretic approach}};
\end{tikzpicture}
\begin{itemize}[leftmargin=0.75cm]
    \item \textbf{Thm.~\ref{thm: CNC for OAs} - new, complete characterisation of KS contextuality}\tikzmark{top1}\\[0.1cm]
    $\ra$ generalisation of Thm.~2 in Ref.~\cite{Frembs2024a}
    \item \textbf{Thm.~\ref{thm: KSNC = n-colourability}, Thm.~\ref{thm: n-colourability for non-maximal OAs} - KS contextuality as a colouring problem}\\[0.1cm]
    $\ra$ distinction from KS-colourings of KS sets\tikzmark{bot1}
    \vspace{.3cm}
    \item \textbf{Thm.~\ref{thm: KSNC = marginal KSNC} - KS contextuality in marginal approach}\tikzmark{top2}\\[0.1cm]
    $\ra$ KS noncontextuality in Meyer-Clifton-Kent models \cite{Meyer1999,Kent1999,CliftonKent2000}
    \item \textbf{Thm.~\ref{thm: KSNC of truncated OAs}, Cor.~\ref{cor: acyclic = KS} - acyclicity implies KS noncontextuality}\\[0.1cm]
    $\ra$ comparison with Vorob'ev theorem \cite{Vorob1962} (Thm.~\ref{thm: Vorob'ev theorem})
    \item \textbf{Thm.~\ref{thm: single cycles in d=3 are KS noncontextual} - KS contextuality requires multiple context cycles}\\[0.1cm]
    $\ra$ distinction from n-cycle contextuality scenario (see Def.~\ref{def: n-cycle scenario}) \cite{AraujoEtAl2013,XuCabello2019}\tikzmark{bot2}
    \vspace{.3cm}
    \item \textbf{Thm.~\ref{thm: KSNC = NCOG}, Thm.~\ref{thm: RH conjecture} - KS contextuality in graph-theoretic approach}\tikzmark{top3}\\[0.1cm]
    $\ra$ proof of a restricted version of the conjecture in Ref.~\cite{RamanathanHorodecki2014}
    \item \textbf{Thm.~\ref{thm: n-colourable implies USI-C}, Thm.~\ref{thm: partial RH conjecture} - characterisation of SI-C graphs and SI-C sets}\\[0.1cm]
    $\ra$ distinction between related notions of ``contextuality"
    \tikzmark{bot3}
\end{itemize}
\vspace{0.4cm}

\begin{table}[htb]
    \centering
    \begin{tabular}{c|ccc}
        \toprule[0.05cm]
        notion of & partial order & states/correlations & orthogonality \\
        ``contextuality'' & of contexts $\CO$ & $\Gamma[\CO]/\QStab(G)$ & graph $G=G(\cO)$ \\
        \midrule\\[-0.3cm]
        $\cO$ is KS & $\nexists$ flat context & $\Gamma_\mathrm{cl}[\CO]$ not & \multirow{2}*{\begin{minipage}{3cm}\centering\vspace{0.35cm} $\chi(G^*)>\dim(\One)$\\ (Thm.~\ref{thm: KSNC = n-colourability},\ref{thm: n-colourability for non-maximal OAs})\end{minipage}} \\
        contextual & connection on $\cC(\cO^*)$ & separating & \\
        (Def.~\ref{def: KSNC}) & (Thm.~\ref{thm: CNC for OAs}) & (Thm.~\ref{thm: KSNC = marginal KSNC}) & \\[0.4cm]
        $\cO$ violates & \multirow{2}*{\begin{minipage}{3cm}\centering\vspace{0.35cm}
        $\CO$ not acyclic\\(Thm.~\ref{thm: Vorob'ev theorem})\end{minipage}} & \multirow{2}*{\begin{minipage}{3.5cm}\centering\vspace{0.35cm}
        $\Gamma_\mathrm{cl}[\CO]\subsetneq\Gamma[\CO]$\\ (Def.~\ref{def: fully classically correlated})\end{minipage}} & \\
        NC inequality & & & \\
        (Def.~\ref{def: fully classically correlated}) & & & \\[0.3cm]
        $\pi(G)$\footnotemark\ is a & $\nexists$ flat context & $\nexists\rho\in\SH$ s.t. & \multirow{2}*{\begin{minipage}{3cm}\centering\vspace{0.35cm}
        $\chi(G^*)>\dim(\cH)$\\ (Cor.~\ref{cor: unital SI-C set})\end{minipage}} \\
        SI-C set & connection on $\cC(\cO^*)$ & $\rho|_{\pi(G)}\in\Stab(G)$ & \\
        (Def.~\ref{def: SI-C graph}) & (Thm.~\ref{thm: CNC for OAs}) & (Def.~\ref{def: SI-C graph}) & \\[0.1cm]
        \bottomrule[0.05cm]
    \end{tabular}
    \caption{Comparison of different notions of ``contextuality", as defined in the algebraic, marginal and graph-theoretic approaches, together with their respective (partial) characterisations in terms of the partial order of contexts $\CO$ of an observable algebra $\cO$, (classical) states $(\Gamma_\mathrm{cl}[\CO]\subset)\Gamma[\CO]$ on $\cO$ (see Def.~\ref{def: state}), respectively (classical) correlations $(\Stab(G)\subset)\QStab(G)$ on $G$ (see Eq.~\ref{eq: cl-qu-nd correlations}), as well as the chromatic number $\chi(G^*)$ of the orthogonality graph $G^*=G(\cO^*)$, where $\cO^*$ is a maximal extension of $\cO$ (see Def.~\ref{def: max extension}). Here, $\pi:G\ra\PH$ is a unital realisation (see Def.~\ref{def: unital realisation}), that is, $G=G(\cO)$ and $\cO=\cO(\pi(V))\subset\LHsa$.}
    \label{tab: summary}
\end{table}
\footnotetext{For the present comparison, we restrict to orthogonality graphs $G=\GO$ (with unital realisation, see Def.~\ref{def: unital realisation}), since all notions of contextuality in Tab.~\ref{tab: comparison} readily apply in this case. In Sec.~\ref{sec: SI-C sets and graphs}, we further generalise to orthogonality graphs with freely completable realisations (see Tab.~\ref{tab: KSC vs SI-C}).\label{fn: unital only}}

Finally, our goal is to apply the tools introduced here (and in Ref.~\cite{Frembs2024a}) towards an in-depth analysis of KS contextuality, e.g. as a resource in quantum computation. We discuss some avenues for future research along those lines in the outlook (Sec.~\ref{sec: conclusion}).

\section{Algebraic characterisation of Kochen-Specker contextuality}\label{sec: algebraic KSNC}

In Ref.~\cite{Frembs2024a} we reformulated the notion of Kochen-Specker (KS) contextuality \cite{KochenSpecker1967} in new algebraic form, for the special case of spin-$1$ quantum systems. This reformulation revealed stronger constraints than those arising from the existence of valuations, which reconciles the notion of KS contextuality with proofs of the KS theorem in the form of violations of noncontextuality inequalities \cite{Pitowsky1991,Larsson2002,KlyachkoEtAl2008,BadziagBengtssonCabelloPitowsky2009,YuOh2012} in those cases. Here, we lift this reformulation to a complete characterisation of KS contextuality, and generalise it to (and beyond) the setting of (finite-dimensional) partial algebras \cite{KochenSpecker1967}. Since Kochen-Specker contextuality was originally defined in this setting, this shows that the tools introduced in Ref.~\cite{Frembs2024a} capture the full extent of KS contextuality. Moreover, it will allow us to carefully compare and contrast it with various related notions of contextuality in the literature, specifically with the ``marginal approach" \cite{AbramskyBrandenburger2011,AcinFritzLeverrierSainz2015,BudroniEtAl2022} in Sec.~\ref{sec: algebraic vs marginal approach} and with the graph-theoretic approach \cite{CabelloSeveriniWinter2014} in Sec.~\ref{sec: algebraic vs graph-theoretic approach}.

\subsection{Observables, events, contexts}\label{sec: OAs}

Various notions of ``contextuality" exist, making it rather difficult to differentiate between them. One contribution of this work is a thorough comparison between the most commonly used notions, applied to ideal measurements (see Sec.~\ref{sec: algebraic vs marginal approach} below). In order to do so, we first establish the relevant notation and terminology we will be using.\\

\textbf{(Abstract) Observable algebras.} We will introduce the relevant structures by comparison with the familiar setting of quantum theory. There, observables are represented by self-adjoint (Hermitian) operators $O\in\LHsa$ acting on a finite-dimensional Hilbert space $\cH$.\footnote{More generally, an observable in quantum mechanics is a self-adjoint operator in a von Neumann algebra $\cN\subset\BH$, that is, a weakly closed subalgebra of the algebra of bounded operators on a (separable) Hilbert space $\cH$. Presently, we will only consider finite-dimensional (quantum) systems.} The commutator equips $\LHsa$ with a reflexive, symmetric, yet generally not transitive relation. Physically, the commutator encodes a notion of \emph{joint measurability} or \emph{compatibility} between quantum observables. Indeed, Heisenberg's uncertainty relation prohibits an experimenter from measuring non-commuting observables jointly with arbitrary precision. As a consequence, the physical interpretation of algebraic operations between noncommuting operators remains somewhat opaque. However, these operations have a well-defined meaning when restricted to commuting operators.\footnote{Indeed, algebraic operations between jointly measurable observables may be considered part of an experimenter's ability to analyse her experimental data. In this view, they become consistency conditions for formal representations of the observables of a system, as we will emphasise below.} This suggests to restrict algebraic relations to compatible observables.

\begin{definition}\label{def: abstract observable algebra}
    An \emph{abstract observable algebra $\cO$ over a field $K$} is a set with a binary relation $\odot \subset \cO\times\cO$ (``joint measurability'' or ``compatibility''), an identity element $\One\in\cO$ such that $\One\odot O$ for all $O\in\cO$, and such that addition, multiplication and scalar multiplication are defined and close on the domain of $\odot$, that is, for all $O_1,O_2,O_3 \in \cO$ with $O_1\odot O_2$, $O_2\odot O_3$ and $O_1\odot O_3$, one also has $(O_1+O_2)\odot O_3$, $(O_1O_2)\odot O_3$ and $\lambda O_1\odot O_3$ for all $\lambda \in K$.\footnote{Throughout, we will take $K = \R$, and thus usually omit specifying it explicitly.}
\end{definition}

$\LHsa$ with the anti-commutator defines an \emph{(abstract) observable algebra}, where the algebraic operations between noncommuting observables are forgotten. A key motivation to generalise from this example is to consider observable subalgebras of quantum theory. For instance, restrictions to subalgebras appear naturally in the context of quantum computation, and experimental setups more generally. Quantum observable algebras satisfy an additional property known as \emph{Specker's principle} \cite{Specker1960,LiangSpekkensWiseman2011,Henson2012,FritzEtAl2013,Cabello2012}: any triple of pairwise compatible elements $O_1,O_2,O_3 \in \cO$ with $O_1\odot O_2$, $O_2\odot O_3$ and $O_1\odot O_3$ is fully compatible, that is, $O_1\odot O_2\odot O_3$. With quantum theory in mind, Kochen and Specker define partial algebras as observable algebras that satisfy Specker's principle \cite{KochenSpecker1967}. Yet, our formalism applies without this restriction.\footnote{Another motivation to study general observable algebras is the search for possible extensions to quantum theory, as well as to identify (physically motivated) principles that single out quantum theory amongst such extensions. For instance, Specker's principle has been suggested as such a principle \cite{Specker1960,Henson2012,FritzEtAl2013,Cabello2012}. The tools developed in this work and in Ref.~\cite{Frembs2024a} might therefore also prove useful towards the reconstruction paradigm in quantum foundations \cite{Hardy2001,Hardy2001a,Grinbaum2007,Jaeger2019,Mueller2021,MuellerGarner2023}.}

Note that Def.~\ref{def: abstract observable algebra} is abstract in the sense that it does not specify the elements of $\cO$ explicitly. To do so, we will represent observables $O\in\cO$ in an abstract observable algebra $\cO$ as random variables (measurable functions) $O:\Sigma_O\ra\R$ on some measurable spaces $\Sigma_O$, such that whenever $O\odot O'$ there exists a measurable space $\Sigma_{O,O'}\supset\Sigma_O,\Sigma_{O'}$ (hence, $O:\Sigma_{O,O'}|_{\Sigma_O}\ra\R$ and $O':\Sigma_{O,O'}|_{\Sigma_{O'
}}\ra\R$). In this view, the algebraic structure between compatible observables is that of an algebra of measurable functions.

\begin{definition}\label{def: OA}
    An \emph{observable algebra $\cO$} is an abstract observable algebra whose elements are given as random variables $O:\Sigma_O\ra\R$ with $\Sigma_O$ a measurable space for every $O\in\cO$.
\end{definition}

Every observable algebra thus defines an abstract observable algebra; however, different observable algberas may correspond to the same abstract observable algebra. Below, we will restrict to observables with finitely many outcomes, thus leaving topological issues for future work. As a consequence, the measurable spaces $\Sigma_C\supset\{\Sigma_O\}_{O\in C}$ corresponding to any subset of mutually compatible observables $C\subset\cO$, that is, with $O\odot O'$ for all $O,O'\in C$, will always have finite cardinality.\footnote{In the general case, one may want to restrict to measurable spaces such that the space of measurable functions on it has nice algebraic properties. A natural choice is to consider Hyperstonean spaces, which are categorically equivalent to (the dual of) commutative von Neumann algebras \cite{Pavlov2021}. An alternative choice is to consider (locally) compact Hausdorff spaces and continuous functions (vanishing at infinity), which are dually equivalent to unital (non-unital) commutative $C^*$-algebras.}\\

\textbf{Contexts and their order.} Next, we take an alternative perspective: instead of considering an observable algebra $\cO$ as a set of individual observables, we collect observables into \emph{(measurement) contexts}, that is, subsets $C\subset \cO$ of jointly measurable ones. What type of subsets? Recall that every self-adjoint operator $O\in\LHsa$ has a spectral decomposition $O = \sum_{i=1}^{\dim(\cH)} \lambda_ip_i$ for $\lambda_i\in\R$, $p_i\in\PH$ with $p_ip_j=\delta_{ij}$ and $\PH$ the (lattice of) projections on $\cH$. We denote by $\spec(O)=\{\lambda_i\}_i$ the \emph{spectrum}, that is, the set of eigenvalues of $O$. Given any function $h:\R\ra\R$ we obtain another self-adjoint operator by $h(O) = \sum_{i=1}^{\dim(\cH)} h(\lambda_i)p_i$.\footnote{For observables with continuous spectra, $h$ needs to be a measurable function.} Functional relations of this form may be considered part of the side-processing of our experimental setup, and should thus be preserved in any formalisation of $\cO$.\footnote{For instance, an experimenter is free to relabel the outcomes of observables in her experiment.} This motivates to include all such relations in our description, that is, we represent a \emph{(measurement) context} in quantum theory as a commutative subalgebra $C\subset\LHsa$, and collect all contexts into the partial order $\CH = (\{C \subset \LHsa \mid C \mathrm{\ commutative}\},\subset)$.\footnote{For infinite-dimensional quantum systems, a natural topology on observable algebras is the weak operator topology, such that contexts are represented by commutative von Neumann subalgebras of $\cO=\cNsa$ of some von Neumann algebra $\cN$. Alternatively, one may consider compact Hausdorff spaces which are dually equivalent to the category of unital commutative $C^*$-algebras \cite{HLS2009,HLS2010,VanDenBergHeunen2012,Heunen2014,HardingHeunen2021}. Here, we restrict to finite-dimensional algebras, for which all relevant operator topologies coincide. We find it useful to henceforth explicitly distinguish this in our notation by denoting finite-dimensional commutative algebras (contexts) by $C$ instead of $V$, and the corresponding partial order of contexts by $\CO$ instead of $\VO$, as was used in related previous work \cite{IshamDoeringI,IshamDoeringII,IshamDoeringIII,IshamDoeringIV,DoeringFrembs2019a,FrembsDoering2022a,Frembs2022a,FrembsDoering2023}.}

$\cO$ is generally not an algebra itself. Yet, since addition, multiplication and scalar multiplication close on the domain of $\odot$ by Def.~\ref{def: abstract observable algebra}, we may still collect jointly measurable observables into contexts in the form of commutative algebras (over $K=\R$), generated by collections of jointly measurable observables.\footnote{Under the restriction to random variables with finitely many outcomes, the respective commutative algebra arises with respect to the closure under addition, multiplication and scalar multiplication of its generating elements. More generally, one further has to impose topological constraints.} The definition of the partial order of contexts $\CH$ in Ref.~\cite{DoeringIsham2011,DoeringFrembs2019a} thus admits the following generalisation.

\begin{definition}\label{def: context category}
    Let $\cO$ be an observable algebra over the field $K$. The set of commutative algebras $C\subset\cO$ over $K$, ordered by inclusion, is called the \emph{partial order of contexts} (or \emph{context category}\footnote{A partial order is a category with objects the elements in the partial order and morphisms its relations.}) $\CO$ of $\cO$.

    We call a context $C\in\CO$ \emph{maximal} if $C \subset C'$ for $C' \in \CO$ implies $C=C'$, and write $\COm$ for the set of maximal contexts in $\CO$.
\end{definition}

We obtain a ``spectral decomposition" for any observable $\cO\ni O:\Sigma_C\ra\R$ of the form $O=\sum_\sigma\sigma p_\sigma$ with $\sigma\in\R$ and where $p_\sigma:=\chi_{\Sigma_\sigma}:\Sigma_C\ra\{0,1\}$ is the indicator function of the subset $\Sigma_\sigma\subset\Sigma_C$ defined by $\chi_{\Sigma_\sigma}(\sigma)=1$ if and only if $\sigma\in\Sigma_\sigma\subset\Sigma_C$. We will write $\PC$ for the (representation of the) Boolean algebra $\BC=2^{\Sigma_C}$ in terms of indicator functions on $\Sigma_C$, and write the algebraic relations as addition and multiplication (instead of set union and intersection), in analogy with the quantum case. Stressing this analogy, we will similarly refer to the set $\spec(O)=\{\sigma_i\}_i$ in the spectral decomposition $O=\sum_i\sigma_i p_i$ of an observable $O\in\cO$ as its \emph{spectrum} and to the elements $p_i\in\PC$ as \emph{projections} or \emph{events}.

Since every commutative algebra $C\subset\cO$ contains the identity element $\One$, we have $C_\One:=\{\lambda\One\mid\lambda\in K\} \subset C$ for all $C\in\CO$. Hence, $C_\One$ is the least element of $\CO$. Conversely, any partial order of finite-dimensional commutative algebras with least element $C_\One$ defines an observable algebra $\cO = \bigcup_{C\in\CO} C$, with identity given by (identifying) the maximal indicator functions $\One:=\chi_{\Lambda_C}$ in every maximal context $C\in\COm$, and by setting $O\odot O'$ if and only if there exists $C \in \CO$ such that $O,O' \in C$ \cite{CannonDoering2018}.\\

\textbf{Event algebras.} Next, we consider the restriction of an observable algebra $\cO$ to its elementary dichotomic (``yes-no") observables, equivalently to its \emph{events} or \emph{projections}. These correspond to the elements $p^O_i$ in the spectral decomposition $O=\sum_{i=1}^{\spec(O)}\sigma p^O_\sigma$ for $\sigma\in\R$ of observables $O\in\cO$. Note that $\{p^O_\sigma\}_\sigma$ generate the Boolean subalgebra $\PC\cong 2^{\Sigma_C}$, given by the restriction of the minimal commutative subalgebra $C\subset\cO$ containing $O$ to its projections.\footnote{As noted above, we identify the elements $p\in\PC$ with indicator functions $\chi:\Sigma_C\ra\R$ and thus write the algebraic operations in $\PC$ additively and multiplicatively.}

It is a defining property of a classical system that all observables are jointly measurable, hence, are contained in a single measurement context. In other words, every classical observable can be represented as a measurable function $O\in L_\infty(\Lambda)$ with respect to a single measurable space $(\Lambda,\sigma)$ (the ``state space'' of the theory). In this case, the events are defined by the $\sigma$-algebra of measurable subsets of $\Lambda$ which defines a Boolean algebra $\cB(\Lambda)$ under union, intersection and complement.\footnote{One further obtains a complete Boolean algebra after quotienting out the $\sigma$-ideal of null sets.}

On the other hand, the spectral decomposition of observables $O,O'\in\cO$ in a general observable algebra $\cO$ may differ, that is, may involve different events. Still, whenever two observables are jointly measurable $O,O'\in C$, they have a common spectral decomposition, hence, describe events in a common Boolean algebra (generated by the elements in) $\PC$. Instead of a single Boolean algebra, the events of an observable algebra are thus described by a collection of Boolean algebras \cite{KochenSpecker1967,HardingNavara2011,Kochen2015,DoeringHarding2016,CannonDoering2018,HardingEtAl2019,AbramskyBarbosa2020}.

\begin{definition}\label{def: event space}
    Let $\cO$ be an observable algebra. The \emph{event algebra $\PO$ of $\cO$} is the collection of Boolean algebras $\PC$ for all contexts $C\in\CO$, that is, $\PO=\bigcup_{C\in\CO}\PC$.
\end{definition}

Our notation $\PO$ is supposed to indicate that its elements are represented as indicator functions on the event spaces $\Sigma_C$ corresponding to sets of jointly measurable observables. In analogy with Def.~\ref{def: abstract observable algebra}, we may also define an abstract event algebra $\BO$. An (abstract) event algebra is a partial Boolean algebra if it satisfies Specker's principle. In particular, every orthomodular lattice gives rise to a partial Boolean algebra that is moreover transitive \cite{Gudder1972}. For instance, the event algebra $\cP(\LHsa)$ arises from the lattice of projections $\PH$ on $\cH$ by restricting operations to commuting elements only.

Finally, we remark that the information coded into an event algebra can also be cast as a hypergraph, as defined in Ref.~\cite{FritzChaves2013,AcinFritzLeverrierSainz2015}, with vertices the atoms and hyperedges given non-minimal projections in $\BO$ (equivalently, by Boolean subalgebras $\BC\in\BO$). Conversely, a hypergraph defines an abstract event algebra $\BO$ once we identify the identity element as $\One=\sum_{v\in e} v$ for every maximal hyperedge $e\in E$. From this perspective (and by Lm.~\ref{lm: embeddings vs orthomorphisms vs order-preserving maps} below), both formalisms are thus concerned with the same underlying structure. However, whereas the latter studies contextuality in terms of correlations (see also Sec.~\ref{sec: algebraic vs marginal approach}), our approach is entirely algebraic and thus complementary to it.\footnote{In addition, we point out that the hypergraph approach \cite{ChavesFritz2012,AcinFritzLeverrierSainz2015} is specifically concerned with the study of correlations on composite systems. In this regard, the perspective of contexts seems helpful in that specifying an identity element (hence, a minimal context) for local subsystems automatically ensures that product correlations are non-signalling. Indeed, the Cartesian product of context categories allows for a natural description of Bell's theorem \cite{DoeringFrembs2019a}; moreover, Ref.~\cite{FrembsDoering2022a,Frembs2022a} characterise both quantum and separable quantum states from general non-signalling correlations, under the additional assumption that observables (hence, their partial order of contexts) are locally quantum \cite{KlayRandallFoulis1987,Wallach2002,BarnumEtAl2010}.\label{fn: identity element in hypergraphs}}\\

\textbf{Finite-dimensional and maximal observable algebras.}

\begin{definition}\label{def: dim}
    Let $\cO$ be an observable algebra with event algebra $\PO$. An additive function $\dim:\PO\ra\mathbb{N}$ with $\dim(p)+\dim(p')=\dim(p+p')$ whenever $p,p'\in\PC$ with $pp'=0$ for some $C\in\CO$ is called a \emph{dimension function on $\cO$}.

    $\cO$ is called \emph{finite-dimensional} if there exists a dimension function on $\cO$.
\end{definition}

A dimension function assigns an integral size to every projection $p\in\PO$, and we will write $\cP_k(\cO)=\{p\in\PO\mid\dim(p)=k\}$ with $1\leq k\leq d$ for the set of $k$-dimensional projections in $\PO$. Often, we will not be interested in the exact choice of dimension function, and simply say $\cO$ has dimension $d$ if there exists a dimension function such that  $\dim(\One)=d$. For instance, the (minimal) dimension function on $\PH$ is given by $\dim(p)=\mathrm{rk}(p)$ for every $p\in\PH$. Clearly, by restricting to finite-dimensional observable algebras, we ensure that observables have only finitely many outcomes.\footnote{However, the converse is generally not true: there exist observable algebras whose elements are finite random variables, yet which do not admit a dimension function (see App.~\ref{app: On realisations}).} Notably, this is a restriction of the type of observable algebra in Def.~\ref{def: abstract observable algebra}. Nevertheless, since we will restrict to finite-dimensional observable algebras throughout we will sometimes omit stating this explicitly.

\begin{definition}\label{def: max OA}
    Let $\cO$ be a finite-dimensional observable algebra of dimension $\dim(\One)=d$. Then $\cO$ is called \emph{maximal} if every maximal context $C\in\COm$ is minimally generated by $d$ elements, equivalently if there exists a resolution of the identity $\One=\sum_{i=1}^d p^C_i$ for minimal projections $p^C_i\in\PC$ with $\dim(p^C_i)=1$.\footnote{Since $p^C_i\in\POone$ are minimal, this implies that $p^C_ip^C_j=0$ whenever $i\neq j$.}
\end{definition}

We will write $|C|:=|\cP_\mathrm{min}(C)|$ for the cardinality of a minimal generating set of $C$.\footnote{Here, $p\in\PC$ is called minimal, and denoted $p\in\mc{P}_\mathrm{min}(C)$, if $p'p=p'$ implies $p'=p$ for all $p'\in\PC$. A set of observables, e.g. dichotomic observables, equivalently events is called a minimal generating set if removing any element of the set will generate a strictly smaller context.} Clearly, if an observable algebra $\cO$ is maximal then all maximal contexts $C\in\COm$ are isomorphic. However, the converse is not true, since a maximal observable algebra also specifies a dimension function which may not exist (see App.~\ref{app: On realisations}). Moreover, not every finite-dimensional observable algebra is maximal, since the (discrete) spectra of observables in an observable algebra $\cO$ generally do not have the same cardinality. Consequently, they may generate commutative algebras that are not isomorphic. Note that quantum observable algebras $\cO=\oplus_i\cL(\cH_i)$ are maximal, with minimal projections given by rank-$1$ projections. While we will mostly restrict the discussion to maximal observable algebras in the main text, in App.~\ref{app: non-maximal OAs} we show that the characterisation of Kochen-Specker contextuality for (finite-dimensional) observable algebras can be reduced to that of maximal ones.

\subsection{Kochen-Specker contextuality for observable algebras}\label{sec: KSNC for OAs}

Having introduced the necessary terminology in the previous section, it is straightforward to generalise the notion of Kochen-Specker (KS) noncontextuality in Ref.~\cite{KochenSpecker1967,Frembs2024a} to observable algebras. We reiterate that in their original work, Kochen and Specker considered partial algebras which are observable algebras satisfying in addition Specker's principle \cite{KochenSpecker1967}. We briefly review the motivation behind the original definition of KS contextuality, before reformulating it in terms of an embedding problem for observable algebras. For a more detailed introduction, see Ref.~\cite{KochenSpecker1967,DoeringFrembs2019a,Frembs2024a}.\\

\textbf{Kochen-Specker contextuality as an embedding problem.} Recall that an observable $O\in C\subset\cO$ has a spectral decomposition $O = \sum_{i=1}^{|C|}\sigma_ip_i$ for $\sigma_i\in\spec(O)$, $p_i\in\PC$ and $|C|=|\cP_\mathrm{min}(C)|$. Moreover, $h(O)=\sum_{i=1}^{|C|} h(\sigma_i)p_i\in\cO$ for any function $h:\R\ra\R$. We call an observable \emph{non-degenerate (for $C\in\CO$)} if $|\spec(O)|=|C|$ and \emph{degenerate} otherwise. It is easy to see that every context is generated by a single non-degenerate observable,\footnote{More generally, every abelian von Neumann algebra is generated by a single Hermitian operator.}
\begin{align}\label{eq: generating observables}
    C(O) = \{h(O) \mid h: \R \ra \R\}\; .
\end{align}
Note that a generating observable for $C$ is not unique, however any two generating observables $O,O'$ for $C$ are related by an invertible function $h: \R\ra\R$ such that $O' = h(O)$. If $h$ is non-invertible on $\spec(O)$, then $\spec(h(O))\subset\spec(O)$ and $h(O)$ is degenerate.\footnote{Observable algebras thus generalise the situation in quantum mechanics, where two self-adjoint operators $O,O'\in\LHsa$ commute if and only if there exists another self-adjoint operator $O^*\in\LHsa$ and (measurable) functions $h,h':\R\ra\R$ such that $O=h(O^*)$ and $O'=h'(O^*)$ \cite{BuschLahtiPellonpaa_QuantumMeasurement}.}

It is natural to require maps between observable algebras $\epsilon:\cO\ra\cO'$ to preserve such functional relationships between observables in the same context, that is,\footnote{Again, if we take these functional relations to be part of an experimenter's ability to reason about and analyse her measurement data, they should be independent of our theoretical model for $\cO$.}\footnote{One may raise the objection that we should not identify the observables $g(O)$ and $g'(O')$ if $O,O'$ are generating observables of two distinct contexts $C,C'\in\CO$, since by our interpretation of such functional relations, the observables $g(O)$ and $g'(O')$ a priori also correspond to different measurement contexts (that of $O$ and $O'$, respectively). The response to this objection is that it is natural to identify observables (and their measurement contexts) mathematically whenever they cannot be distinguished experimentally, that is, whenever both yield the same statistics in any state (cf. Ref.~\cite{Spekkens2007}).}
\begin{equation}\label{eq: KSNC}
    \forall O \in \cO, \forall g:\R\ra\R: \quad \epsilon(g(O))=g(\epsilon(O))\; .
\end{equation}
Following Ref.~\cite{KochenSpecker1967,Frembs2024a}, we call any faithful map $\epsilon:\cO\ra\cO'$ that preserves the functional constraints in Eq.~(\ref{eq: KSNC}) an \emph{embedding of $\cO$}. Specifically, Kochen-Specker noncontextuality is concerned with the existence of a \emph{classical embedding of $\cO$}, that is, an embedding into a classical observable algebra of the form $\cO'=L_\infty(\Lambda)$ for a measurable space $\Lambda$.\footnote{If $\cO$ is a finitely generated observable algebra, then $\Lambda$ can be taken to be a finite set. Note also that a finitely generated observable algebra with a classical embedding is automatically finite-dimensional (see Fn.~\ref{fn: fin gen and classical implies fin dim}). However, for general finite-dimensional observable algebras $\Lambda$ need not be finite.}

\begin{definition}\label{def: KSNC}
    Let $\cO$ be an observable algebra. Then $\cO$ is called \emph{Kochen-Specker (KS) noncontextual} if it admits a classical embedding. Otherwise, $\cO$ is called \emph{KS contextual}.\footnote{The notion of KS noncontextuality is arguably a minimal requirement for a classical representation of $\cO$. It only refers to the algebraic structure, in particular, it does not also demand that states on $\cO$ be represented classically (cf. Def.~\ref{def: fully classically correlated}), which generically results in a stronger condition (see Sec.~\ref{sec: algebraic vs marginal approach}).}
\end{definition}

In fact, any (not just a classical) embedding encodes a notion of \emph{noncontextuality}, namely the representation $\epsilon(O)\in\cO'$ of $O\in\cO$ (e.g. in terms of a measurable function on some measurable space) does not depend on the context, that is, on other jointly measurable observables: if $\tO\in C,C'$ with $C\neq C'$, generating observables $C=C(O)$ and $C'=C(O')$ and functions $\tO=g(O)=g'(O')$, then $g(\epsilon(O))=\epsilon(\tO)=g'(\epsilon(O'))$. Since the identity map is a valid (classical) embedding, every classical system is Kochen-Specker noncontextual. In contrast, the famous Bell-Kochen-Specker theorem demonstrates that quantum theory is Kochen-Specker contextual \cite{Bell1966,KochenSpecker1967}. As one of the most crucial insights into the structure of quantum mechanics, it has sparked a vast array of research in the almost six decades since it was first formulated. In particular, contextuality and nonlocality fuel the promise that quantum computation will one day outperform present-day classical computers \cite{Raussendorf2013,HowardEtAl2014,FrembsRobertsBartlett2018,BravyiGossetKoenig2018,FrembsRobertsCampbellBartlett2023}. In this setting, the embedding problem in Def.~\ref{def: KSNC} (for quantum subsystems) is closely related to the existence of a classical simulation \cite{AaronsonGottesman2004,Gross2006,JozsaMiyake2008,PashayanBartlettGross2020,XuEtAl2023}. Yet, despite significant progress both on the foundational and computational side, the mathematical structure underlying Kochen-Specker contextuality has remained obscure \cite{YuOh2012,BudroniEtAl2022,JokinenEtAl2024}. In order to further unveil it, we first cast the notion of Kochen-Specker noncontextuality into context form.\\

\textbf{Kochen-Specker contextuality as order structure.} The different structures in Sec.~\ref{sec: OAs} - observable algebras, event algebras and the partial order of contexts - have given rise to various definitions of contextuality, e.g. for measurement scenarios \cite{AbramskyBrandenburger2011}, (hyper)graphs \cite{CabelloSeveriniWinter2014,ChavesFritz2012,AcinFritzLeverrierSainz2015}, partial orders of contexts etc. \cite{IshamButterfieldI,IshamDoeringI,DoeringIsham2011,HLS2010,DoeringFrembs2019a}. This easily creates the impression of various inequivalent notions of contextuality. In Sec.~\ref{sec: algebraic vs marginal approach} (and Sec.~\ref{sec: algebraic vs graph-theoretic approach}), we will indeed identify two distinct notions of contextuality, namely the algebraic notion of KS noncontextuality and the notion based on (demanding) correlations (to be classical). Despite this distinction, the underlying structures are equivalent (up to redundancies, see App.~\ref{app: unified perspective}). In particular, they are independent of the representation in terms of observable algebras, event algebras or context categories.

\begin{lemma}\label{lm: embeddings vs orthomorphisms vs order-preserving maps}
    Let $\cO,\cO'$ be finite-dimensional observable algebras with context categories $\CO,\cC(\cO')$ and event algebras $\PO,\cP(\cO')$, respectively. Then $\epsilon:\cO\ra\cO'$ is an embedding if and only if $\varphi:\PO\ra\cP(\cO')$ is a faithful orthomorphism\footnote{Here, by an orthomorphism we mean a map $\varphi:\PO\ra\cP(\cO')$ such that $\varphi(\One)=\One'$, $\varphi(\One-p)=\One'-\varphi(p)$ for all $p\in\PO$ as well as $\varphi(p+q)=\varphi(p)+\varphi(q)$ for all $p,q\in\PO$, $pq=0$.} if and only if $\eta:\CO\ra\cC(\cO')$ is a faithful order-preserving map.
\end{lemma}

\begin{proof}
    Clearly, an embedding of observable algebras $\epsilon:\cO\ra\cO'$ restricts to a faithful orthomorphism $\varphi=\epsilon|_\PO:\PO\ra\cP(\cO')$. Conversely, every faithful orthomorphism $\varphi:\PO\ra\cP(\cO')$ generates an embedding of observable algebras $\epsilon:\cO\ra\cO'$ by setting $\epsilon(O)=\sum_{i=1}^{|C|} \sigma_i\varphi(p_i)$ for every $\cO\ni O=\sum_{i=1}^{|C|}\sigma_ip_i$ with $p_i\in\PC$ and $|C|=|\cP_\mathrm{min}(C)|$.

    Since the order relations in $\CO$ are defined by the functional constraints in $\cO$, every embedding of observable algebras $\epsilon:\cO\ra\cO'$ induces a map $\eta:\CO\ra\cC(\cO')$ of the form $\eta(C)=\epsilon[C]$. Clearly, $\eta$ is faithful. Moreover, let $\tC,C\in\CO$ be contexts in $\cO$ with generating observables $\tO,O\in\cO$, that is, $\tC=\tC(\tO)$ and $C=C(O)$ as in Eq.~(\ref{eq: generating observables}), and such that $\tO=g(O)$. Then it follows from Eq.~(\ref{eq: KSNC}) that $\eta$ preserves order relations:
    \begin{align*}
        \eta(\tC)=\epsilon[\tC]
        &=\epsilon[\{\widetilde{h}(\tO)\mid \widetilde{h}:\R\ra\R\}] \\
        &=\epsilon[\{\widetilde{h}(g(O))\mid \widetilde{h}:\R\ra\R\}] \\
        &\subset\epsilon[\{h(O)\mid h:\R\ra\R\}]
        =\epsilon[C] = \eta(C)\; .
    \end{align*}
    
    Conversely, from a faithful order-preserving map $\eta:\CO\ra\cC(\cO')$, we construct a faithful orthomorphism (and thus a faithful embedding by the first assertion) as follows. First, we represent the minimal element $C_\One\in\CO$ by the unique two-element Boolean algebra $\cB(C_\One)$ generated by a single element which we denote by $\One$ and $0=\One^\perp$. Second, for every context $C_\One\neq C\in\CO$ such that $C\geq C'>C_\One$ implies $C=C'$ (here, we use that $\cO$ is finite-dimensional), define a Boolean algebra $\BC$ generated by two elements $p_C$ and $\One$. Proceeding in this manner, by induction in the length of chains in $\CO$, we thus obtain representations of the elements $C\in\CO$ in terms of Boolean algebras $\BC$ and their minimal generating set. Similarly, we obtain a representation of the contexts $C'\in\cC(\cO')$ in terms of Boolean algebras $\cB(C')$ and their generating elements.
    
    Now, define $\varphi:\PO\ra\cP(\cO')$ by $\varphi(p_C)=p_{C'}$ whenever $\eta(C)=C'$, in particular, we have $\varphi(\One)=\One'$. By construction, we further have $\varphi(\One-p_C)=\One'-\varphi(p_C)$ for all $C\in\CO$, hence, for all $p_C\in\PO$, as well as $\varphi(p_\tC+p_C) = \varphi(p_\tC)+\varphi(p_C)$ whenever $\tC\subset C$, hence for all $p_\tC,p_C\in\PC$ with $p_\tC p_C=0$. Consequently, $\varphi$ is an orthomorphism, it is faithful since $\eta$ is, and extends to an embedding by the first assertion.
\end{proof}

Lm.~\ref{lm: embeddings vs orthomorphisms vs order-preserving maps} shows that even though we formulated Kochen-Specker noncontextuality in terms of observable algebras in Def.~\ref{def: KSNC}, we could have equivalently defined it in terms of context categories. In other words, the functional constraints in Eq.~(\ref{eq: KSNC}) are fully contained in, and thus a property of the partial order of contexts $\CO$ only. What is more, this is true not only for classical embeddings, but applies to embeddings into observable algebras of other forms, too, e.g. observable algebras in quantum theory $\cO\subset\LHsa$. This will become important when we compare the notion of Kochen-Specker contextuality in Def.~\ref{def: KSNC} with other notions of contextuality, specifically with those in the marginal and graph-theoretic approach. We thus echo, yet generalise the insight from Ref.~\cite{IshamButterfieldI,IshamDoeringI,IshamDoeringII,IshamDoeringIII,IshamDoeringIV,DoeringFrembs2019a}: \emph{(Kochen-Specker) contextuality is order structure.}

But how do we infer whether $\cO$ is Kochen-Specker contextual from $\CO$ only? Observe that the same function $g:\R\ra\R$ generally induces different context inclusions for different choices of generating observables of a context $C\in\CO$, that is, $C(g(O)) \neq C(g(O'))$ even if $C=C(O)=C(O')$. In particular, when $O,O'\in\cO$ share the same spectrum $\spec(O)=\spec(O')$, yet their spectral decomposition is distinct (related by a permutation of the generating projections), then the action of $g$ yields observables generating different subcontexts. As we will see (in Thm.~\ref{thm: CNC for OAs}) in the next section, the interplay of such relations across contexts fully characterises KS contextuality.

\subsection{Complete characterisation of Kochen-Specker contextuality}\label{sec: CNC for OAs}

Lm.~\ref{lm: embeddings vs orthomorphisms vs order-preserving maps} asserts that whether an observable algebra $\cO$ is Kochen-Specker (non)contextual is a property of its partial order of contexts $\CO$. In order to extract from $\CO$ the relevant information to solve the classical embedding problem in Def.~\ref{def: KSNC}, we generalise the tools introduced in Ref.~\cite{Frembs2024a} from quantum to maximal observable algebras.\\

\textbf{Context connection.} We first generalise Def.~1 in Ref.~\cite{Frembs2024a}.

\begin{definition}\label{def: context connection}
    Let $\cO$ be a maximal observable algebra of dimension $\dim(\One)=d$ with partial order of contexts $\CO$. A \emph{(context) connection $\fl = (l_{C'C})_{C,C'\in\COm}$ on $\CO$} is a collection of bijective maps $l_{C'C}:\mc{P}_1(C)\ra\mc{P}_1(C')$ with $l_{CC'} = l^{-1}_{C'C}$ and such that
    \begin{align}\label{eq: noncontextual context connection}
        l_{C'C}|_{\mc{P}_1(C\cap C')} = \id \quad \forall C,C' \in \COm\; .
    \end{align}
\end{definition}

Any context connection thus respects the identification of projections $\mc{P}(C \cap C')$ in subcontexts $C \cap C'$ for all $C,C'\in\COm$. Note that since $|\mc{P}_1(C)| = |\mc{P}_1(C')|=d$ by assumption, $l_{C'C}|_{\cP_1(C)}$ is a bijection (see Fig.~\ref{fig: context connection}).\\

\begin{figure}
    \centering
    \scalebox{1.3}{\begin{tikzpicture}[node distance=2.75cm, every node/.style={scale=0.7}]
    \node(Vp1)                        {};
    \node(Va2c)   [above= 1cm of Vp1]        {$C$};
    \node(V1bc)   [left=0.5cm of Va2c]        {};
    \node(Vp2)   [right of=Vp1]        {$C\cap C'$};
    \node(Vp3)   [right of=Vp2]        {$C' \cap C''$};
    \node(Vab3)   [above= 1cm of Vp3]        {};
    \node(V123)   [above= 1cm of Vp2]        {$C'$};
    
    \node(Vp'2)   [right of=Vp3]      {$C'' \cap C'''$};
    \node(Vp'1)   [right of=Vp'2]       {};
    \node(Va2'c)   [above=1cm of Vp'1]        {$C'''$};
    \node(V1'bc)   [right=0.5cm of Va2'c]        {};
    \node(V1'2'3)   [above= 1cm of Vp'2]      {$C''$};
    
    \node(Vone)  [below= 1cm of Vp3]  {$C_\One$};
    
    \node(3DotsUpLeft)     [left=0.3cm of V1bc]    {$\cdots$};
    \node(3DotsUpRight)     [right=0.3cm of V1'bc]    {$\cdots$};
    \node(3DotsDownLeft)   [left=0.3cm of Vp1]    {$\cdots$};
    \node(3DotsDownRight)   [right=0.3cm of Vp'1]    {$\cdots$};

    \draw [->,dashed](Vone) -- (Vp1);
    \draw [->](Vone) -- (Vp2);
    \draw [->](Vone) -- (Vp3);
    \draw [->](Vone) -- (Vp'2);
    \draw [->,dashed](Vone) -- (Vp'1);
    
    \draw [->,dashed](Vp1) -- (V123);
    \draw [->,dashed](Vp1) -- (V1bc);
    \draw [->](Vp2) -- (V123);
    \draw [->](Vp2) -- (Va2c);

    \draw [->](Vp3) -- (V123);
    \draw [->](Vp3) -- (V1'2'3);
    \draw [->,dashed](Vp3) -- (Vab3);
    
    \draw [->](Vp'2) -- (V1'2'3);
    \draw [->](Vp'2) -- (Va2'c);
    \draw [->,dashed](Vp'1) -- (V1'2'3);
    \draw [->,dashed](Vp'1) -- (V1'bc);

    \draw [->,db,thick] (Va2c) to[bend left] node[midway,below] {\large\textbf{$l_{C'C}$}} (V123);
    \draw [->,db,thick] (V123) to[bend left] node[midway,below] {\large\textbf{$l_{C''C'}$}} (V1'2'3);
    \draw [->,db,thick] (Va2c) to[bend left] node[midway,below] {\large\textbf{$l_{C''C}\ \ \ \ \ $}} (V1'2'3);
    \draw [->,db,thick] (V1'2'3) to[bend left] node[midway,below] {\large\textbf{$l_{C'''C''}$}} (Va2'c);
    \draw [->,db,thick] (Va2c) to[bend left] node[midway,above] {\large\textbf{$l_{C'''C}$}} (Va2'c);
    \draw [->,db,thick,dashed] (Va2'c) to[bend left] (Vab3);
    \draw [->,db,thick] (V123) to[bend left] node[midway,below] {\large\textbf{$\ \ \ \ \ \ \ l_{C'''C'}$}} (Va2'c);

    \draw [->,db,thick,dashed] (Vab3) to[bend left] (V123) ;
    \draw [->,db,thick,dashed] (V1'2'3) to[bend left] (Vab3) ;
    \draw [->,db,thick,dashed] (Vab3) to[bend left] (Va2c);

    \draw [-,db,thick,dotted] (V1bc) -- (Va2c) ;
    \draw [-,db,thick,dotted] (Va2'c) -- (V1'bc) ;
\end{tikzpicture}}
    \caption{Schematic of a context connection $\fl = (l_{C'C})_{C,C' \in \COm}$ (reprinted from Ref.~\cite{Frembs2024a}). Context connections preserve elements in subcontexts: $l_{C'C}|_{\mc{P}_1(C\cap C')} = \mathrm{id}$ for all maximal contexts $C,C' \in \COm$.}
    \label{fig: context connection}
\end{figure}
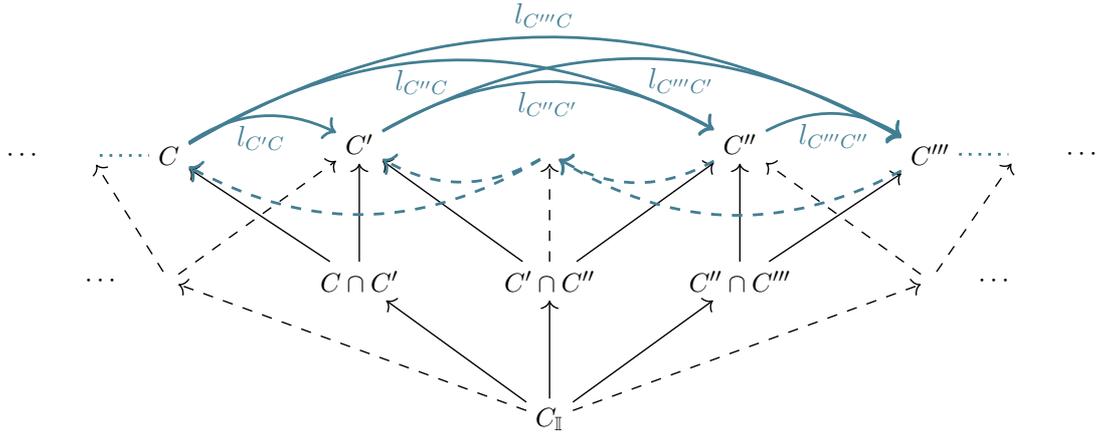

\textbf{Context cycles.} Below, we will consider the noncontextuality constraints imposed by Eq.~(\ref{eq: KSNC}), and encoded in the inclusion relations of $\CO$ by Lm.~\ref{lm: embeddings vs orthomorphisms vs order-preserving maps}, across multiple contexts at the same time. To this end, we generalise Def.~2 from Ref.~\cite{Frembs2024a}.

\begin{definition}\label{def: context cycle}
    A \emph{context cycle $\gamma$ in $\COm$} is a tuple $\gamma=(C_0,\cdots,C_{n-1})$ of 
    maximal contexts $C_i\in\COm$, together with subcontexts $C_{i\cap i+1}=C_i \cap C_{i+1}$ for $i\in\zz_n$.\footnote{Here, $\zz_n$ denotes the abelian group under addition modulo $n$, in particular, $(n-1)+1 = 0 \in \zz_n$.}
\end{definition}

Context cycles have been considered elsewhere in similar form, e.g. in the study of sheaf cohomology \cite{AbramskyEtAl2015,Caru2017,BeerOsborne2018,Aasnass2020}. However, our formalism based on context connections and the results we obtain differ from those in these works. In particular, our approach is algebraic in nature, and not in terms of states (empirical models defined on measurement scenarios, see App.~\ref{app: measurement scenarios vs OAs}), and the characterisation of KS contextuality in Thm.~\ref{thm: CNC for OAs} is complete, not restricted to cyclic or finitely generated models. In Sec.~\ref{sec: algebraic vs marginal approach}, we will compare the algebraic and marginal approach to contextuality in more detail.\\

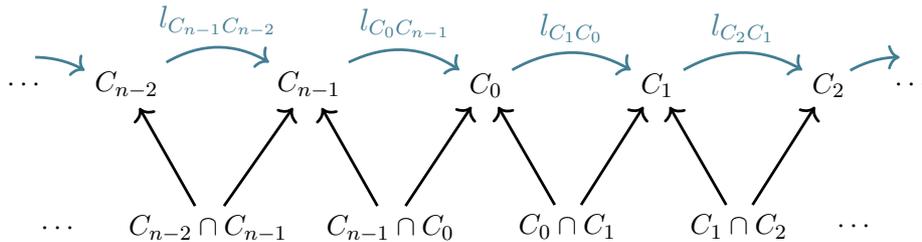
\begin{figure}
    \centering
    \scalebox{1.3}{\begin{tikzpicture}[node distance=3cm, every node/.style={scale=0.75}]
    \node(V-2)   {$C_{n-2}$};
    \node(V-1)   [right=1cm of V-2]        {$C_{n-1}$};
    \node(V0)   [right=1.1cm of V-1]        {$C_0$};
    \node(V1)   [right=1.2cm of V0]        {$C_1$};
    \node(V2)   [right=1.2cm of V1]        {$C_2$};
    
    \node(V-21)   [below=1cm of V-2, xshift=1.1cm]   {$C_{n-2}\cap C_{n-1}$};
    \node(V-10)   [below=1cm of V-1, xshift=1.1cm]   {$C_{n-1}\cap C_0$};
    \node(V01)   [below=1cm of V0, xshift=1.1cm]   {$C_0\cap C_1$};
    \node(V12)   [below=1cm of V1, xshift=1.1cm]   {$C_1\cap C_2$};
    
    \node(3DotsUpLeft)     [left=0.3cm of V-2]    {$\cdots$};
    \node(3DotsUpRight)     [right=0.3cm of V2]    {$\cdots$};
    \node(3DotsDownLeft)   [left=0.3cm of V-21]    {$\cdots$};
    \node(3DotsDownRight)   [right=0.3cm of V12]    {$\cdots$};

    \draw [->,thick] (V-21) -- (V-2);
    \draw [->,thick] (V-21) -- (V-1);
    \draw [->,thick] (V-10) -- (V-1);
    \draw [->,thick] (V-10) -- (V0);
    \draw [->,thick] (V01) -- (V0);
    \draw [->,thick] (V01) -- (V1);
    \draw [->,thick] (V12) -- (V1);
    \draw [->,thick] (V12) -- (V2);

    \draw [->,db,thick] (V-2) to[bend left] node[midway,above] {\large$l_{C_{n-1}C_{n-2}}$} (V-1);
    \draw [->,db,thick] (V-1) to[bend left] node[midway,above] {\large$l_{C_0C_{n-1}}$} (V0);
    \draw [->,db,thick] (V0) to[bend left] node[midway,above] {\large$l_{C_1C_0}$} (V1);
    \draw [->,db,thick] (V1) to[bend left] node[midway,above] {\large$l_{C_2C_1}$} (V2);

    \node(V-2x)   [left=-0.1cm of V-2,yshift=0.2cm]   {};
    \draw [<-,db,thick] (V-2x) arc (60:90:1);
    \node(V2x)   [right=-0.15cm of V2,yshift=0.2cm]   {};
    \draw [->,db,thick] (V2x) arc (120:90:1);
\end{tikzpicture}}
    \caption{Schematic of (black) a context cycle $\gamma = (C_0,\cdots,C_{n-1})$, with $C_i \in \COm$ and $C_{(i+1)\cap i} = C_{i+1} \cap C_i$ for all $i \in \zz_n$, and (blue) elements of a context connection $\fl$ (reprinted from Ref.~\cite{Frembs2024a}).}
    \label{fig: context cycle}
\end{figure}

\textbf{Algebraic characterisation of Kochen-Specker noncontextuality.} The main result of this section is a generalisation of Thm.~2 in Ref.~\cite{Frembs2024a} from quantum to general finite-dimensional, but not necessarily maximal observable algebras (see Tab.~\ref{tab: summary}).

\begin{theorem}\label{thm: CNC for OAs}
    Let $\cO$ be a finite-dimensional observable algebra with partial order of contexts $\CO$. Then $\cO$ is Kochen-Specker noncontextual, that is, admits a classical embedding $\epsilon: \cO \ra L_\infty(\Lambda)$, if and only if there exists a context connection $\fl = (l_{C'^*C^*})_{C^*,C'^*\in\sCOm}$ on the partial order of contexts $\sCO\supset\CO$ of its maximal extension in Lm.~\ref{lm: min extension to max observable algebra}, which satisfies the \emph{triviality} constraints,
    \begin{equation}\label{eq: CNC}
        \circ_{i=0}^{n-1} l_{C^*_{(i+1)}C^*_i} = \id\; ,
    \end{equation}
    for every context cycle $\gamma = (C^*_0,\cdots,C^*_{n-1})$ in $\cC_\mathrm{max}(\cO^*)$.
\end{theorem}

A context connections satisfying the triviality constraints in Eq.~(\ref{eq: CNC}) will be called \emph{flat}.

\begin{proof}
    (Sketch) The proof of Thm.~2 in Ref.~\cite{Frembs2024a} only uses that maximal commutative algebras have the same cardinality $|C|=|\cP_1(C)|=d$ for all $C\in\cC_\mathrm{max}(\cO)$, where for quantum observable algebras $\cO=\LHsa$ the dimension $d=\dim(\cH)$ is determined by the dimension of the Hilbert space $\cH$. The result thus carries over easily to maximal observable algebras. While Thm.~2 in \cite{Frembs2024a} established that the constraints in Eq.~(\ref{eq: CNC}) are necessary for KS noncontextuality, here we also show that they are sufficient. We provide the details of the proof for the case of maximal observable algebras in App.~\ref{app: proof of CNC for OAs}.

    Moreover, by Lm.~\ref{lm: min extension to max observable algebra} in App.~\ref{app: non-maximal OAs}, every finite-dimensional observable algebra $\cO$ can be extended to a maximal observable algebra $\cO^*$ and, by Lm.~\ref{lm: KSNC inherits from maximality}, $\cO$ is Kochen-Specker noncontextual if and only if its maximal extension $\cO^*$ is.
\end{proof}

Numerous proofs of the Kochen-Specker theorem construct so-called Kochen-Specker sets (see Def.~\ref{def: KS-colouring}) \cite{KochenSpecker1967,Mermin1990,Peres1991,ZimbaPenrose1993,Kernaghan1994,KernaghanPeres1995,CabelloEstebanzGarcia-Alcaine1997}, many of which are easily seen to directly exploit the cyclic nature of the constraints in Eq.~(\ref{eq: CNC}). Moreover, ``cyclic contextuality scenarios" have been studied elsewhere \cite{AraujoEtAl2013,Caru2017,BeerOsborne2018,DzhafarovKujalaCervantes2020}, in particular, Ref.~\cite{AraujoEtAl2013,XuCabello2019} characterise cyclic contextuality scenarios in the marginal approach to contextuality. However, as we will see in Sec.~\ref{sec: algebraic vs marginal approach}, the relevant notion of contextuality differs from Kochen-Specker noncontextuality in Def.~\ref{def: KSNC}. Indeed, somewhat surprisingly, we will show that cyclic scenarios are never Kochen-Specker noncontextual in Sec.~\ref{sec: single cycles}.

Yet, before turning to a detailed comparison between the algebraic and marginal approach in Sec.~\ref{sec: algebraic vs marginal approach}, we first cast Thm.~\ref{thm: CNC for OAs} as a colouring problem, which will allow us to contrast it with the notion of Kochen-Specker colourings which feature in traditional proofs of the Kochen-Specker theorem. Moreover, this reformulation will serve as a preparation for our characterisation of state-independent contextuality graphs in Sec.~\ref{sec: algebraic vs graph-theoretic approach}.

\subsection{Kochen-Specker colouring vs d-colouring}\label{sec: KS-colouring vs colouring}

Following the original argument by Kochen and Specker \cite{KochenSpecker1967}, many proofs of the Kochen-Specker (BKS) theorem\footnote{As in Ref.~\cite{Frembs2024a}, by the ``Kochen-Specker theorem" for an observable algebra $\cO$ we mean the assertion that $\cO$ does not admit a classical embedding, hence, a proof that $\cO$ is Kochen-Specker (KS) contextual. We recall that Ref.~\cite{Bell1966} first proved that finite-dimensional quantum systems of dimension at least three are KS contextual, while Ref.~\cite{KochenSpecker1967} inferred this result by constructing a finite KS set, equivalently a finitely generated observable subalgebra $\cO\subset\LHsa$ with $\dim(\cH)=3$.\label{fn: "KS-theorem"}} show the nonexistence of \emph{valuations} $v:\cO\ra\R$ such that (i) $v(O)\in\spec(O)$ for all $O\in\cO$ (``spectrum rule") and (ii) $v(g(O))=g(v(O))$ for all functions $g:\R\ra\R$ (``functional composition (FUNC) principle''). It is easy to see that a valuation defines a \emph{Kochen-Specker (KS)} or \emph{$2$-colouring} by restriction.

\begin{definition}\label{def: KS-colouring}
    Let $\cO$ be a maximal observable algebra of dimension $\dim(\One)=d$ with event algebra $\PO$ and partial order of contexts $\CO$. A \emph{Kochen-Specker (KS-)colouring $\zeta_\mathrm{KS}$ of $\cO$} is a map $\zeta_\mathrm{KS}:\POone\ra\{0,1\}$ such that for every maximal context $C\in\COm$ there exists exactly one one-dimensional projection $p\in\cP_1(C)$ with $\zeta_\mathrm{KS}(p) = 1$.
\end{definition}

Conversely, by linear extension any Kochen-Specker colouring induces a map $v:\cO\ra\R$ satisfying (i) and (ii), hence, defines a valuation. Any subset of projections not admitting a valuation, equivalently a KS colouring, is called a \emph{Kochen-Specker (KS) set}.

However, the existence of valuations is only a necessary but not sufficient condition for KS noncontextuality. This was emphasised in the marginal approach to contextuality \cite{BudroniEtAl2022} (see also Ref.~\cite{Frembs2024a} and Sec.~\ref{sec: algebraic vs marginal approach} below), where it was shown that there exist noncontextuality inequalities that are violated by every quantum state, hence, yield a proof of the KS theorem, yet for which the involved observables, equivalently their generating one-dimensional projections $\cP\subset\PHone$, do not constitute a KS set \cite{YuOh2012}. In Ref.~\cite{Frembs2024a}, we showed that this is not in contradiction with the notion of KS contextuality, but that the proof of the KS theorem in Ref.~\cite{YuOh2012} also follows from Eq.~(\ref{eq: KSNC}). Motivated by this result, and fostered by Thm.~\ref{thm: CNC for OAs}, we now lift this argument and cast the full set of KS noncontextuality constraints in Eq.~(\ref{eq: KSNC}) as a colouring problem.

\begin{definition}\label{def: n-colouring}
    Let $\cO$ be a maximal observable algebra of dimension $\dim(\One)=d$, event algebra $\PO$ and partial order of contexts $\CO$. A \emph{$d$-colouring of $\cO$} is a map $\zeta: \mc{P}_1(\cO) \ra [d]$ such that for every maximal context $C\in\COm$ there exists exactly one one-dimensional projection $p\in\cP_1(C)$ with $\zeta(p)=k$ for every $k\in[d]=\{0,\cdots,d-1\}$.
\end{definition}

We say that $\cO$ is $d$-colourable if it admits a $d$-colouring. Clearly, a $d$-colouring (for $d\geq 2$) induces a KS-colouring (hence, a valuation) by identifying any one colour with $1$ and the remaining $d-1$ colours with $0$. Yet, a KS-colouring can generally not be extended to a $d$-colouring. For an explicit example, see Ref.~\cite{YuOh2012,Frembs2024a}.

For ease of comparison with the graph-theoretic approach in Sec.~\ref{sec: algebraic vs graph-theoretic approach}, the following definition and terminology will be useful.

\begin{definition}[associated orthogonality graph]\label{def: associated orthogonality graph}
    Let $\cO$ be an observable algebra with event algebra $\PO$. Then its \emph{associated orthogonality graph $\GO=(V,E)$} is defined by $V=\POm$, where $(p,q)\in E$ if and only if $p\neq q$ and $p,q\in\PC$ for some context $C\subset\cO$.
\end{definition}

It is easy to see that the existence of a $d$-colouring of a maximal observable algebra $\cO$ is equivalent to $\chi(G(\cO))=d$.

\begin{theorem}\label{thm: KSNC = n-colourability}
    Let $\cO$ be a maximal observable algebra with $\dim(\One)=d$, event algebra $\PO$ and partial order of contexts $\CO$. Then $\cO$ is Kochen-Specker noncontextual if and only if $\cO$ admits a $d$-colouring, equivalently $\chi(G(\cO))=d$.
\end{theorem}

\begin{proof}
    Assume first that $\cO$ admits a $d$-colouring $\zeta$, and denote by $\zeta_C: \cP_1(C) \ra [d]$ the respective colouring in maximal contexts $C\in\COm$. From these we define  bijections $l_{C'C}: \cP_1(C) \ra \cP_1(C')$ by $\zeta_{C'} \circ l_{C'C} = \zeta_C$, that is, by identifying those (one-dimensional) projections in $C$ and $C'$ that are assigned the same colour by $\zeta$. Lifting the maps $l_{C'C}$ to orthomorphisms defines a context connection on $\cO$: the condition $l_{C'C}|_{\cP(C\cap C')} = \id$ holds since $\zeta$ is a $d$-colouring which implies $\zeta_C|_{\cP_1(C\cap C')} = \zeta_{C'}|_{\cP_1(C\cap C')}$. Moreover, since $\zeta=(\zeta_C)_{C\in\COm}$ is a $d$-colouring, $\fl=(l_{C'C})_{C,C'\in\COm}$ clearly satisfies the triviality constraints in Eq.~(\ref{eq: CNC}), hence, $\cO$ is KS noncontextual by Thm.~\ref{thm: CNC for OAs}.

    Conversely, assume that there exists a flat context connection $\fl = (l_{C'C})_{C,C'\in\COm}$ (that is, $
    \fl$ satisfies the triviality constraints in Eq.~(\ref{eq: CNC})). Pick any maximal context $C_0\in\COm$ and any $d$-colouring $\zeta_{C_0}:\cP_1(C_0)\ra[d]$. For every other context $C_0 \neq C\in\COm$, fix a $d$-colouring $\zeta^{C_0}_C:\cP_1(C)\ra[d]$ by $\zeta^{C_0}_C \circ l_{CC_0} = \zeta_{C_0}$. Clearly, $\zeta^{C_0}_C$ is compatible with $\zeta_{C_0}$, that is, $\zeta^{C_0}_C|_{\cP_1(C\cap C_0)} = \zeta_{C_0}|_{\cP_1(C \cap C_0)}$ since $l_{CC_0}|_{\cP_1(C\cap C_0)} = \id$ (by definition of a context connection). Now, $(\zeta^{C_0}_C)_{C\in\COm}$ defines a $d$-colouring if the colourings $\zeta^{C_0}_C$ and $\zeta^{C_0}_{C'}$ are compatible for all $C,C'\in\COm$. But this follows since
    \begin{align*}
        \zeta^{C_0}_{C'} \circ l_{C'C_0}
        = \zeta_{C_0}
        = \zeta^{C_0}_C \circ l_{CC_0}
        = \zeta^{C_0}_C \circ (l_{C_0C'} \circ l_{C'C})^{-1}
        = \zeta^{C_0}_C \circ l_{CC'} \circ l_{C'C_0}\; ,
    \end{align*}
    where we used $l_{C_0C'} \circ l_{C'C} \circ l_{CC_0} = \id$ from Eq.~(\ref{eq: CNC}) for the context cycle $(C_0,C,C')$ and $l^{-1}_{CC'}=l_{C'C}$. Indeed, multiplication by $l^{-1}_{C'C_0}$ from the right yields $\zeta^{C_0}_{C'} = \zeta^{C_0}_C \circ l_{CC'}$ for every $C,C'\in\COm$ which implies compatibility since $l_{C'C}|_{\cP_1(C\cap C')}=\id$.
\end{proof}

Thm.~\ref{thm: KSNC = n-colourability} makes explicit that the restriction to valuations (equivalently, KS-colourings) in traditional proofs of the KS theorem is a severe restriction on the constraints imposed by Kochen-Specker noncontextuality (see Tab.~\ref{tab: colourings for KS}).

Finally, we generalise Thm.~\ref{thm: KSNC = n-colourability} to non-maximal observable algebras.

\begin{theorem}\label{thm: n-colourability for non-maximal OAs}
    Let $\cO$ be an observable algebra with maximal extension $\cO^*$ for $\dim(\One)=d$. Then $\cO$ is Kochen-Specker noncontextual if and only if $\cO^*$ admits a $d$-colouring, equivalently $\chi(G(\cO^*))=d$.
\end{theorem}

\begin{proof}
    Since $\cO$ is KS contextual if and only if $\cO^*$ is KS contextual for any (hence, every) maximal extension by Lm.~\ref{lm: KSNC inherits from maximality}, the result follows immediately from Thm.~\ref{thm: KSNC = n-colourability}.
\end{proof}

Clearly, generally $\chi(G(\cO))<\chi(G(\cO^*))$ since $G(\cO)$ arises from $G(\cO^*)$ by subsequently deleting the vertices corresponding to minimal projections in $\cO^*$, and in each step the chromatic number either remains the same or is reduced by one. Nevertheless, in Sec.~\ref{sec: SI-C sets and graphs} we will see that it is sufficient to consider the chromatic number of subgraphs $G\subset G(\cO^*)$ whenever $\overline{\pi(V)}=\cO^*$ for $\pi:G\ra\cP_1(\cO^*)$ is a faithful completion (see Lm.~\ref{lm: restricted chromatic extension}).\\

In summary, the chromatic number of the orthogonality graph associated with any maximal extension of an observable algebra is a complete invariant of Kochen-Specker contextuality (see Tab.~\ref{tab: summary}).

\begin{table}
    \centering
    \begin{tabular}{c|ccc}
        \toprule[0.05cm]\\[-0.3cm]
        observable algebra $\cO$ & \multicolumn{3}{c}{(partial) invariants of Kochen-Specker contextuality} \\[0.1cm]
        \midrule\\[-0.3cm]
        \multirow{3}{4cm}{\centering maximal} & $\nexists \chi_\mathrm{KS}\ \ (\Leftrightarrow \cI_\mathrm{Val}(G(\cO))=\emptyset$) & $\stackrel{\mathrm{Ref.}~\cite{KochenSpecker1967}}{\Longrightarrow:}$ & $\cO$ KS contextual \\
        & $\Downarrow$ $\not\Uparrow$ & & \\
        & $\chi(G(\cO))>\dim(\One)$ & $\stackrel{\mathrm{Thm.}~\ref{thm: KSNC = n-colourability}}{\Longleftrightarrow}$ & $\cO$ KS contextual \\[0.2cm]
        $\cap$ & & $\Uparrow$ & \\[.2cm]
        finite-dimensional & $\chi(G(\cO^*))>\dim(\One)$ & $\stackrel{\mathrm{Thm.~\ref{thm: n-colourability for non-maximal OAs}}}{\Longleftrightarrow}$ & $\cO$ KS contextual \\[0.1cm]
        \bottomrule[0.05cm]
    \end{tabular}
    \caption{As first argued by Kochen and Specker, the existence of valuations is a necessary condition for KS noncontextuality \cite{KochenSpecker1967}. In turn, traditional proofs of KS contextuality show the nonexistence of valuations, equivalently Kochen-Specker colourings (see Def.~\ref{def: KS-colouring}) \cite{KochenSpecker1967,Mermin1990,Peres1991,ZimbaPenrose1993,Kernaghan1994,KernaghanPeres1995,CabelloEstebanzGarcia-Alcaine1997}. However, constructing KS sets is not necessary to prove KS contextuality \cite{YuOh2012}. By Thm.~\ref{thm: KSNC = n-colourability}, a necessary and sufficient condition for KS contextuality for the case of maximal observable algebras is the existence of a $d$-colouring. This generalises Thm.~4 in Ref.~\cite{Frembs2024a} which considers the case of maximal quantum observable algebras of dimension three. What is more, the characterisation extends to non-maximal observable algebras with respect to their maximal extensions by Thm.~\ref{thm: n-colourability for non-maximal OAs}.}
    \label{tab: colourings for KS}
\end{table}

\clearpage

\section{Algebraic vs marginal approach}\label{sec: algebraic vs marginal approach}

In this section, we relate the characterisation of Kochen-Specker contextual observable algebras in Thm.~\ref{thm: CNC for OAs} to some of the more commonly used formalisms for studying contextuality in the literature. Our focus will be on the marginal approach to contextuality \cite{BudroniEtAl2022}, which is based on the study of noncontextuality inequalities \cite{Pitowsky1991,Larsson2002,KlyachkoEtAl2008,Cabello2008b,BadziagBengtssonCabelloPitowsky2009,KleinmannEtAl2012,YuOh2012,XuCabello2019} and has been the dominating approach in recent years; it comprises the sheaf-theoretic \cite{AbramskyBrandenburger2011,AbramskyEtAl2015,AbramskyEtAl2017} and (hyper)graph-theoretic \cite{ChavesFritz2012,FritzChaves2013,CabelloSeveriniWinter2014,AcinFritzLeverrierSainz2015,AmaralCunha2018} variants. In App.~\ref{app: unified perspective}, we provide a comprehensive formal comparison between the respective frameworks, which permits us to cast the basic ingredients needed for the present exposition in the language of observable algebras defined in Sec.~\ref{sec: basics of marginal approach}. Based on this, our first objective is to precisely disambiguate the two notions of contextuality employed in the algebraic and marginal setting, respectively. In particular, we will translate the notion of KS contextuality into the marginal setting in Sec.~\ref{sec: marginal approach to KSC}.

Our second objective is to export the insights gained from Thm.~\ref{thm: CNC for OAs}, and to demonstrate their utility by explicitly showing how they allow us to easily recover and go beyond a number of existing key results in the literature. Specifically, we will prove that for acyclic observable algebras the constraints in Thm.~\ref{thm: CNC for OAs} are always satisfiable. In this way, we will obtain a new explanation for the existence of noncontextual models such as those devised in Ref.~\cite{Meyer1999,Kent1999,CliftonKent2000}. Moreover, in Sec.~\ref{sec: (a)cycles} we will relate this result with Vorob'ev's theorem (Thm.~\ref{thm: Vorob'ev theorem}), which states that acyclicity fully characterises classicality in the marginal approach. In contrast, we highlight that KS noncontextuality does not imply acyclicity, which we demonstrate by studying the general $n$-cycle contextuality scenario \cite{AraujoEtAl2013}, as well as an explicit example in the form of the well-known CHSH scenario \cite{CHSH1969}.

Nevertheless, we will obtain a full characterisation of KS contextuality in terms of graph-colourings in Sec.~\ref{sec: algebraic vs graph-theoretic approach}, which may be seen as complementary to this section. Together, these results provide a unified perspective between the algebraic and marginal approach to Kochen-Specker contextuality. They also highlight the fruitfulness of the concept of \emph{context connection}, and thus of the framework laid out here and in Ref.~\cite{Frembs2024a}.\\

\textbf{Finitely generated observable algebras and minimal generating sets.} The main motivation for the generalisation from quantum to more general observable algebras is that in experimental setups one is usually concerned with (finite-size) subsets of measurements only.\footnote{Importantly, this comment applies to various architectures in quantum computation and quantum error correction which utilise finite gate sets and measurements.} In other words, we will mainly be interested in observable subalgebras $\cO\subset\LHsa$ in quantum theory. The various formalisms in the marginal approach usually restrict to finite measurement scenarios by default (see App.~\ref{app: measurement scenarios vs OAs} for details). In order to compare the characterisation of Kochen-Specker noncontextualy in Thm.~\ref{thm: CNC for OAs} to these, we will henceforth also restrict to the ``finiteness'' of experimental setups.\footnote{Most proofs of the Kochen-Specker theorem \cite{KochenSpecker1967} (and related formalisations of contextuality) restrict to finitely many observables or events. Nevertheless, proofs involving infinitely many observables do exist \cite{Bell1966,Doering2004}, and are based on continuity arguments following Gleason's theorem \cite{Gleason1975}. We leave the interesting problem of adding topology to the study of Kochen-Specker contextuality for future work.}

Of course, already a single context, that is, a single commutative algebra $C\subset\cO$ contains infinitely many observables. However, the observables in a context $C$ are not independent, but are functionally related. More precisely, recall from Eq.~(\ref{eq: generating observables}) that every context is generated by a single (although generally not unique) observable with spectral decomposition $O = \sum_{i=1}^{|C|} \sigma_ip_i$ for $\sigma_i\in\spec(O)$, $p_i\in\PC$ and $|C|=|\cP_\mathrm{min}(C)|=|\spec(O)|$.

Recall that we interpret the relations between a generating observable $O$ of $C$, and other observables $\tO=g(O)\in C$ for $g:\R\ra\R$ as part of the side-processing of an experimental arrangement, e.g. a quantum computer. In this view, an experimental setup formally corresponds with infinitely many observables, yet with finitely many generating observables. In other words, experimental setups correspond with observable algebras whose partial order of contexts is finite.

\begin{definition}\label{def: finitely generated OAs}
    An observable algebra $\cO$ is called \emph{finitely generated (fin. gen.)} if the set of commutative algebras in $\CO$ is finite, equivalently if $\COm$ is a finite set.

    Similarly, $\cO$ is called countably infinitely generated if $\COm$ is a countable set.
\end{definition}

At first sight, by removing part of the algebraic structure in commutative algebras $C\in\CO$, measurement scenarios in the sheaf-theoretic approach (see Def.~\ref{def: measurement scenario}) and orthogonality graphs in the graph-theoretic approach (see Def.~\ref{def: orthogonality graph}) achieve a significant reduction in the representation of ``contextuality scenarios" \cite{FritzChaves2013}. Yet, as it turns out this seeming reduction generally results in a redundant representation from the perspective of Kochen-Specker contextuality in the first case, and leads to different notions of contextuality in the second (see Sec.~\ref{sec: algebraic vs graph-theoretic approach}). In a nutshell, the measurement labels in measurement scenarios are generally not minimal as generating sets for $\cO$ (see App.~\ref{app: measurement scenarios vs OAs}), while orthogonality graphs lack a notion of identity element which results in states on orthogonality graphs being incompatible with subsystem embedding (see App.~\ref{app: orthogonality graphs vs OAs}).\footnote{To a lesser extent this is also true for the hypergraph approach \cite{ChavesFritz2012,AcinFritzLeverrierSainz2015}. However, for single systems hypergraphs are equivalent to event algebras, and thus equivalent to $\cO$ by Lm.~\ref{lm: embeddings vs orthomorphisms vs order-preserving maps}.} In any case, up to a choice of generating set the correspondence between observable algebras and measurement scenarios is one-to-one; while the correspondence between observable algebras and orthogonality graphs is one-to-one up to adding an identity element. For the precise statements, we refer to App.~\ref{app: unified perspective}.\footnote{We reiterate that in this work we only consider the case of ideal observables.} Without loss of generality, and emphasising this unifying trait of our formalism, we will therefore keep the subsequent presentation in terms of finitely generated (and finite-dimensional) observable algebras.

\subsection{Basic elements of the marginal approach}\label{sec: basics of marginal approach}
 
The key tenet of the marginal approach is to study contextuality in terms of states.

\begin{definition}\label{def: state}
    Let $\cO$ be an observable algebra with partial order of contexts $\CO$. A \emph{state $\gamma$ on $\cO$} is a collection of probability distributions $\gamma = (\gamma_C)_{C\in\CO}$ such that
    \begin{align}\label{eq: no-disturbance}
        \gamma_\tC = \gamma_C|_\tC \quad\quad \forall \tC,C\in\CO \mathrm{\ with\ } \tC\subset C\; ,
    \end{align}
    where $\gamma_C|_\tC$ denotes the marginalised distribution with respect to $\tC$. The set of all states on $\cO$ is denoted by $\Gamma[\CO]$.\footnote{Equivalently, $\Gamma[\CO]$ is the set of global sections of the probabilistic presheaf over $\CO$. Here, we will use an abbreviated version of the notation in Ref.~\cite{DoeringFrembs2019a,FrembsDoering2023}, leaving the choice of (probabilistic) presheaf over $\CO$ implicit. More generally, the topos- and sheaf-theoretic approaches to contextuality study the existence of global sections of other presheaves \cite{IshamButterfieldI,Doering2008,DoeringFrembs2019a}. Amongst other things, this results in a fine-grained notion and corresponding hierarchy of contextuality as shown in Ref.~\cite{AbramskyBrandenburger2011}.}
\end{definition}

The partial order of contexts $\CO$ thus encodes marginalisation constraints on states. The terminology in Def.~\ref{def: state} is standard in the topos approach to quantum theory \cite{IshamDoeringI,IshamDoeringII,IshamDoeringIII,IshamDoeringIV,AbramskyBrandenburger2011,DoeringFrembs2019a}, the hypergraph approach \cite{ChavesFritz2012,AcinFritzLeverrierSainz2015,AmaralCunha2018} as well as in the literature on quantum logic which studies properties of orthomodular lattices \cite{BirkhoffVonNeumann1936,Putnam1979,Kalmbach1983}.\footnote{A partial Boolean algebra $\cP$ is an orthomodular lattice \cite{Kalmbach1983} if and only if it is \emph{transitive}, that is, if and only if for all $a,b,c\in\cP$, $a<b$ and $b<c$ implies $a<c$ \cite{Gudder1972}.} In contrast, when contextuality is viewed as a generalisation of nonlocality, a state is often referred to as a \emph{correlation} \cite{PopescuRohrlich1994,Popescu2014}; in the sheaf-theoretic approach, a state is also called an \emph{empirical model} \cite{AbramskyBrandenburger2011}. Similarly, the noncontextuality condition in Eq.~(\ref{eq: no-disturbance}) is referred to by various names, including \emph{no-disturbance} \cite{RamanathanEtAl2012,FrembsDoering2022a,Frembs2022a} and the \emph{sheaf condition} \cite{AbramskyBrandenburger2011}.

\begin{definition}\label{def: marginal problem}
    Let $\cO$ be a (finitely generated) observable algebra with partial order of contexts $\CO$ and event algebra $\PO$. A state $\gamma = (\gamma_C)_{C\in\CO}\in\Gamma[\CO]$ is said to admit a \emph{solution to the marginal problem over $\CO$} if there exists an orthomorphism $\phi_\gamma:\PO\ra\cB(\Lambda)$ for some measure space $\Lambda=\Lambda_{\phi_\gamma}$ with Boolean $\sigma$-algebra $\cB(\Lambda)$ and probability measure $\mu_\gamma\in L^{\geq 0}_1(\Lambda)$ such that for all $C\in\CO$ and $p\in\PC$,
    \begin{align}\label{eq: marginal problem}
        \gamma_C(p)
        = \int_\Lambda d\lambda\ \mu_\gamma(\lambda) \phi_\gamma(p)(\lambda)
        = \int_\Lambda d\lambda\ \mu_\gamma(\lambda) \chi(\lambda\mid p)
        = \mu_\gamma(\Lambda_p)\; ,
    \end{align}
    where $\phi_\gamma(p)=\chi(\cdot\mid p)$ denotes the indicator function for the (measurable) subset $\Lambda\supset\Lambda_p\in\cB(\Lambda)$ representing $p\in\PC$ under $\phi_\gamma$, defined by $\chi(\lambda\mid p) = 1$ if and only if $\lambda\in\Lambda_p$.
\end{definition}

The representation of the marginal problem in terms of the second equality in Eq.~(\ref{eq: marginal problem}) is common in the literature on the noncontextuality polytope \cite{Pitowsky1991,KleinmannEtAl2012}. There, $\Lambda$ is taken to represent a space of ``hidden variables'', in analogy with (the assumption of local hidden variables in) Bell's theorem \cite{Bell1964}. Often $\chi(\lambda\mid p)$ is only required to be a probability distribution over $\Lambda$ for every $p\in\PO$. Yet, by Fine's theorem \cite{Fine1982a,SuppesZanotti1981,AbramskyBrandenburger2011} if $\gamma\in\Gamma[\CO]$ admits a solution to the marginal problem with $\chi$ a probability distribution then it also admits a solution to the marginal problem with $\chi$ an indicator function.\footnote{For unsharp measurements, Eq.~(\ref{eq: marginal problem}) needs to be modified, see e.g. Ref.~\cite{Spekkens2005,YuOh2013,KunjwalSpekkens2015,KunjwalSpekkens2018}.}

Making this simplification lets us rephrase the definition of solutions to the marginal problem in terms of the existence of orthomorphisms $\phi_\gamma:\PO\ra\cB(\Lambda)$. (Possibly up to a $\sigma$-ideal of negligible sets)\footnote{For finitely generated observable algebras $\Lambda$ can be taken to be a finite set. Moreover, in this case an optimal representation in terms of valuations is constructed in Lm.~\ref{lm: optimal embedding} below.} $\phi_\gamma$ satisfies $\supp(\chi(\cdot\mid\One))=\Lambda$, $\supp(\chi(\cdot\mid 0))=\emptyset$ and whenever $p,p'\in\PC$ for $C\in\CO$ then $\supp(\chi(\cdot\mid pp'))=\supp(\chi(\cdot\mid p))\cap\supp(\chi(\cdot\mid p'))$ and $\supp(\chi(\cdot\mid p+p')) = \supp(\chi(\cdot\mid p))\cup\supp(\chi(\cdot\mid p'))$. Since these constraints all pertain to individual contexts, they can be justified operationally, for instance, by demanding that the statistics of repeated measurements of compatible observables be reproduced, as is usually the case in experimental tests of contextuality \cite{Cabello2008b,GuehneEtAl2010,Winter2014}.

Finally, note that phrasing Def.~\ref{def: marginal problem} in terms of orthomorphisms automatically ensures that the indicator functions are independent of contexts, that is, $\chi(\lambda\mid p)=\chi(\lambda\mid p,C)$ for all contexts $C\in\CO$ with $p\in\PC$. Eq.~(\ref{eq: marginal problem}) thus incorporates the assumption that the representation in terms of hidden variables $\Lambda$ is independent of contexts.

We emphasise that $\phi_\gamma$ is not required to be faithful. This will become important in the comparison between different notions of contextuality in the marginal approach below.

\subsection{Contextuality in the marginal approach}\label{sec: marginal approach to KSC}

In this section we compare the marginal approach to contextuality with the algebraic one, as defined in Def.~\ref{def: KSNC}. We will first highlight that the respective notions of contextuality differ, that is, that they capture different aspects of ``(non-)classicality''. Specifically, we find that the more commonly adopted notion concerns the existence of non-classical (``contextual") correlations which is generically stronger than that defined by the classical embedding problem underlying Kochen-Specker noncontextuality. In a second step, we will then provide a reformulation of the latter in terms of a refined marginal problem.

\subsubsection{Two distinct notions of classicality}\label{sec: inequivalent notions of contextuality}

The algebraic approach to Kochen-Specker contextuality is concerned with the existence of a classical (or other type of) embedding for an observable algebra. In contrast, the marginal approach to contextuality studies the subset of states on observable algebras that is compatible with classical (or other types of) embeddings.\\

\textbf{Classical correlations.} Def.~\ref{def: marginal problem} characterises a subset of \emph{classical states} $\Gamma_\mathrm{cl}[\CO] \subset \Gamma[\CO]$ on $\cO$, as those states that are compatible with, that is, arise as states under some orthomorphism $\phi:\PO\ra\cB(\Lambda_\phi)$ for some measurable space $\Lambda_\phi$ with Boolean $\sigma$-algebra $\cB(\Lambda_\phi)$. More precisely, $\Gamma_\mathrm{cl}[\CO]=\bigcup_{\phi} L^{\geq 0}_1(\Lambda_\phi)|_{\phi(\PO)}$, where $|_{\phi(\PO)}$ denotes marginalisation as in Def.~\ref{def: state}. Similarly, we can define $\Gamma_\mathrm{qm}[\CO]\subset\Gamma[\CO]$ as those states on $\cO$ that are compatible with orthomorphisms into a quantum system. Together, this defines the following hierarchy of states (see also Sec.~\ref{sec: basics of graph-theoretic approach}):
\begin{align}\label{eq: cl-qu-nd states}
    \Gamma_\mathrm{cl}[\CO]\subset \Gamma_\mathrm{qm}[\CO]\subset \Gamma[\CO] =: \Gamma_\mathrm{nd}[\CO]\; , 
\end{align}
where the latter state space contains all non-disturbing states.\footnote{In analogy with the definition of classical and quantum states, $\Gamma_\mathrm{nd}[\CO]$ can be defined as the set of states compatible with orthomorphisms of $\PO$ into general observable algebras. Since the identity map defines such an orthomorphism we have $\Gamma[\CO]=\Gamma_\mathrm{nd}[\CO]$. Similarly, by restricting to observable algebras with additional structure, for instance, within the setting of general probabilistic theories \cite{Plavala2023,MuellerGarner2023,CataniGalleyGonda2024}, one may further refine the hierarchy of states in Eq.~(\ref{eq: cl-qu-nd states}).}

In general, the inclusions in Eq.~(\ref{eq: cl-qu-nd states}) are strict and give rise to the well-known and much studied problems of characterising the respective state spaces \cite{Vorob1962,Pitowsky1991,KleinmannEtAl2012,AcinFritzLeverrierSainz2015,XuCabello2019}. Notably, all sets in Eq.~(\ref{eq: cl-qu-nd states}) are convex, hence, can be described in terms of linear inequalities \cite{BadziagBengtssonCabelloPitowsky2009,KleinmannEtAl2012}. What is more, $\Gamma_\mathrm{cl}[\CO]$ is a polytope \cite{Pitowsky1991}, hence, it is characterised by its finitely many facet-defining \emph{noncontextuality inequalities}; in the Bell-local case, these are known as \emph{Bell inequalities}. In turn, the violation of such inequalities, can thus be used to rule out a classical state space model for $\cO$.\footnote{Maximal violations of noncontextuality inequalities by quantum mechanics generalise (``Tsirelson"-) bounds on Bell inequalities \cite{Tsirelson1980} from nonlocal to contextual correlations, see e.g. Ref.~\cite{AraujoEtAl2013}.} Bell's theorem \cite{Bell1964} and the (loophole-free) experimental violation of Bell inequalities \cite{AspectEtAl1981,ZeilingerEtAl2015,ShalmEtAl2015} irrefutably show that quantum physics is non-local: its state space is not a classical state space that is compatible with composition (in terms of products of measure spaces \cite{DoeringFrembs2019a}). A similar approach has been taken for contextuality \cite{KlyachkoEtAl2008,Cabello2008b,YuOh2012,KleinmannEtAl2012,AraujoEtAl2013,CabelloSeveriniWinter2014,XuCabello2019}; in analogy with the (Bell) local case, the following notion of classicality is therefore often adopted.

\begin{definition}\label{def: fully classically correlated}
    Let $\cO$ be an observable algebra with partial order of contexts $\CO$. Then $\cO$ is said to be \emph{(fully) classically correlated} if every state $\gamma\in\Gamma[\CO]$ on $\cO$ admits a solution to the marginal problem defined by $\CO$, that is, if $\Gamma[\CO]=\Gamma_\mathrm{cl}[\CO]$.
\end{definition}

If an observable algebra is classically correlated, then it cannot violate any noncontextuality inequality, and vice versa. Conversely, if an observable algebra is not fully classically correlated, then there exists some state that cannot be explained in terms a classical model, hence, is contextual in the sense of the marginal approach. If such correlations can be observed, they amount to an ``experimental test of contextuality" \cite{BartosikEtAl2009,KirchmairEtAl2009,AmselemEtAl2009,GuehneEtAl2010,MoussaEtAl2010,LapkiewiczEtAl2011,Winter2014}. It is with respect to Def.~\ref{def: fully classically correlated}, that the often quoted relation ``nonlocality implies contextuality" is to be understood.\\

\textbf{Classical correlations vs classical embeddings.} In turn, the above implication is not true when contextuality is understood in the sense of Eq.~(\ref{eq: KSNC}). Indeed, the existence of a contextual state does \emph{not} refute the existence of a classical embedding.

Since this goes against the commonly held dictum stated above, we take a moment to provide some intuition by means of a well-known example. It is straightforward to define a classical embedding of the observable algebra $\cO_\mathrm{CHSH}$ defined by the CHSH scenario \cite{CHSH1969} (for details, see Ex.~\ref{ex: classical embedding of CHSH} below). Yet, clearly this does not imply that every state on $\cO_\mathrm{CHSH}$ can be represented in terms of a classical state, that is, by a probability distribution on the state space $\Lambda$ under the embedding. In particular, (most) entangled states are not classical (see Fig.~\ref{fig: classical embedding for CHSH} (b)). We remark that in the Bell-local case of composite systems at spacelike separation, the violation of Bell inequalities requires entanglement, hence, is a state-dependent notion; separable states do not violate Bell inequalities. More generally, whenever an observable algebra admits a classical embedding, the violation of a noncontextuality inequality is necessarily state-dependent. In contrast, the violation of noncontextuality inequalities can also be state-independent (see Tab.~\ref{tab: state-(in)dependence} and Sec.~\ref{sec: SI-C sets and graphs}).

\begin{proposition}\label{prop: nonclassical correlations in KS noncontextual OAs}
    There exist observable algebras $\cO$ which are Kochen-Specker noncontextual, yet which are not fully classically correlated, hence, which exhibit non-classical (e.g. quantum) correlations, that is, for which $\Gamma[\CO]\neq\Gamma_\mathrm{cl}[\CO]$.\footnote{Analogously, there exist observable algebras that admit quantum embeddings while also exhibiting super-quantum correlations \cite{MuellerGarner2023}. This leads to the fascinating problem of characterising quantum correlations \cite{PopescuRohrlich1994,PawloskiEtAl2009,BarnumEtAl2010,Henson2012,FritzEtAl2013,Cabello2013,Popescu2014,FrembsDoering2022a} and the related research programme on reconstructions of quantum theory, e.g. from information-theoretic principles \cite{Rovelli1996,Hardy2001,Hardy2001a,Grinbaum2007,ChiribellaDArianoPerinotti2011,Hoehn2017,HoehnWever2017,Mueller2021}.}
\end{proposition} 

Consequently, nonlocality in the sense of violations of Bell inequalities does \emph{not} imply KS contextuality as defined by the embedding problem in Eq.~(\ref{eq: KSNC})!

What about the reverse implication? On a formal level this turns out to be false, too (see Lm.~\ref{lm: single state}). However, it holds under the mild assumption that states on $\cO$ are separating (for details, see App.~\ref{app: OAs with separating states}). For separating observable algebras, Def.~\ref{def: fully classically correlated} is strictly stronger than Def.~\ref{def: KSNC}. In particular, finitely generated classically correlated observable subalgebras $\cO\subset\LHsa$ in quantum theory are KS noncontextual.

In summary, the notions of classicality in Def.~\ref{def: KSNC} and Def.~\ref{def: fully classically correlated} differ, hence, define two different problems with different interpretations and consequences. While KS noncontextuality in Eq.~(\ref{eq: KSNC}) clearly implies the existence of \emph{some} classical states, Def.~\ref{def: fully classically correlated} demands that \emph{every} state on $\cO$ is represented by a classical state under the embedding $\epsilon:\cO\ra L_\infty(\Lambda)$.
Note that both constraints appear in Ref.~\cite{KochenSpecker1967},\footnote{Eq.~(\ref{eq: KSNC}) is Eq.~(4) in Ref.~\cite{KochenSpecker1967}, while Def.~\ref{def: fully classically correlated} generalises Eq.~(2) in Ref.~\cite{KochenSpecker1967} from quantum to general observable algebras.} yet the KS theorem$^{\ref{fn: "KS-theorem"}}$ only addresses Eq.~(\ref{eq: KSNC}), and proves more than the mere existence of contextual correlations. For a concise comparison between the two approaches, see Tab.~\ref{tab: comparison}.

\begin{table}[hbt!]
    \centering
    \begin{tabular}{c|cc}
        \toprule[0.05cm]
        & notion of ``classicality" & characterisation \\
        \midrule\\[-.4cm]
        algebraic & KS noncontextuality (Def.~\ref{def: KSNC}) & $\exists$ flat context connection $\fl$ on $\CO$ \\
        $\CO$ & $\forall g:\R\ra\R:\ \epsilon(g(O))=g(\epsilon(O))$ & $\circ_{i=0}^{n-1} l_{C_{i+1}C_i} = \mathrm{id}\quad \forall (C_0,\cdots,C_{n-1})$ \\[0.2cm] & \rotatebox{90}{\color{red}{$\mathbf{\neq}$}} & \rotatebox{90}{\color{red}{$\mathbf{\neq}$}} \\[0.2cm]
        marginal & classically correlated (Def.~\ref{def: fully classically correlated}) & determine all facet-defining \\
        $\Gamma[\CO]$ & $\Gamma[\CO]=\Gamma_\mathrm{cl}[\CO]$ & NC inequalities \cite{Pitowsky1991,Cabello2008b,BadziagBengtssonCabelloPitowsky2009,KleinmannEtAl2012} \\[0.1cm]
        \bottomrule[0.05cm]
    \end{tabular}
    \caption{Comparison of approaches to noncontextuality: Kochen-Specker (KS) noncontextuality in Eq.~(\ref{eq: KSNC}) is algebraic, that is, it is defined in terms of functional relations between observables $\cO$ (first row), and characterised by the triviality constraints in Eq.~(\ref{eq: CNC}); the marginal approach studies states $(\gamma_C)_{C\in\CO}\in\Gamma[\CO]$ on $\cO$, which satisfy no-disturbance by definition and are ``classical" if they can be represented as in Eq.~(\ref{eq: marginal problem}) (second row). The respective notions of classicality generally differ (see Prop.~\ref{prop: nonclassical correlations in KS noncontextual OAs} and Lm.~\ref{lm: single state}); nevertheless, in Thm.~\ref{thm: KSNC = marginal KSNC} we translate KS noncontextuality into the marginal approach.}
    \label{tab: comparison}
\end{table}

While KS noncontextuality differs from the notion of classicality in Def.~\ref{def: fully classically correlated}, in the next section we show that it can be given a reformulation as a marginal problem.

\subsubsection{Kochen-Specker contextuality in the marginal approach}\label{sec: characterising KSNC using marginal approach}

The notion of classicality in the marginal approach is concerned with characterising the state spaces in Eq.~(\ref{eq: cl-qu-nd states}). In contrast, Kochen-Specker noncontextuality is concerned with the underlying classical embedding problem defined in Def.~\ref{def: KSNC}. In this section, we will obtain a reformulation of the latter within the marginal approach. It is instructive to first consider the notion of contextuality obtained by exchanging the universal quantifier in Def.~\ref{def: fully classically correlated} with an existential one.

\begin{definition}\label{def: existence of classical correlations}
    Let $\cO$ be an observable algebra with partial order of contexts $\CO$. Then $\cO$ \emph{exhibits classical correlations} if there exists a state $\gamma\in\Gamma[\CO]$ on $\cO$ that admits a solution to the marginal problem defined by $\CO$, that is, if $\emptyset\neq\Gamma_\mathrm{cl}[\CO]\subset\Gamma[\CO]$.
\end{definition}

In particular, the converse of Def.~\ref{def: existence of classical correlations} implies that every state on $\cO$ violates some noncontextuality inequality.\footnote{We will address the more subtle distinction between this notion and the one where all states violate the same noncontextuality inequality in Sec.~\ref{sec: SI-C sets and graphs}.} Clearly, if $\cO$ is KS noncontextual, then it exhibits classical correlations. Moreover, Def.~\ref{def: existence of classical correlations} is weaker than Def.~\ref{def: fully classically correlated} which (under the restriction separating observable algebras, see App.~\ref{app: OAs with separating states}) is stronger than KS noncontextuality. Indeed, Def.~\ref{def: fully classically correlated} proves too weak to imply KS noncontextuality.

\begin{lemma}\label{lm: solutions to marginal problem do not imply KSNC}
    There exist observable algebras $\cO$ which are Kochen-Specker contextual, yet which exhibit classical correlations (see Def.~\ref{def: existence of classical correlations}), that is, for which $\Gamma_\mathrm{cl}[\CO]\neq\emptyset$,
\end{lemma}

\begin{proof}
    Let $\cO_1,\cO_2$ be two observable algebras such that $\cO_1$ is KS contextual and $\cO_2$ is KS noncontextual. Then their product $\cO=\cO_1\times\cO_2$ is also KS contextual.\footnote{Similarly, if both $\cO_1$ and $\cO_2$ are KS noncontextual, their product is KS noncontextual, too. For extensions of the notion of products of observable and partial Boolean algebras, see Ref.~\cite{DoeringFrembs2019a,AbramskyBarbosa2020}.} However, since $\cO_2$ admits a classical embedding $\epsilon:\cO_2\ra L_\infty(\Lambda)$ by assumption, there exist states $\gamma\in L_1(\Lambda)$ which define (non-faithful) states on $\cO$, hence, $\Gamma_\mathrm{cl}[\CO]\neq\emptyset$.
\end{proof}

Consequently, the notions of contextuality in Def.~\ref{def: KSNC} and Def.~\ref{def: existence of classical correlations} also differ. Mathematically, the reason is that the orthomorphism, defined by a state $\gamma\in\Gamma_\mathrm{cl}[\CO]$ as by Def.~\ref{def: marginal problem}, is not required to be faithful, hence, does not induce an embedding by Lm.~\ref{lm: embeddings vs orthomorphisms vs order-preserving maps}. Yet, this immediately suggests how to strengthen the existence of classical states in Def.~\ref{def: existence of classical correlations}: to the existence of (sets of) classical states that separate events in $\PO$.

\begin{definition}\label{def: marginal noncontextuality}
    Let $\cO$ be an observable algebra with partial order of contexts $\CO$. Then $\cO$ is called \emph{marginally Kochen-Specker (KS) noncontextual} if for all projections $p,p'\in\PO$, $p\neq p'$, there exists a classical state $\gamma\in\Gamma_\mathrm{cl}[\CO]\subset\Gamma[\CO]$ such that
    $0\neq\gamma(p)\neq\gamma(p')\neq 0$ on $\cO$.
\end{definition}

We call the set of states in Def.~\ref{def: marginal noncontextuality} a \emph{separating set} (see also Def.~\ref{def: separating set}). For finitely generated observable algebras a separating set $\{\gamma_i\}_{i=1}^N$ is finite and its existence thus equivalent to the existence of a single state with that property. Indeed, by choosing appropriate mixing parameters $\omega_i\in[0,1]$ with $\sum_{i=1}^N \omega_i=1$, $\gamma = \sum_{i=1}^N \omega_i\gamma_i$ satisfies $0\neq\gamma(p)\neq\gamma(p')\neq 0$ for all $0\neq p,p'\in\PO$, and clearly $\gamma\in\Gamma_\mathrm{cl}[\CO]$.

As the terminology suggests, Def.~\ref{def: marginal noncontextuality} is equivalent to (the classical embedding problem defined by) Kochen-Specker noncontextuality in Def.~\ref{def: KSNC}, as we now prove.

\begin{theorem}\label{thm: KSNC = marginal KSNC}
    Let $\cO$ be a finitely generated observable algebra. Then $\cO$ is Kochen-Specker noncontextual if and only if it is marginally Kochen-Specker noncontextual.
\end{theorem}

\begin{proof}
    ``$\Rightarrow$": If $\cO$ is KS noncontextual, then there exists a classical embedding $\epsilon:\cO\ra L_\infty(\Lambda)$ for some measurable space $\Lambda$, which restricts to a faithful orthomorphism of event algebras $\varphi=\epsilon|_\PO:\PO\ra\cB(\Lambda)$. Let $\widetilde{\Lambda} = (\Lambda,\sigma(\epsilon))$ denote the measurable space with $\sigma$-algebra generated by $\mathrm{im}(\varphi(\PO))$. Every state $\mu\in L^{\geq 0}_1(\widetilde{\Lambda})\subset L^{\geq 0}_1(\Lambda)$ then defines a state $\gamma\in\Gamma_\mathrm{cl}[\CO]$ on $\cO$ via Eq.~(\ref{eq: marginal problem}). Finally, since $\epsilon$ is faithful by assumption, for every $p,p'\in\PO$ with $p\neq p'$ we can also find a state $\mu\in L_1(\widetilde{\Lambda})$ such that $0\neq\mu(\varphi(p))\neq\mu(\varphi(p'))\neq 0$, hence, $\cO$ is marginally KS noncontextual.

    ``$\Leftarrow$": We need to show the existence of a classical embedding $\epsilon: \cO \ra L_\infty(\Lambda)$. Since $\cO$ is marginally KS noncontextual, there exists a classical state $\gamma=(\gamma_C)_{C\in\CO}\in\Gamma_\mathrm{cl}[\CO]$, that is, $\gamma$ admits a solution $\mu_\gamma\in L_1(\Lambda)$ to the marginal problem defined by $\CO$ (see Def.~\ref{def: marginal problem}) for some measurable space $\Lambda$, that is, for all $C\in\CO$ and $p\in\PC$,
    \begin{align*}
        \gamma_C(p)
        = \int_\Lambda d\lambda\ \mu_\gamma(\lambda)\phi_\gamma(p)(\lambda)
        = \int_\Lambda d\lambda\ \mu_\gamma(\lambda) \chi(\lambda\mid p)
        = \int_{\Lambda_p} d\lambda\ \mu_\rho(\lambda) =\mu(\Lambda_p)\; ,
    \end{align*}
    where $\phi_\gamma:\PO\ra\cB(\Lambda)$ is an orthomorphism, which represents every event $p\in\PO$ by (the indicator function with respect to) the measurable subset $\Lambda_p\subset\Lambda$. Furthermore, by assumption we may assume that $0\neq\gamma(p)\neq\gamma(p')\neq 0$ for all $p,p'\in\PO$ with $p\neq p'$. Consequently, $\varphi(p)=\chi(\lambda\mid p)\neq\chi(\lambda\mid p')=\varphi(p')$, that is, $\varphi:\PO\ra\cB(\Lambda)$ is faithful, hence, lifts to an embedding of observable algebras $\epsilon:\cO\ra L_\infty(\Lambda)$ by Lm.~\ref{lm: embeddings vs orthomorphisms vs order-preserving maps}.
\end{proof}

By Thm.~\ref{thm: KSNC = marginal KSNC}, Def.~\ref{def: marginal noncontextuality} translates Kochen-Specker noncontextuality into the language of the marginal approach, and thus completes the respective characterisation in Tab.~\ref{tab: summary}. Together with the preceding lemmata, it clarifies the subtle differences between similar notions in Def.~\ref{def: fully classically correlated} and Def.~\ref{def: existence of classical correlations}, in particular, with regards to state-(in)dependence (see Tab.~\ref{tab: state-(in)dependence}). We will further discuss the notion of state-independent contextuality (SI-C) in Sec.~\ref{sec: SI-C sets and graphs}, specifically the distinction between SI-C sets and SI-C graphs.

\begin{table}[htb!]
    \centering
    \renewcommand{\arraystretch}{1.75}
    \begin{tabular}{c|c}
        \toprule[0.05cm]
        $\cO$ (finitely generated, separating) & state-independent \\
        \midrule
        fully classically correlated (Def.~\ref{def: fully classically correlated}) & \checkmark \\
        (Prop.~\ref{prop: non-empty fully classical implies KS noncontextual}) $\Downarrow$ $\not\Uparrow$ (Prop.~\ref{prop: nonclassical correlations in KS noncontextual OAs}) & \\
        \hspace{.4cm} KSNC (Def.~\ref{def: KSNC}) $\stackrel{\mathrm{Thm}.~\ref{thm: KSNC = marginal KSNC}}{\Longleftrightarrow}$ marginal KSNC (Def.~\ref{def: marginal noncontextuality}) \hspace{.4cm} & \xmark \\
        \hspace{1.25cm} $\Downarrow$ $\not\Uparrow$ (Lm.~\ref{lm: solutions to marginal problem do not imply KSNC}) & \\
        exhibits classical correlations (Def.~\ref{def: existence of classical correlations}) & \xmark \\[0.1cm]
        \midrule
        violation of some NC inequality ($\neg$ Def.~\ref{def: fully classically correlated}) & \xmark \\
        (Prop.~\ref{prop: non-empty fully classical implies KS noncontextual}) $\not\Downarrow$ $\Uparrow$ (Prop.~\ref{prop: nonclassical correlations in KS noncontextual OAs}) & \\
        KSC ($\neg$ Def.~\ref{def: KSNC}) $\stackrel{\mathrm{Thm}.~\ref{thm: KSNC = marginal KSNC}}{\Longleftrightarrow}$ marginal KSC ($\neg$ Def.~\ref{def: marginal noncontextuality}) & \xmark \\
        \hspace{1.25cm} 
        $\not\Downarrow$ $\Uparrow$ (Lm.~\ref{lm: solutions to marginal problem do not imply KSNC}) & \\
        every state violates some of NC inequality ($\neg$ Def.~\ref{def: existence of classical correlations}) & \checkmark \\[0.1cm]
        \bottomrule[0.05cm]
    \end{tabular}
    \caption{Relation between KS noncontextuality (top) and KS contextuality (bottom) with closely related notions of contextuality in the marginal approach \cite{BudroniEtAl2022}. While both the nonexistence of violations of noncontextuality inequalities (Def.~\ref{def: fully classically correlated}) as well as the state-independent violation of some noncontextuality inequality (converse of Def.~\ref{def: existence of classical correlations}) are inherently state-independent notions, the reformulation of both KS contextuality and KS noncontextuality in the marginal approach (see Def.~\ref{def: marginal noncontextuality}) are state-dependent.}
    \label{tab: state-(in)dependence}
\end{table}

Beyond its conceptual significance, Thm.~\ref{thm: KSNC = marginal KSNC} opens up the possibility to relate the algebraic and marginal approach to contextuality. Marginal noncontextuality is analysed in terms of noncontextuality inequalities (see Tab.~\ref{tab: comparison}). As a consequence of Thm.~\ref{thm: KSNC = marginal KSNC}, these need to encode the triviality constraints in Thm.~\ref{thm: CNC for OAs}. A careful analysis using context connections and context cycles thus seems a promising avenue to better characterise the noncontextuality polytope, which in general remains an open problem \cite{BudroniEtAl2022}.

\subsection{Acyclic observable algebras and Vorob'ev theorem}\label{sec: (a)cycles}

In this section, we further expand on the comparison between the marginal and algebraic approach to Kochen-Specker contextuality. We begin in Sec.~\ref{sec: truncated OAs} by proving that for \emph{acyclic observable algebras} (see Def.~\ref{def: acyclic context category}) the triviality constraints in Thm.~\ref{thm: CNC for OAs} can always be satisfied. In Sec.~\ref{sec: Vorob'ev theorem}, we further relate this property with Vorob'ev's theorem \cite{Vorob1962} (see also Ref.~\cite{XuCabello2019}), which shows that acyclicity fully characterises classically correlated observable algebras (see Def.~\ref{def: fully classically correlated}). In Sec.~\ref{sec: single cycles}, we then turn to KS contextuality, specifically we make explicit that KS noncontextuality does not also imply acyclicity by proving that a single context cycle is never KS contextual - in stark contrast with the characterisation of the $n$-cycle contextuality scenario in the marginal approach \cite{AraujoEtAl2013}.

\subsubsection{Trivial context cycles and Meyer-Clifton-Kent models}\label{sec: truncated OAs}

In this section, we will identify a property of (finitely generated) observable algebras and their partial orders of contexts under which the constraints in Thm.~\ref{thm: CNC for OAs} are always satisfiable, hence, which implies KS noncontextuality. In doing so, we will obtain a new explanation for the KS noncontextuality of the models constructed in Ref.~\cite{Meyer1999,Kent1999,CliftonKent2000}.

Thm.~\ref{thm: CNC for OAs} recasts KS noncontextuality as constraints on context connections along context cycles. It is thus instructive to study individual cycles first.

\begin{definition}\label{def: n-cycle scenario}
    A maximal observable algebra $\cO$ with $\dim(\One)=d$ and maximal contexts $\COm = \{C_0,\cdots,C_{n-1}\}$ such that $C_\One \neq C_{i\cap j} := C_i \cap C_j$ if and only if $j=i-1,i,i+1$ is called an \emph{$n$-cycle contextuality scenario in dimension $d$}, and denoted by $\cO=\cO(n,d)$.\footnote{Note that this implies $C_{ij}\cap C_{kl} = C_\One$ unless $i=k$, $j=l$ or $i=l$, $j=k$.}
\end{definition}

The terminology is clear: $\cO(n,d)$ contains a single non-trivial context cycle. Def.~\ref{def: n-cycle scenario} has appeared before in Ref.~\cite{AraujoEtAl2013,FritzChaves2013,XuCabello2019};\footnote{Cyclic scenarios of different types have also been studied in various other approaches \cite{AbramskyEtAl2015,Caru2017,BeerOsborne2018,DzhafarovKujalaCervantes2020}.} here, we adapted it to our notation. Note also that we do not assume that subcontexts $C_{i\cap j}$ have isomorphic generating sets.\\

\textbf{Truncated observable algebras.} We start with the following easy observation.

\begin{lemma}\label{lm: trivial cycle}
    Let $\cO(n,d)$ be a $n$-cycle contextuality scenario in dimension $d$ with non-trivial context cycle $(C_0,\cdots,C_{n-1})$. If $C_{(i+1)i} = C_i \cap C_{i+1} = C_\One = \R\One$ for at least one $i \in \zz_n$, then $\cO(n,d)$ is Kochen-Specker noncontextual.
\end{lemma}

\begin{proof}
    After possibly relabeling the indices in the context cycle, we may assume that $C_{0(n-1)}=C_0 \cap C_{n-1}=C_\One$. Let $\fl$ be any context connection on $\cC(\cO(n,d))$ and define $\pi:=\prod_{i=0}^{n-1} l_{(i+1)i}\in\Sym(\cP_1(C_0))\cong S_d$. Then $\fl'$ with $l'_{(i+1)i} = l_{(i+1)i}$ for all $i \in \{0,\cdots,n-2\}$ and $l'_{0(n-1)} = \pi^{-1} \circ l_{0(n-1)}$ is also a context connection (it satisfies Eq.~(\ref{eq: noncontextual context connection}) in Def.~\ref{def: context connection}), for which Eq.~(\ref{eq: CNC}) is clearly satisfied, hence, $\cO(n,d)$ is KS noncontextual by Thm.~\ref{thm: CNC for OAs}.
\end{proof}

We call a context cycle $(C_0,\cdots,C_{n-1})$ \emph{trivial} if $C_{(i+1)i}=C_\One$ for some $0\leq i \leq n-1$. Lm.~\ref{lm: trivial cycle} suggests to remove trivial context cycles from $\cO$. We can do so by removing the least element $C_\One$ from $\CO$, that is, by removing (functions of) the trivial observable $\One$ from $\cO$. We denote by $\rcO=\cO/\{\lambda\One\mid\lambda\in\R\}$ the resulting \emph{truncated observable algebra} with \emph{truncated partial order of contexts} $\rCO=\CO/\{C_\One\}$.

While Lm.~\ref{lm: trivial cycle} shows that a single trivial context cycle is always KS noncontextual, it does a priori not apply to several trivial cycles: the reduction from $\cO$ to $\cO^\circ$ results in fewer constraints in Eq.~(\ref{eq: CNC}), in other words, the constraints between multiple trivial context cycles may be incompatible with Eq.~(\ref{eq: CNC}). However, the following theorem asserts that for finitely generated observable algebras this is never the case.

\begin{theorem}\label{thm: KSNC of truncated OAs}
    Let $\cO$ be a countably infinitely generated observable algebra with partial order of contexts $\CO$. Then $\cO$ is Kochen-Specker noncontextual if and only if $\rcO$ is.
\end{theorem}

\begin{proof}
    By Lm.~\ref{lm: KSNC inherits from maximality}, it is sufficient to establish the claim for $\cO$ and $\rcO$ maximal.
    
    ``$\Rightarrow$'': Since $\rCO \subset \CO$, this direction is obvious.

    ``$\Leftarrow$'': Consider the (finite) graph $G_1 = (V,E_1)$, whose set of vertices is given by $V = \COm$ and whose set of edges $E_1 \subset V \times V$ is given by $E_1 = \{(C',C) \mid C \cap C' \neq C_\One\}$. $G_1$ thus describes the context overlaps in $\rCOm$.\footnote{Evidently, this graph differs from and should not be confused with the orthogonality graph associated to $\cO$ as defined in Def.~\ref{def: associated orthogonality graph} below.} Now, assume that $\rcO$ is KS noncontextual, that is, by Thm.~\ref{thm: CNC for OAs} there exists a flat context connection $\fl_1$ on $\rCO$ satisfying the triviality constraints in Eq.~(\ref{eq: CNC}). Next, we successively add edges, forming new graphs $G_m=(V,E_m)\mapsto G_{m+1}=(V,E_{m+1})$, until we obtain the fully-connected graph $G_N=(V,E_N)$ with $E_N=V\times V$, which corresponds to the partial order of contexts $\COm$. Note that since vertices in $G_m$ correspond to maximal contexts in $\CO$ and edges to context overlaps, context cycles are represented by cycles in $G_m$.

    Let $(C',C) \in E_{m+1}/E_m$ denote the edge added in the $m$-th inductive step. For any paths\footnote{A path $p$ in $G=(V,E)$ is any tuple $p=(v_1,\cdots,v_n)$ of vertices $v_i\in V$, $i=\{1,\cdots,n\}$ such that consecutive vertices are connected by an edge, that is, $(v_i,v_{i+1})\in E$ for all $i=\{1,\cdots,n-1\}$.} $p_{C_0}^C = (C_0,\cdots,C_k=C)$ and $p_{C'}^{C_0} = (C_{k+1}=C',\cdots,C_{n-1},C_0)$ in $G_m$, this adds a new context cycle $\alpha = p_{C'}^{C_0} \circ (C',C) \circ p_{C_0}^C := (C_0,\cdots,C_k=C,C_{k+1}=C',\cdots,C_{n-1})$ to $G_{m+1}$, and thus a new triviality constraint in Eq.~(\ref{eq: CNC}).\footnote{Here, $\circ$ denotes concatenation of paths with the same end and start point, and such that $(\cdots,C_{i-1},C_i,C_i,C_{i+1},\cdots)=(\cdots,C_{i-1},C_i,C_{i+1},\cdots)$.} Clearly, by setting $\fl({m+1})=\fl(m)\cup\fl_{C'C}$ with $\fl_{C'C}=(\prod_{i=k}^{n-1}\fl_{(i+1)i}(m))^{-1} (\prod_{i=0}^{k-1}\fl_{(i+1)i}(m))^{-1}$ we have $\circ_{i=0}^{n-1} l_{(i+1)i}(m+1) = \mathrm{id}$. Now, consider any other paths $q_{C_0}^C=(C_0,D_1,\cdots,D_{k-1},D_k=C)$ and $q_{C'}^{C_0}=(D_{k+1}=C',D_{k+2}\cdots,D_{n-1},C_0)$ in $G_m$ with the same end points as $p_{C_0}^C$ and $p_{C'}^{C_0}$, respectively. Then the context cycle $\beta = q_{C'}^{C_0} \circ (C',C) \circ q_{C_0}^C$ can be written as
    \begin{align*}
        \beta &= q_{C'}^{C_0} \circ (C',C) \circ q_{C_0}^C \\
        &= q_{C'}^{C_0} \circ (p_{C_0}^{C'} \circ p_{C'}^{C_0}) \circ (C',C) \circ (p_{C_0}^C \circ p_C^{C_0}) \circ q_{C_0}^C \\
        &= (q_{C'}^{C_0}\circ p_{C_0}^{C'}) \circ \alpha \circ (p_C^{C_0}\circ q_{C_0}^C)\; ,
    \end{align*}
    where e.g. $p_C^{C_0}=(p_{C_0}^C)^{-1}=(C_k=C,C_{k-1},\cdots,C_0)$ denotes the context path $p_{C_0}^C$ in reversed order such that $p_C^{C_0}\circ p_{C_0}^C=(C_0)$, and similarly for the other paths.
    
    Let $\Xi(G_m)$ denote the space of context connections restricted to context overlaps (between maximal contexts $\COm$) given by $E_m$. Then, $\Xi(G_1)=\Xi(\rCOm)$ and $\Xi(G_N)=\Xi(\COm)$. By the inductive hypothesis, there exists $\fl(m)\in\Xi(G_m)$ satisfying the triviality constraints in Eq.~(\ref{eq: CNC}), in particular, this applies to the context cycles $(q_{C'}^{C_0}\circ p_{C_0}^{C'})$ and $(p_C^{C_0}\circ q_{C_0}^C)$ which lie in $G_m$. Moreover, $\fl(m+1)$ satisfies Eq.~(\ref{eq: CNC}) for $\alpha$ (as shown above), and since $\fl(m+1)|_{G_m}=\fl(m)$ by construction, $\fl(m+1)$ thus satisfies the constraint in Eq.~(\ref{eq: CNC}) for $\beta$, and thus for all context cycles $\gamma'$ in $G_{m+1}$ (which contain $(C',C)$). Finally, by induction we thus find a context connection $\fl(N)\in\Xi(G_N)=\Xi(\COm)$ such that Eq.~(\ref{eq: CNC}) holds, hence, the result follows by Thm.~\ref{thm: CNC for OAs}.
\end{proof}

Building on the characterisation of the classical embedding problem (defined in Eq.~(\ref{eq: KSNC})) for observable algebras in terms of context connections on their partial order of contexts in Thm.~\ref{thm: CNC for OAs}, Thm.~\ref{thm: KSNC of truncated OAs} further simplifies the characterisation of Kochen-Specker noncontextuality for finitely generated observable algebras. (In Thm.~\ref{thm: optimal representation} in App.~\ref{app: reduced context categories} we prove additional reduction arguments of similar type, in particular, we will show that the representation in terms of measurement scenarios, compatibility graphs etc is generally redundant.) In what follows, we discuss two immediate consequences of Thm.~\ref{thm: KSNC of truncated OAs}.\\

\textbf{Acyclic observable algebras.} Thm.~\ref{thm: KSNC of truncated OAs} motivates the following definition.

\begin{definition}[acyclic observable algebras]\label{def: acyclic context category}
    Let $\cO$ be an observable algebra with partial order of contexts $\CO$. $\cO$ is called \emph{acyclic} if $\CO$ contains only trivial context cycles.
\end{definition}

Note that acyclicity of $\cO$ is a property of its partial order of contexts, and we thus also call $\CO$ \emph{acyclic} if $\rCO = \CO/\{C_\One\}$. Thm.~\ref{thm: KSNC of truncated OAs} immediately implies the following.

\begin{corollary}\label{cor: acyclic = KS}
    Let $\cO$ be a countably infinitely generated acyclic observable algebra. Then $\cO$ is Kochen-Specker noncontextual.
\end{corollary}

\begin{proof}
    If $\cO$ is acyclic, then there are no context cycles in $\cO^\circ$, hence, Eq.~(\ref{eq: CNC}) in Thm.~\ref{thm: CNC for OAs} poses no constraints and $\rCOm$ is trivially Kochen-Specker noncontextual. But then $\cO$ is also Kochen-Specker noncontextual by Thm.~\ref{thm: KSNC of truncated OAs}.
\end{proof}

The importance of cyclicity in proofs of the Kochen-Specker theorem has been recognised before. Indeed, many proofs explicitly exhibit a ``cyclic nature" in terms of their violation of the triviality constraints in Eq.~(\ref{eq: CNC}) \cite{Mermin1990,Mermin1993,Kernaghan1994,KernaghanPeres1995,CabelloEstebanzGarcia-Alcaine1997}. By Thm.~\ref{thm: CNC for OAs}, this is necessarily the case, and Cor.~\ref{cor: acyclic = KS} further distills this cyclic nature of KS contextuality.\\

\textbf{Meyer-Clifton-Kent (MCK) models.} It has been suggested that by restricting the set of quantum observables one may avoid the interpretational implications of the Kochen-Specker theorem \cite{Pitowsky1983,Pitowsky1985}. Motivated by the fact that two quantum observables, whose difference has infinitesimal operator norm, are experimentally indistinguishable, various restrictions to observable algebras $\cO\subset\LHsa$ that are dense in $\LHsa$ have been discussed. In Ref.~\cite{Meyer1999,Kent1999}, it was first shown that unlike rays in $\R^d$, rays in $\mathbb{Q}^d$ admit valuations (``truth functions''), that is, KS-colourings (see Def.~\ref{def: KS-colouring}). While this is not sufficient to bypass the Kochen-Specker theorem (see Sec.~\ref{sec: KS-colouring vs colouring}), the observable subalgebras $\cO\subset\LHsa$ constructed in Ref.~\cite{CliftonKent2000} admit a ``full collection of truth functions'', that is, are $d$-colourable and thus indeed Kochen-Specker noncontextual by Thm.~\ref{thm: KSNC = n-colourability}. With Cor.~\ref{cor: acyclic = KS}, we can give a new proof and structural explanation of this fact.

\begin{corollary}\label{cor: MCK models}
    The maximal observable algebras $\cO_\mathrm{MCK}$ in Ref.~\cite{CliftonKent2000} are Kochen-Specker noncontextual.
\end{corollary}

\begin{proof}
    Let $\cO_\mathrm{MCK}$ denote the maximal observable algebras generated by the projections in Thm.~2.2 in Ref.~\cite{CliftonKent2000}. By construction, any two non-identical maximal contexts $C,C'\in\cC_\mathrm{max}(\cO_\mathrm{MCK})$, $C\neq C'$ (equivalently, resolutions of the identity) overlap trivially $C\cap C' = \emptyset$. Consequently, $\cO_\mathrm{MCK}$ is acyclic and the result follows from Cor.~\ref{cor: acyclic = KS}.
\end{proof}

Acyclicity has also been studied in the marginal approach, where it characterises classically correlated observable algebras (see Def.~\ref{def: fully classically correlated}) by Vorob'ev's theorem \cite{Vorob1962}. In the next section, we will review Vorob'ev's theorem, compare it with Cor.~\ref{cor: acyclic = KS} and find that the characterisation of Kochen-Specker noncontextuality is more complex.

\subsubsection{Characterisation of classically correlated observable algebras}\label{sec: Vorob'ev theorem}

The problem of characterising all finitely generated observable algebras $\cO$ with $\Gamma[\CO]=\Gamma_\mathrm{cl}[\CO]$ (see Def.~\ref{def: fully classically correlated}) has a simple answer and was first obtained by Vorob'ev \cite{Vorob1962}. Here, we will phrase it in the setting of observable algebras. To this end, we will relate it with the formulation in Ref.~\cite{XuCabello2019} - in terms of the compatibility graphs $\cG(X)$ of measurement scenarios $M=(X,\cM)$, which we interpret as generating sets $X=\cX(\cO)$ of their uniquely associated observable algebras $\cO$ by Lm.~\ref{lm: OA-representation of measurement scenarios}. In particular, without loss of generality we will choose $\cX(\cO)=\PO$, interpreted as elementary ``yes-no" observables. For a detailed comparison between the respective notions, see App.~\ref{app: unified perspective}.

\begin{lemma}\label{lm: chordal=acyclic}
    Let $\cO$ be a finitely generated observable algebra with partial order of contexts $\CO$, event algebra $\PO$. Then $\cO$ is acyclic if and only if $\cG(\PO)$ is chordal.\footnote{A graph $G=(V,E)$ is \emph{chordal} if for any cycle $(v_0,\cdots,v_{n-1},v_0)$ with $v_i\in V$ for all $0\leq i\leq n-1$ and $n\geq 4$ there exists an edge connecting non-adjacent vertices in the cycle.}
\end{lemma}

\begin{proof} 
    ``$\Leftarrow$": Assume that $\cO$ is not acyclic and let $(C_0,\cdots,C_{n-1})$ be a non-trivial context cycle in $\CO$ of minimal length, in particular, $C_i\neq C_j$ for $i,j\in\zz_n$, $i\neq j$. For $i\in\zz_n$, denote by $p_{(i+1)i}$ the maximal projection of the subcontexts $C_{(i+1)i}=C_{(i+1)}\cap C_i$, then
    \begin{align*}
        C(p_{(i+1)i},\One):=\R p_{(i+1)i}+\R(\One-p_{(i+1)i})\subset C_{(i+1)i}\; .
    \end{align*}
    Since therefore $p_{(i+1)i},p_{i(i-1)}\in\cP(C_i)$, the corresponding dichotomoic observables are connected by an edge in the compatibility graph $\cG(\PO)$. It follows that $(p_{(i+1)i})_{i=0}^{n-1}$ defines a cycle in $\cG(\PO)$ which is non-chordal. Indeed, if it was chordal, there would exist $j\notin\{i-1,i,i+1\}$ such that $p_{(i+1)i}\in C_j$, hence, $C_i\cap C_j\neq C_\One$ contradicting that $(C_0,\cdots,C_{n-1})$ is a minimal non-trivial context cycle.

    ``$\Rightarrow$": Conversely, assume that $\cG(\PO)$ is non-chordal and let $(p_{10},\cdots,p_{0(n-1)})$ be a minimal non-chordal cycle in $\cG(\PO)$, that is, (the dichotomic observables defined by) $p_{ji}$ and $p_{kl}$ are connected by an edge in $\cG(\PO)$ if and only if $i=k$ or $j=l$. Define $C_{(i+1)i}:=C(p_{(i+1)i},\One)=\R p_{(i+1)i}+\R(\One-p_{(i+1)i})$ and let $C_i\in\COm$ be maximal contexts such that $C_{(i+1)i}=C_{(i+1)}\cap C_i$ for all $i\in\zz_n$. Then we have $C_{ij}= C_i\cap C_j=C_\One=C(\One)$ for $j\notin\{i-1,i,i+1\}$, since otherwise there would exist a smaller cycle in $\cG(\PO)$ contradicting minimality of $(p_{01},\cdots,p_{(n-1)0})$. It follows that $(C_0,\cdots,C_{n-1})$ is a non-trivial context cycle in $\CO$, hence, $\cO$ is acyclic.
\end{proof}

\begin{theorem}[Vorob'ev theorem \cite{Vorob1962,XuCabello2019}]\label{thm: Vorob'ev theorem}
    Let $\cO$ be a finitely generated observable algebra with partial order of contexts $\CO$. Then $\Gamma[\CO]=\Gamma_\mathrm{cl}[\CO]$ if and only if $\CO$ is acyclic.
\end{theorem}

\begin{proof}
    For any choice of generating set $\cX(\cO)$, we obtain a measurement scenario $\cM(\cX(\cO))$ by Def.~\ref{def: associated measurement scenario}. Clearly, if $\cX(\cO)=\PO$ is chordal then the same holds for any other choice of generating set, hence, $\cX(\cO)$ is chordal if and only if $\cO$ is acyclic by Lm.~\ref{lm: chordal=acyclic}. By the main result in Ref.~\cite{XuCabello2019} it follows that $\Gamma[\CO]=\Gamma_\mathrm{cl}[\CO]$.
    
    Conversely, if $\Gamma[\CO]=\Gamma_\mathrm{cl}[\CO]$ then the compatibility graph $\cG(\cX(\cO))$ of any measurement scenario corresponding to a generating set $\cX(\cO)$ of $\cO$ (via Lm.~\ref{lm: OA-representation of measurement scenarios}) is chordal by the main theorem in Ref.~\cite{XuCabello2019}. Hence, $\cO$ is acyclic by Lm.~\ref{lm: chordal=acyclic}. (For more details on the relation between measurement scenarios and observable algebras, see App.~\ref{app: measurement scenarios vs OAs}.)
\end{proof}

Crucially, the statement in Thm.~\ref{thm: Vorob'ev theorem} is a property of the partial order of contexts of an observable algebra only. As such it is similar to the characterisation of Kochen-Specker noncontextuality in terms of context connections in Thm.~\ref{thm: CNC for OAs}. However, by Prop.~\ref{prop: nonclassical correlations in KS noncontextual OAs} the respective notions of contextuality in Def.~\ref{def: fully classically correlated} and Def.~\ref{def: KSNC} differ. Since an acyclic observable algebra is also Kochen-Specker noncontextual by Cor.~\ref{cor: acyclic = KS}, it follows that acyclicity cannot also be sufficient for KS noncontextuality, that is, the converse of Cor.~\ref{cor: acyclic = KS} must fail.

\subsubsection{Kochen-Specker contextuality in single context cycles}\label{sec: single cycles}

Thm.~\ref{thm: CNC for OAs} casts the functional constraints in Eq.~(\ref{eq: KSNC}) as triviality constraints for context connections. This motivates the question whether a single (non-trivial) context cycle can already be sufficient to detect KS contextuality.

Recall that Ref.~\cite{AraujoEtAl2013} characterises the $n$-cycle contextuality scenario in dimension $d$ in Def.~\ref{def: n-cycle scenario} from the perspective of the marginal approach and Def.~\ref{def: fully classically correlated}. There, it is shown that $\Gamma_\mathrm{cl}[\cO(n,d)]\subsetneq \Gamma_\mathrm{qm}[\cO(n,d)]\subsetneq\Gamma_\mathrm{nd}[\cO(n,d)]$ (see Eq.~(\ref{eq: cl-qu-nd states})), in particular, there exist quantum correlations that violate some of the noncontextuality inequalities which define the noncontextuality polytope $\Gamma_\mathrm{cl}[\cO(n,d)]$. What is more, a closer look at the main result in Ref.~\cite{XuCabello2019} reveals that the direction ``fully classically correlated $\Rightarrow$ acyclic'' in Thm.~\ref{thm: Vorob'ev theorem} follows from the existence of beyond classical correlations in the $n$-cycle contextuality scenario. Since the other direction holds for both notions of contextuality, in particular, ``acyclic $\Rightarrow$ KS noncontextual'' by Cor.~\ref{cor: acyclic = KS}, and since both notions are distinct, the converse of Cor.~\ref{cor: acyclic = KS} must fail in general. To see this explicitly, we first discuss a well-known example in four-dimensional quantum theory, namely, we construct a classical embedding of the observable algebra corresponding to the CHSH scenario \cite{CHSH1969}.

\begin{example}[Classical embedding for the CHSH scenario]\label{ex: classical embedding of CHSH}
    Let $\Lambda=\Lambda_A\times\Lambda_B=\{\lambda_{ij}\}_{i,j=-2}^2$ denote a set consisting of $16$ elements, corresponding to the product of two ``elementary systems" with $\Lambda_A=\{\lambda^A_i\}_{i=-2}^2$ and $\Lambda_B=\{\lambda^B_j\}_{j=-2}^2$. Measurements are represented by certain sub-partitions of $\Lambda$.\footnote{This construction is reminiscent of Spekkens' toy theory \cite{Spekkens2007}, where measurements further satisfy the ``knowledge-balance principle", by which a ``maximally informative measurement" yields as much knowledge about $\Lambda$ as one lacks, together with consistency under repeated measurements and state-update rule (for details, see Ref.~\cite{Spekkens2007}). These additional constraints result in the fact that measurements generally do not commute, and are thus in conflict with a classical embedding of the form in Eq.~(\ref{eq: KSNC}). Here, we do not impose the ``knowledge-balance principle", and instead simply consider the usual operations on (measurable) functions in $L_\infty(\Lambda)$.} Maximally informative measurements partition the joint system into 4 subsets containing 4 ``ontic" states each. Non-maximally informative measurements arise by coarse-graining of maximal ones, where coarse-graining obtains the natural meaning on $\Lambda$ (see Fig.~\ref{fig: classical embedding for CHSH}).
    
    The observable algebra $\cO_\mathrm{CHSH}$ underlying the CHSH scenario \cite{CHSH1969} is generated by four product measurements $S_A\otimes S_B$, $S_A\otimes S'_B$, $S'_A\otimes S_B$ and $S'_A\otimes S'_B$ on a pair of spin-$\frac{1}{2}$ systems, corresponding to two choices of measurements for Alice and Bob. Together, these generate four maximal contexts with subcontexts given by the coarse-grained observables $S_A\otimes\One_B,S'_A\otimes\One_B$ and $\One_A\otimes S_B,\One_A\otimes S'_B$ and trivial context generated by $\One=\One_A\otimes\One_B$.
    
    We represent Alice's local spin-$\frac{1}{2}$ observables as (measurable) functions on $\Lambda_A$ via
    \begin{align*}
        f_{S_A} &:= \{\{\lambda^A_{-2},\lambda^A_{-1}\}\ra \bar{1},\ 
        \{\lambda^A_1,\lambda^A_2\}\ra 1\}, \\[0.2cm]
        f_{S'_A} &:= \{\{\lambda^A_{-2},\lambda^A_2\}\ra \bar{1},\ 
        \{\lambda^A_{-1},\lambda^A_1\}\ra 1\}, \\[0.2cm]
        f_{\One_A} &:= \{\{\Lambda_A\ra 1\}\; ,
    \end{align*}
    and similarly for Bob. Product observables are represented via $f_{S_AS_B}=f_{S_A}\times f_{S_B}$, e.g.
    \begin{align*}
        f_{S_AS_B} &:= \{\{\lambda^A_{-2},\lambda^A_{-1}\}\times\{\lambda^B_{-2},\lambda^B_{-1}\}\ra 1,\ 
        \{\lambda^A_{-2},\lambda^A_{-1}\}\times\{\lambda^B_1,\lambda^B_2\}\ra\bar{1},\\
        &\quad\ \ \{\lambda^A_1,\lambda^A_2\}\times\{\lambda^B_{-2},\lambda^B_{-1}\}\ra\bar{1},\ 
        \{\lambda^A_1,\lambda^A_2\}\times\{\lambda^B_1,\lambda^B_2\}\ra 1\}\; .
    \end{align*}
    These functions together with their coarse-graining relations are shown in Fig.~\ref{fig: classical embedding for CHSH} (a).

    The map $\epsilon:\cO_\mathrm{CHSH}\ra L_\infty(\Lambda)$ sending the generating observables in $\cO_\mathrm{CHSH}$ to (measurable) functions above thus defines a classical embedding for the CHSH scenario. Clearly, we have $\emptyset\neq\Gamma_\mathrm{cl}[\cO_\mathrm{CHSH}]$; moreover, $\Gamma_\mathrm{cl}[\cO_\mathrm{CHSH}]\subsetneq\Gamma[\cO_\mathrm{CHSH}]$ since the CHSH scenario allows for entanglement. Explicitly, the marginals arising from a Bell state, e.g. $|\psi\rangle=\frac{1}{\sqrt{2}}(|00\rangle+|11\rangle)$ which results in the (maximal) violation of the CHSH inequality in quantum theory \cite{CHSH1969,Tsirelson1980}, is not a valid probability distribution in $L^{\geq 0}_1(\Lambda)$.\footnote{This should not be confused with the representation of Bell entangled states as epistemic states on $\Lambda$ in Spekkens' toy theory, which rely on a different (to the classical Bayesian) state-update rule defined in line with the ``knowledge-balance principle" \cite{Spekkens2007}.} Indeed, denoting the outcomes of the respective observables by lower case letters, the events defined by $s_As_B=1$, $s_As'_B=1$, $s'_As_B=1$ and $s'_As'_B=\bar{1}$ are shown in Fig.~\ref{fig: classical embedding for CHSH} (b), together with the corresponding marginals of $\gamma_{|\psi\rangle}\in\Gamma_\mathrm{qu}[\cC(\cO_{\mathrm{CHSH}})]$, which maximises the (quantum-mechanical) expectation value $\mathbb{E}_{|\psi\rangle}(S_AS_B+S_AS'_B+S'_AS_B-S'_AS'_B)=2\sqrt{2}$ \cite{CHSH1969,Tsirelson1980}.\footnote{Note that the maximal violation of the CHSH inequality for general observable algebras (hence, for general probabilistic theories) even exceeds the (quantum) Tsirelson bound \cite{PopescuRohrlich1994,Tsirelson1980}.} Assuming that $\gamma_{|\psi\rangle}\in L_1(\Lambda)$,
    \begin{align*}
        \sum_{\lambda\in\Lambda}\gamma_{|\psi\rangle}(\lambda)+2\sum_{\lambda\in S\subset\Lambda}\gamma_{|\psi\rangle}(\lambda)=4\cos^2\left(\frac{\pi}{8}\right)>3\; .
    \end{align*}
    Yet, since $\sum_{\lambda\in\Lambda}\mu(\lambda)=1$, hence, $\sum_{\lambda\in S\subset\Lambda}\mu(\lambda)<1$ for every classical state $\mu\in L_1(\Lambda)$ with $\mu>0$, this implies $\gamma_{|\psi\rangle}\not >0$, hence, $\gamma_{|\psi\rangle}\notin \Gamma_\mathrm{cl}[\cC(\cO_{\mathrm{CHSH}})]$.
\end{example}

\begin{figure}
    \centering
    \begin{subfigure}{\textwidth}
        \centering
        \scalebox{0.6}{\input{CHSH_embedding}}
        \subcaption{Classical embedding $\epsilon:\cO_\mathrm{CHSH}\ra L_\infty(\Lambda)$ for the CHSH scenario \cite{CHSH1969}. The product observables $S\in\cO_\mathrm{CHSH}$ of local spin-$\frac{1}{2}$ observables $S_A,S'_A,\One_A$ for Alice and $S_B,S'_B,\One_B$ for Bob are represented as functions $f_S$ on the ``ontic'' state space $\Lambda$. Coarse-graining relations between these observables (and the respective contexts they generate) are indicated by dashed arrows.}
    \end{subfigure}
\end{figure}
\begin{figure}[ht]\ContinuedFloat
    \centering
    \begin{subfigure}{\textwidth}
        \centering
        \scalebox{0.8}{\begin{tikzpicture}

    \node [] (0) at (-2.5, 2.5) {};
    \node [] (1) at (2.5, 2.5) {};
    \node [] (2) at (-2.5, -2.5) {};
    \node [] (3) at (2.5, -2.5) {};
    \node [] (4) at (-1.25, 2.5) {};
    \node [] (5) at (0, 2.5) {};
    \node [] (6) at (1.25, 2.5) {};
    \node [] (7) at (-2.5, 1.25) {};
    \node [] (8) at (-2.5, 0) {};
    \node [] (9) at (-2.5, -1.25) {};
    \node [] (10) at (-1.25, -2.5) {};
    \node [] (11) at (0, -2.5) {};
    \node [] (12) at (1.25, -2.5) {};
    \node [] (13) at (2.5, -1.25) {};
    \node [] (14) at (2.5, 0) {};
    \node [] (15) at (2.5, 1.25) {};
    \node [] (21) at (0, -3.25) {$\Lambda$};
    \node [] (23) at (5, 4) {};
    \node [] (24) at (8, 4) {};
    \node [] (25) at (5, 1) {};
    \node [] (26) at (8, 1) {};
    \node [] (27) at (5.75, 4) {};
    \node [] (28) at (6.5, 4) {};
    \node [] (29) at (7.25, 4) {};
    \node [] (30) at (5, 3.25) {};
    \node [] (31) at (5, 2.5) {};
    \node [] (32) at (5, 1.75) {};
    \node [] (33) at (5.75, 1) {};
    \node [] (34) at (6.5, 1) {};
    \node [] (35) at (7.25, 1) {};
    \node [] (36) at (8, 1.75) {};
    \node [] (37) at (8, 2.5) {};
    \node [] (38) at (8, 3.25) {};
    \node [] (42) at (5, -1) {};
    \node [] (43) at (8, -1) {};
    \node [] (44) at (5, -4) {};
    \node [] (45) at (8, -4) {};
    \node [] (46) at (5.75, -1) {};
    \node [] (47) at (6.5, -1) {};
    \node [] (48) at (7.25, -1) {};
    \node [] (49) at (5, -1.75) {};
    \node [] (50) at (5, -2.5) {};
    \node [] (51) at (5, -3.25) {};
    \node [] (52) at (5.75, -4) {};
    \node [] (53) at (6.5, -4) {};
    \node [] (54) at (7.25, -4) {};
    \node [] (55) at (8, -3.25) {};
    \node [] (56) at (8, -2.5) {};
    \node [] (57) at (8, -1.75) {};
    \node [] (118) at (-8, 4) {};
    \node [] (119) at (-5, 4) {};
    \node [] (120) at (-8, 1) {};
    \node [] (121) at (-5, 1) {};
    \node [] (122) at (-7.25, 4) {};
    \node [] (123) at (-6.5, 4) {};
    \node [] (124) at (-5.75, 4) {};
    \node [] (125) at (-8, 3.25) {};
    \node [] (126) at (-8, 2.5) {};
    \node [] (127) at (-8, 1.75) {};
    \node [] (128) at (-7.25, 1) {};
    \node [] (129) at (-6.5, 1) {};
    \node [] (130) at (-5.75, 1) {};
    \node [] (131) at (-5, 1.75) {};
    \node [] (132) at (-5, 2.5) {};
    \node [] (133) at (-5, 3.25) {};
    \node [] (137) at (-8, -1) {};
    \node [] (138) at (-5, -1) {};
    \node [] (139) at (-8, -4) {};
    \node [] (140) at (-5, -4) {};
    \node [] (141) at (-7.25, -1) {};
    \node [] (142) at (-6.5, -1) {};
    \node [] (143) at (-5.75, -1) {};
    \node [] (144) at (-8, -1.75) {};
    \node [] (145) at (-8, -2.5) {};
    \node [] (146) at (-8, -3.25) {};
    \node [] (147) at (-7.25, -4) {};
    \node [] (148) at (-6.5, -4) {};
    \node [] (149) at (-5.75, -4) {};
    \node [] (150) at (-5, -3.25) {};
    \node [] (151) at (-5, -2.5) {};
    \node [] (152) at (-5, -1.75) {};
    \node [] (155) at (-4.5, -2.5) {};
    \node [] (156) at (-3, 0) {};
    \node [] (157) at (-4.5, 2.5) {};
    \node [] (158) at (3, 0) {};
    \node [] (159) at (4.5, 2.5) {};
    \node [] (160) at (4.5, -2.5) {};
    \node [] (161) at (-3.5, 1.5) {\rotatebox{-59}{\scriptsize{$s_As_B=1$}}};
    \node [] (162) at (-3.5, -1.5) {\rotatebox{59}{\scriptsize{$s'_As_B=1$}}};
    \node [] (163) at (3.5, 1.5) {\rotatebox{59}{\scriptsize{$s_As'_B=1$}}};
    \node [] (164) at (3.5,-1.5) {\rotatebox{-59}{\scriptsize{$s'_As'_B=\bar{1}$}}};

    \fill[blue!40!white] (-5,4) rectangle (-6.5,2.5);
    \fill[blue!40!white] (-6.5,2.5) rectangle (-8,1);
    
    \fill[blue!40!white] (-5.75,-1) rectangle (-7.25,-2.5);
    \fill[blue!40!white] (-5,-2.5) rectangle (-5.75,-4);
    \fill[blue!40!white] (-7.25,-2.5) rectangle (-8,-4);

    \fill[blue!40!white] (8,3.25) rectangle (6.5,1.75);
    \fill[blue!40!white] (6.5,1.75) rectangle (5,1);
    \fill[blue!40!white] (6.5,3.25) rectangle (5,4);

    \fill[blue!40!white] (5.75,-1) rectangle (7.25,-1.75);
    \fill[blue!40!white] (5,-1.75) rectangle (5.75,-3.25);
    \fill[blue!40!white] (7.25,-1.75) rectangle (8,-3.25);
    \fill[blue!40!white] (5.75,-3.25) rectangle (7.25,-4);

    \fill[blue!40!white] (-2.5,2.5) rectangle (-1.25,1.25);
    \fill[blue!60!white] (-1.25,2.5) rectangle (1.25,1.25);
    \fill[blue!40!white] (1.25,2.5) rectangle (2.5,1.25);
    \fill[blue!40!white] (-2.5,1.25) rectangle (0,0);
    \fill[blue!60!white] (0,1.25) rectangle (2.5,0);
    \fill[blue!60!white] (-2.5,0) rectangle (-1.25,-1.25);
    \fill[blue!40!white] (-1.25,0) rectangle (2.5,-1.25);
    \fill[blue!60!white] (1.25,0) rectangle (2.5,-1.25);
    \fill[blue!60!white] (-2.5,-1.25) rectangle (0,-2.5);
    \fill[blue!40!white] (0,-1.25) rectangle (2.5,-2.5);


    \draw (0.center) to (1.center);
    \draw (0.center) to (2.center);
    \draw (2.center) to (3.center);
    \draw (1.center) to (3.center);
    \draw [thin,dotted] (7.center) to (15.center);
    \draw [thin,dotted] (9.center) to (13.center);
    \draw [thin,dotted] (6.center) to (12.center);
    \draw [thin,dotted] (4.center) to (10.center);
    \draw [thin,dotted] (5.center) to (11.center);
    \draw [thin,dotted] (8.center) to (14.center);
    \draw (23.center) to (24.center);
    \draw (23.center) to (25.center);
    \draw (25.center) to (26.center);
    \draw (24.center) to (26.center);
    \draw [thin,dotted] (30.center) to (38.center);
    \draw [thin,dotted] (31.center) to (37.center);
    \draw [thin,dotted] (32.center) to (36.center);
    \draw [thin,dotted] (28.center) to (34.center);
    \draw [thin,dotted] (29.center) to (35.center);
    \draw [thin,dotted] (27.center) to (33.center);
    \draw (42.center) to (43.center);
    \draw (42.center) to (44.center);
    \draw (44.center) to (45.center);
    \draw (43.center) to (45.center);
    \draw [thin,dotted] (49.center) to (57.center);
    \draw [thin,dotted] (50.center) to (56.center);
    \draw [thin,dotted] (51.center) to (55.center);
    \draw [thin,dotted] (47.center) to (53.center);
    \draw [thin,dotted] (48.center) to (54.center);
    \draw [thin,dotted] (46.center) to (52.center);
    \draw (118.center) to (119.center);
    \draw (118.center) to (120.center);
    \draw (120.center) to (121.center);
    \draw (119.center) to (121.center);
    \draw [thin,dotted] (125.center) to (133.center);
    \draw [thin,dotted] (126.center) to (132.center);
    \draw [thin,dotted] (127.center) to (131.center);
    \draw [thin,dotted] (123.center) to (129.center);
    \draw [thin,dotted] (124.center) to (130.center);
    \draw [thin,dotted] (122.center) to (128.center);
    \draw (137.center) to (138.center);
    \draw (137.center) to (139.center);
    \draw (139.center) to (140.center);
    \draw (138.center) to (140.center);
    \draw [thin,dotted] (144.center) to (152.center);
    \draw [thin,dotted] (145.center) to (151.center);
    \draw [thin,dotted] (146.center) to (150.center);
    \draw [thin,dotted] (142.center) to (148.center);
    \draw [thin,dotted] (143.center) to (149.center);
    \draw [thin,dotted] (141.center) to (147.center);
    \draw [->,thick] (156.center) to (157.center);
    \draw [->,thick] (156.center) to (155.center);
    \draw [->,thick] (158.center) to (159.center);
    \draw [->,thick] (158.center) to (160.center);
\end{tikzpicture}}
        \subcaption{The marginals $\gamma_{|\psi\rangle}(s_As_B=1)=\gamma_{|\psi\rangle}(s_As'_B=1)=\gamma_{|\psi\rangle}(s'_As_B=1)=\gamma_{|\psi\rangle}(s'_As'_B=\bar{1})=\cos^2(\frac{\pi}{8})$ of a (Bell-entangled) quantum state $\gamma_{|\psi\rangle}\in\Gamma_\mathrm{qu}[\cC(\cO_\mathrm{CHSH})]$ which maximally violates the CHSH inequality \cite{CHSH1969,Tsirelson1980} in quantum mechanics (similarly, for the ``Popescu-Rohrlich" box $\gamma_\mathrm{PR}\in\Gamma[\cC(\cO_\mathrm{CHSH})]$ with $\gamma_\mathrm{PR}(s_As_B=1)=\gamma_\mathrm{PR}(s_As'_B=1)=\gamma_\mathrm{PR}(s'_As_B=1)=\gamma_\mathrm{PR}(s'_As'_B=\bar{1})=\frac{1}{2}$). The shaded regions in the middle show single and triple overlaps of the distributions representing the respective events on the top/bottom left/right.}
    \end{subfigure}
    \caption{The CHSH scenario \cite{CHSH1969} admits a classical embedding (a), hence, is Kochen-Specker noncontextual, while also exhibiting nonclassical correlations e.g. in the form of entangled states and more general non-signalling correlations (b).}
    \label{fig: classical embedding for CHSH}
\end{figure}
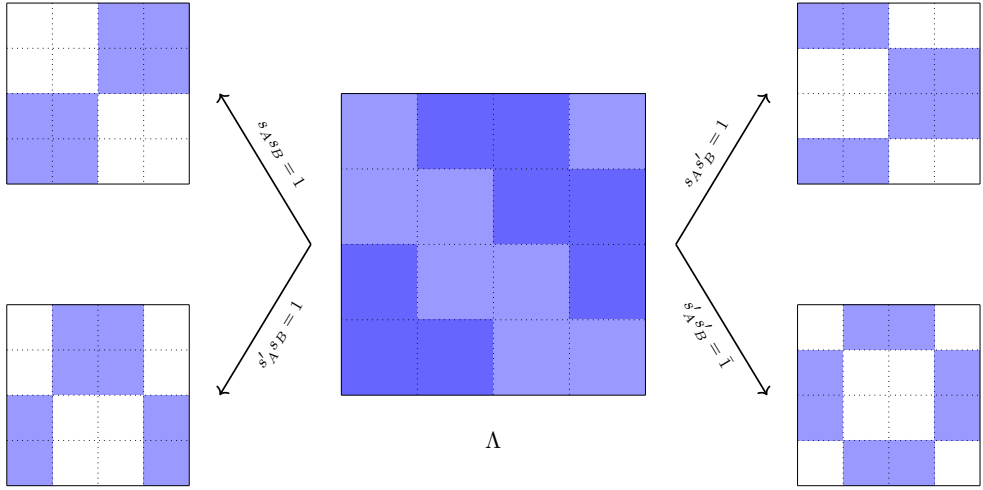

Ex.~\ref{ex: classical embedding of CHSH} shows that a proof of Kochen-Specker contextuality must generally involve constraints encoded in several non-trivial context cycles. Our next theorem asserts that this is always the case in quantum-mechanical systems of dimension three.

\begin{theorem}\label{thm: single cycles in d=3 are KS noncontextual}
    Let $(C_0,\cdots,C_{n-1})$ be a context cycle in $\cC_\mathrm{max}(\cH)$ with $\dim(\cH)=3$ and $n\in \mathbb{N}$. Then there exist unitaries $(U_{(i+1)i})_{i\in\zz_n}$, $U_{(i+1)i} \in \UH$ such that $C_{i+1}= U_{(i+1)i}.C_i:=U_{(i+1)i} C_i U^*_{(i+1)i}$ and $C_{(i+1)i} = U_{(i+1)i}.C_{(i+1)i}$ for all $i\in\zz_n$, and such that
    \begin{equation}\label{eq: single cycles in d=3 are noncontextual}
        \prod_{i=0}^{n-1} \Ad_\PH(U_{(i+1)i}) = \id \; .
    \end{equation}
    In particular, there exists a flat context connection on $\cC(n,d)$, defined by $(C_0,\cdots,C_{n-1})$.
\end{theorem}

\begin{proof}
    We first proceed by induction in the cycle length $n$ to show that we can reduce to minimal non-trivial context cycles. To this end, we consider the following two cases.
    \begin{itemize}
        \item[(i)] First, let $C_i=C_j$ for some $i\neq j$ in the context cycle $(C_0,\cdots,C_{n-1})$, then by the inductive hypothesis there exist unitaries $(U_{(i+1)i})_{i\in\zz_n}$ for the context cycles $(C_i,\cdots,C_{j-1})$ and $(C_j,\cdots,C_{i-1})$ such that $\prod_{k=i}^{j-1} \Ad_\PH(U_{(k+1)k}) = \id = \prod_{k=j}^{i-1} \Ad_\PH(U_{(k+1)k})$. Consequently, we also have $\prod_{k=0}^{n-1} \Ad_\PH(U_{(k+1)k}) = \id$.
    \end{itemize}
    We may thus assume that $C_i\neq C_j$ for all $i,j \in \zz_n, i\neq j$.
    \begin{itemize}
        \item[(ii)] Second, let $C_{(i+1)i}=C_{i(i-1)}$ for some $i\in\zz_n$. Since $d=3$, this implies $C_{(i+1)(i-1)}=C_{(i+1)i}=C_{i(i-1)}$. Indeed, there exists $p \in \mc{P}_1(C_{(i+1)i}), \mc{P}_1(C_{i(i-1)})$ from which $p \in \mc{P}_1(C_{i+1}),\mc{P}_1(C_{i-1})$, and, since $C_{i+1} \neq C_{i-1}$ by (i), $C_{(i+1)(i-1)} \subsetneq C_{(i+1)},C_{(i-1)}$.

        W.l.o.g., let $i=n-1$ and consider the (nontrivial) context cycle $(C_0,\cdots,C_{n-2})$. Since this cycle has length $n-1$, by the inductive hypothesis, there exist unitaries $U_{(i+1)i} \in \UH$, $C_{(i+1)} = U_{(i+1)i}.C_i$, $C_{(i+1)i} = U_{(i+1)i}.C_{(i+1)i}$ for all $i\in\zz_{n-1}$ and such that $\prod_{i=0}^{n-2} \Ad_\PH(U_{(i+1)i}) = \id$. Since $U_{0(n-2)}$ stabilises $C_{0(n-2)}=C_{0(n-1)}=C_{(n-1)(n-2)}$ we can decompose (the non-trivial action on the two-dimensional subspace corresponding to $1-p$ as) $U_{0(n-2)}=U_{0(n-1)}U_{(n-1)(n-2)}$ into unitaries $U_{0(n-1)}$ and $U_{(n-1)(n-2)}$ such that both $C_{0(n-1)}=U_{0(n-1)}.C_{0(n-1)}$ and $C_{(n-1)(n-2)}=U_{(n-1)(n-2)}.C_{(n-1)(n-2)}$, as well as $C_0=U_{0(n-1)}.C_{n-1}$ and $C_{n-1}=U_{(n-1)(n-2)}.C_{n-2}$ hold. Clearly, we then have
        \begin{equation*}
            \prod_{i=0}^{n-1} \Ad_\PH(U_{(i+1)i}) = \left(\prod_{i=0}^{n-2} \Ad_\PH(U_{(i+1)i})\right) = \id \; .
        \end{equation*}
    \end{itemize}

    Now, let $(\tU_{(i+1)i})_{i\in\zz_n}$, $\tU_{(i+1)i} \in \UH$ be unitaries such that $\prod_{i=0}^{n-1} \Ad_\PH(\tU_{(i+1)i}) = \pi \in \Sym(C_0) \cong S_3$, while also satisfying the constraints $C_{i+1} = \tU_{(i+1)i}.C_i$ and $C_{(i+1)i} = \tU_{(i+1)i}.C_{(i+1)i}$. Note that the latter two constraints define $\tU_{(i+1)i}$ only up to permutations $\pi_{(i+1)i} \in \Stab_{\Sym(C_i)}(C_{(i+1)i})$ and $\pi_{i(i+1)} \in \Stab_{\Sym(C_{i+1})}(C_{(i+1)i})$, that is,
    \begin{align}\label{eq: re-defining permutations}
        U_{(i+1)i} = \pi_{i(i+1)}\tU_{(i+1)i}\pi_{(i+1)i}
    \end{align}
    satisfies the same constraints as $\tU_{(i+1)i}$.\footnote{Note also that the action of $U_{(i+1)i}$ between contexts $C_i$ and $C_{(i+1)}$ is left invariant under a transformation $U_{(i+1)i} \ra O_{{(i+1)}}U_{(i+1)i}O_i$ for any operator $O_i\in C_i$ and $O_{(i+1)}\in C_{(i+1)}$.} We will prove the existence of a set of unitaries $(U_{(i+1)i})_{i\in\zz_n}$ satisfying Eq.~(\ref{eq: single cycles in d=3 are noncontextual}) by choosing appropriate permutations as above.
    
    By (i) and (ii), we may assume that $C_{(i+1)i}\neq C_{i(i-1)}$ for all $i\in\zz_n$. By Thm.~\ref{thm: KSNC of truncated OAs}, we may further assume that $C_{(i+1)i},C_{i(i-1)}\neq\C\mathbbm{1}$. In other words, $(C_0,\cdots,C_{n-1})$ defines a $n$-cycle contextuality scenario in dimension $d=3$ (see Def.~\ref{def: n-cycle scenario}), hence, $C_{(i+1)i}$ and $C_{i(i-1)}$ are generated by different projections in $\mc{P}_1(C_i)$ (together with the identity $\mathbbm{1}$) for every $i\in\zz_n$. It follows that the permutations $\pi_{(i+1)i}$ and $\pi_{(i-1)i}$ in Eq.~(\ref{eq: re-defining permutations}) (which stabilise $C_{(i+1)i}$ and $C_{i(i-1)}$, respectively) correspond to distinct transpositions in $\Sym(C_i)$.
        
    Fix a maximal context $C_0$. Without loss of generality, we may assume that $\pi_{10} = \id,\tau_{12} \in \Sym(C_0)$ and $\pi_{(n-1)0}=\id,\tau_{23} \in \Sym(C_0)$.\footnote{In our notation we distinguish between generic permutations $\pi\in\Sym(C_i)$ and explicit elements $\sigma_{\pi(1,\cdots,n)}$; we also write $\tau_{ij}$ for the permutation corresponding to a transposition of elements $i,j\in[n]$.} If $\pi\in\Pi_1=\{\id,\tau_{12},\tau_{23},\tau_{12}\tau_{23}=\sigma_{312}\}$, Eq.~(\ref{eq: single cycles in d=3 are noncontextual}) is thus satisfied for appropriate choices of $\pi_{10}$ in $U_{10} = \tU_{10}\pi_{10}$ and $\pi_{(n-1)0}$ in $U_{0(n-1)}=\pi_{(n-1)0}\tU_{0(n-1)}$. For $\pi\notin\Pi_1$, we further consider a change of unitary $U_{21}=\tU_{21}\pi_{21}$, where $\pi_{21}\in\Stab_{\Sym(C_1)}(C_{21})$, which corresponds to a permutation $\tpi_{21}=U^{-1}_{10}\pi_{21}U_{10}\in\Sym(C_0)$. By the same argument as above, $C_{21}$ and $C_{10}$ are generated by different projections in $\mc{P}_1(C_1)$, hence, $\pi_{21}$ and $\pi_{01}$ induce different transpositions in $\Sym(C_1)$, and thus also in $\Sym(C_0)$ under conjugation by $U^{-1}_{10}$. For this induced symmetry in $\Sym(C_0)$, we therefore either have $\tpi_{21} = \id,\tau_{23}$ or $\tpi_{21} = \id,\tau_{13}$. The unitaries in the first case generate any permutation in
    \begin{equation*}
        \Pi_1\cup\{\tau_{23}\tau_{12}=\sigma_{231}=\sigma^2_{312},\tau_{23}\tau_{12}\tau_{23}=\tau_{13}\} = S_3\; .
    \end{equation*}
    Similarly, the unitaries in the second case generate any permutation in
    \begin{equation*}
        \Pi_1\cup\{\tau_{13},\tau_{13}\tau_{23}=\sigma_{231}=\sigma^2_{312}\} = S_3\; .
    \end{equation*}
    Hence, for any $\pi=\prod_{i=0}^{n-1} \Ad_\PH(\tU_{(i+1)i})\in\Sym(C_0)\cong S_3$, we can find permutations $\tpi_{21},\pi_{10},\pi_{0(n-1)}\in\Sym(C_0)$ (together with $\pi_{(k+1)k}=\id$ for all $k\neq 1,0,n-1$) such that the unitaries $(U_{(i+1)i})_{i\in\zz_n}$, $U_{(i+1)i} = \tU_{(i+1)i}\pi_{(i+1)i}$ satisfy Eq.~(\ref{eq: single cycles in d=3 are noncontextual}). Finally, $(U_{(i+1)i})_{i\in\zz_n}$ defines a context connection $\fl$ by setting $l_{(i+1)i}(p) = U_{(i+1)i}.p$ for every $i\in\zz_n$ (and $l_{ij}$ arbitrary for $j\neq i\pm 1$), which satisfies Eq.~(\ref{eq: CNC}) as a consequence of Eq.~(\ref{eq: single cycles in d=3 are noncontextual}).
\end{proof}

Thm.~\ref{thm: single cycles in d=3 are KS noncontextual} should be contrasted with the study of the $n$-cycle scenario in Ref.~\cite{AraujoEtAl2013,XuCabello2019}, which instead of Kochen-Specker noncontextual observable algebras is concerned with fully classically correlated observable algebras (see Def.~\ref{def: fully classically correlated}). As we have seen, the latter are characterised in terms of acyclic context categories by Vorob'ev's theorem (Thm.~\ref{thm: Vorob'ev theorem}). Moreover, acyclicity implies KS noncontextuality by Cor.~\ref{cor: acyclic = KS}. Yet, the converse does not hold, and thus allows for ``hybrid" scenarios such as the CHSH scenario in Ex.~\ref{ex: classical embedding of CHSH} which exhibits both classical and quantum features (cf. Ref.~\cite{Spekkens2007}). This shows that the characterisation of Kochen-Specker contextuality is more complex than that of Def.~\ref{def: fully classically correlated}: \emph{it requires not just a single context cycle, but constraints arranged along several context cycles}, e.g. as explicated by the topological arrangements in Ref.~\cite{OkayRoberts2017,OkayTyhurstRaussendorf2018,OkayRaussendorf2020,Arkhipov2012}.

It would be interesting to see whether Thm.~\ref{thm: single cycles in d=3 are KS noncontextual} generalises to arbitrary context cycles in higher dimensions. A more ambitious objective is to characterise quantum within more general observable algebras in terms of the constraints on context connections it gives rise to. Such would amount to a geometric characterisation and explanation of Kochen-Specker contextuality in quantum theory.

To (efficiently) characterise and identify Kochen-Specker noncontextual observable algebras remains an open problem - one that is not only of foundational, but also of practical interest. Its foundational importance mainly stems from clarifying the classical-quantum divide, specifically with respect to the emergence of classicality. From a practical perspective, contextuality has been identified as the main resource in quantum computation \cite{Raussendorf2013,HowardEtAl2014,FrembsRobertsBartlett2018,FrembsRobertsCampbellBartlett2023,BravyiGossetKoenig2018}. Characterising observable subalgebras $\cO\subset\LHsa$ in quantum theory in terms of KS contextuality thus presents an important task towards provable quantum advantage. We believe that the formalism presented here and the new tools it introduces will prove useful towards both of these goals.

In the next section, we collect some further evidence in support of this claim. More precisely, by comparing our formalism with the graph-theoretic approach to contextuality \cite{CabelloSeveriniWinter2014}, we will sharpen the (formal) characterisation of KS contextuality in terms of a colouring problem in Thm.~\ref{thm: KSNC = n-colourability} and thereby generalise various known results.

\clearpage

\section{Algebraic vs graph-theoretic approach}\label{sec: algebraic vs graph-theoretic approach}

In this section, we compare the notion of Kochen-Specker contextuality in Def.~\ref{def: KSNC} and its characterisation in terms of the triviality constraints in Thm.~\ref{thm: CNC for OAs}, specifically, in terms of the colouring problem in Thm.~\ref{thm: KSNC = n-colourability} and Thm.~\ref{thm: n-colourability for non-maximal OAs}, with various notions of contextuality for orthogonality graphs in the graph-theoretic approach \cite{CabelloSeveriniWinter2014,RamanathanHorodecki2014,CabelloKleinmannBudroni2015,AmaralCunha2018}.

We start by reviewing the relevant properties of orthogonality graphs in Sec.~\ref{sec: basics of graph-theoretic approach}. In Sec.~\ref{sec: consistency with graph-theoretic approach}, we translate the notion of Kochen-Specker contextuality for observable algebras to an equivalent notion applicable to orthogonality graphs, which is similar the notion of marginal Kochen-Specker contextuality in Def.~\ref{def: marginal noncontextuality}, and restricted to cases where both notions are directly comparable. However, as we show in App.~\ref{app: orthogonality graphs vs OAs} this is generally not the case. Yet, under appropriate assumptions on the realisation of an orthogonality graph, in Sec.~\ref{sec: SI-C sets and graphs}, we relate our results with known criteria for SI-C graphs and SI-C sets, which are expressed in terms of their chromatic number \cite{Cabello2012b,RamanathanHorodecki2014,CabelloKleinmannBudroni2015}. In doing so, we will further answer an open question in Ref.~\cite{RamanathanHorodecki2014}

\subsection{Basic elements of the graph-theoretic approach}\label{sec: basics of graph-theoretic approach}

The graph-theoretic approach was initiated in Ref.~\cite{ChavesFritz2012,FritzChaves2013,CabelloSeveriniWinter2014}, and has proven a successful tool in the study of contextuality. We refer to Ref.~\cite{CabelloSeveriniWinter2014,AcinFritzLeverrierSainz2015,AmaralCunha2018} for a detailed exposition, and restrict ourselves to the very basics needed for our purposes.

To begin with, recall that an (undirected) graph $G = (V,E)$ is specified by the set of its vertices $V$ and edges $E \subset V \times V$. Two vertices that are joined by an edge are called \emph{adjacent}. We will restrict to finite graphs for which $V = \{v_1,\cdots,v_n\}$ (and thus also $E$) are finite sets. A subset $I \subset V$ is called an \emph{independent set} if $(v,v') \notin E$ for all $v,v' \in I$. In contrast, a \emph{clique $\rC$ of $G$} is a fully connected subgraph of $G$.\footnote{We will distinguish cliques in the graph-theoretic approach from contexts (commutative algebras in observable algebras) by using the non-italic font for the former.} We denote by $\omega(G)=\max_{\rC\subset G}|\rC|$ the cardinality of the largest clique in $G$. Finally, the \emph{chromatic number $\chi(G)$ of $G$} is the minimal number of colours needed to colour the vertices of $G$ such that no two adjacent vertices are assigned the same colour.\\

\textbf{Orthogonality graphs and their realisations.} The essential idea behind the graph-theoretic approach is to reduce the study of contextuality to the bare essentials needed to describe the constraints on correlations.

\begin{definition}[orthogonality graph]\label{def: orthogonality graph}
    Let $\cP\subset\PH$ be a finite set of minimal projections on a Hilbert space $\cH$, $\dim(\cH)<\infty$. The \emph{orthogonality graph $G(\cP) = (V,E)$ of $\cP$} has vertices $V = \cP$, and edges $(p,p')\in E\subset\cP\times\cP$ if and only if $pp'=0$.
\end{definition}

Note first that since every non-minimal projection $p\in\PH$ can be decomposed into finitely many orthogonal projections (and since $p$ inherits its orthogonality relations from the elements in this decomposition), for observable subalgebras of quantum theory $\cO\subset\LHsa$ it is sufficient to restrict to the orthogonality graph of its minimal projections.

The properties of projections impose some conditions on the types of graphs that can arise as orthogonality graphs. For instance, since $p^2 = p \neq 0$ (unless $p$ is the zero projection), an orthogonality graph contains no self-loops. Moreover, there exists at most one edge between two nodes, that is, an orthogonality graph contains no multi-edges. Graphs without loops and multi-edges are called \emph{simple}. One may ask whether every finite, simple graph $G$ can be realised in terms of projections on some Hilbert space. Without further constraints, e.g. on the dimension of the Hilbert space, this is indeed the case \cite{HeunenFritzReyes2014}.\footnote{We will come across various more restrictive notions of (quantum) realisations throughout this section. We list these in App.~\ref{app: On realisations}; to the best of our knowledge their solutions remain by and large open.} It follows that the restriction in Def.~\ref{def: orthogonality graph} can be relaxed to the above general constraints on graphs. In particular, every finite, simple \emph{exclusivity graph}, whose vertices represent events and whose edges relate events if they are exclusive \cite{FritzEtAl2013,CabelloSeveriniWinter2014}, can be interpreted as an orthogonality graph.\footnote{Another closely related concept is that of a compatibility graph (see also App.~\ref{app: measurement scenarios vs OAs}), which instead relates events if they are compatible. In the case of ideal observables (specifically, for elementary events), as discussed here, the respective notions are easily seen to encode the same information.}

Yet, reducing the study of contextuality to orthogonality graphs raises a new question: which observable algebras are (faithfully) represented by orthogonality graphs? Recall that by Lm.~\ref{lm: embeddings vs orthomorphisms vs order-preserving maps}, event algebras capture the full structure of observable algebras. Yet only for partial algebras, that is, under the assumption of Specker's principle (or local orthogonality, see Ref.~\cite{ChavesFritz2012,FritzChaves2013,AcinFritzLeverrierSainz2015}) can the necessary structure be reduced to that of minimal projections. Throughout this section, we will therefore assume observable algebras to be (finite-dimensional and finitely generated) partial algebras.

\begin{definition}[realisation]\label{def: realisation}
    Let $G=(V,E)$ be an orthogonality graph and $\cO$ a partial algebra. A \emph{realisation of $G$ in $\cO$} is a map $\pi:G\ra\PO$ such that $\pi(v)\pi(v')=0$ if and only if $(v,v')\in E$. A \emph{quantum realisation of $G$} is a realisation with $\cO\subset\LHsa$.
\end{definition}

Often quantum realisations are restricted to rank-$1$ projections $\pi:G\ra\PHone$, which can equivalently be thought of as maps $\pi:G\ra\C^d$ for $\dim(\cH)=d$. We will generally not make this restriction, yet unlike some authors we will mostly restrict realisations to be \emph{constraint-preserving}, for which $\pi(v)\pi(v')=0$ implies $(v,v')\in E$ (see Tab.~\ref{tab: realisation vs orthomoprhism vs embedding}).\footnote{Other authors define realisations to be neither faithful nor constraint-preserving. As our main interest in orthogonality graphs stems from the characterisation of SI-C sets and SI-C graphs in Sec.~\ref{sec: SI-C sets and graphs}, here we follow the convention in Ref.~\cite{RamanathanHorodecki2014}.\label{fn: realisation}}

\begin{table}[htb!]
    \centering
    \begin{tabular}{c|ccc}
        \toprule
        & unital & faithful & constraint-preserving \\
        \midrule
        realisation & \xmark & \checkmark & \checkmark \\
        orthomorphism & \checkmark & \xmark & \xmark \\
        embedding & \checkmark & \checkmark & \xmark \\
        \bottomrule
    \end{tabular}
    \caption{Comparison between morphisms of orthogonality graphs, event and observable algebras. Realisations are both faithful and constraint-preserving, yet are generally not unital. In contrast, classical embeddings, equivalently faithful orthomorphisms (in terms of functors $\eta$ in Lm.~\ref{lm: embeddings vs orthomorphisms vs order-preserving maps}, this means faithful on objects) are not necessarily constraint-preserving, that is, $\epsilon(p)\epsilon(p')=0$ only if $pp'=0$ (in terms of functors $\eta$ in Lm.~\ref{lm: embeddings vs orthomorphisms vs order-preserving maps}, this means faithful on morphisms).}
    \label{tab: realisation vs orthomoprhism vs embedding}
\end{table}

Observe that orthogonality graphs $G=G(\cO)=G(\POm)$ associated with a partial algebra $\cO$, by definition, admit a realisation where maximal cliques correspond with resolutions of the identity $\One$ in $\cO$. In other words, maximal cliques in such orthogonality graphs are non-trivially related. It is thus natural to constrain the notion of realisation for orthogonality graphs $G=G(\cO)$ such that this additional structure is preserved.

\begin{definition}[unital realisation]\label{def: unital realisation}
    Let $G$ be an orthogonality graph. A realisation $\pi:G\ra\PO$ is called \emph{unital} if for every maximal clique $\rC\subset G$ it holds $\sum_{v\in\rC} \pi(v)=\One$.
\end{definition}

Clearly, not every realisation is unital. While every orthogonality graph has a quantum realisation, it remains an open problem which orthogonality graphs admit a unital quantum realisation \cite{HeunenFritzReyes2014} (see also App.~\ref{app: On realisations}). Consequently, upon restriction of orthogonality graphs to (those associated with) partial algebras, we need to add additional information in the form of the identity element (see also App.~\ref{app: orthogonality graphs vs OAs}).\footnote{In Sec.~\ref{sec: SI-C sets and graphs} below, we discuss this in more detail by introducing a notion of faithful completion (see Def.~\ref{def: faithful graph completion}) with respect to a realisation of an orthogonality graph.} In turn, given a unital realisation we may identify the vertices in an orthogonality graph with the (minimal) generating projections of some partial algebra, that is, we may assume that $G$ is associated with a partial algebra. Indeed, it is easily seen that unital realisations correspond with faithful (and constraint-preserving) orthomorphisms which induce (constraint-preserving) embeddings of partial algebras by Lm.~\ref{lm: embeddings vs orthomorphisms vs order-preserving maps} (see Tab.~\ref{tab: realisation vs orthomoprhism vs embedding}). In Sec.~\ref{sec: consistency with graph-theoretic approach} below, we will therefore first study Kochen-Specker contextuality for orthogonality graphs under the restricting to unital realisations in Def.~\ref{def: unital realisation}.\\

\textbf{Orthogonality graphs as constraints on correlations.} By interpreting orthogonality graphs $G=(V,E)$ as encoding exclusivity relations - representing events that are related if they correspond to different measurement outcomes - a correlation on $G$ is simply defined as a map $\tg:V\ra[0,1]$ that is additive on exclusive events \cite{CabelloSeveriniWinter2014,AcinFritzLeverrierSainz2015,AmaralCunha2018}. General correlations thus constitute the \emph{clique constrained stable set polytope},
\begin{align}\label{eq: qstab}
    \QStab(G) = \{\tg \in \R^{|V|}_+ \mid \sum_{v\in\rC} \tg(v) \leq 1 \quad \forall \mathrm{\ cliques\ } \rC \mathrm{\ in\ G}\}\; .
\end{align}
The inequality in Eq.~(\ref{eq: qstab}) is a consequence of the fact that $G\subset\GO$ is generally only a subgraph for $\cO$ a partial algebra, in particular, $\Gamma[\CO]\subsetneq\QStab(G)$ in this case.\footnote{Since every orthogonality graph $G=(V,E)$ has a quantum realisation on a Hilbert space of dimension $\dim(\cH)\leq |V|$ \cite{HeunenFritzReyes2014}, a partial algebra $\cO\subset\LHsa$ with $G\subset\GO$ always exists (cf. App.~\ref{app: On realisations}).} In turn, it is natural to demand equality for associated orthogonality graphs $G=\GO$, that is, under the restriction to unital realisations. We will at times distinguish unnormalised ``correlations" $\tg\in\QStab(G)$ for $G\subset\GO$ from normalised ``states" $\gamma\in\Gamma[\CO]$ for $G=\GO$, by marking the former with a tilde.

Classical correlations are modelled as probabilistic mixtures of all events that can occur together, that is, that are non-exclusive. These define the \emph{stable set polytope},
\begin{align}\label{eq: stab}
    \begin{split}
    \Stab(G)
    &= \mathrm{conv}\{\tg \in \{0,1\}^{|V|} \mid \tg(v)\tg(v') = 0 \quad \forall\ (v,v') \in E\} \\
    &= \mathrm{conv}\{I(G) \mid I(G) \mathrm{\ independent\ set\ of\ } G\}\; .
    \end{split}
\end{align}
For $G=\GO$, maximal independent sets with support in every maximal clique correspond with valuations, which we denote by $\cI_\mathrm{Val}(\GO)\subset\cI_\mathrm{max}(\GO)$; in particular, $\Gamma_\mathrm{cl}[\CO]\subsetneq\Stab(\GO)$ (for more details on the distinction between states on partial algebras $\cO$ and correlations on $\GO$, see  App.~\ref{app: orthogonality graphs vs OAs}).\footnote{The converse is clearly false: an independent set does generally not define a valuation, since an independent set is not required to assign one element in every maximal clique the value $1$. Consequently, the correlations $\Stab(G)$ are more general than classical correlations for fixed dimension.}

Finally, we consider quantum correlations, that is, those arising from quantum states under the Born rule. More precisely, given an orthogonality graph $G=(V,E)$ with realisation $\pi:G\ra\PH$ and quantum state $\rho\in\SH$ we obtain a correlation of the form $\tg_\rho(v)=\tr[\pi(v)\rho]$ for all $v\in V$. This gives rise to the so-called \emph{theta body},\footnote{Here, realisations are not required to be faithful or constraint-preserving \cite{CabelloSeveriniWinter2014}.}
\begin{align}\label{eq: theta body}
    \Th(G) = \{\tg_\rho \in \R^{|V|}_+ \mid \tg_\rho(v) = \tr[\pi(v)\rho] \mathrm{\ for\ }\rho\in\SH \mathrm{\ and\ }\pi:G\ra\PH \mathrm{\ a\ realisation}\}\; .
\end{align}
For general realisations, $\Gamma_\mathrm{qu}[\CO]\subset\Th(G)$, while equality holds under the restriction to unital realisations (since normalisation in this case follows from the normalisation of quantum states). Finally, for every orthogonality graph $G$ one finds the hierarchy
\begin{align}\label{eq: cl-qu-nd correlations}
    \mathrm{Stab}(G) \subset \Th(G) \subset \QStab(G)\; .
\end{align}
which for orthogonality graphs associated with a partial algebra $G=\GO$ closely resembles Eq.~(\ref{eq: cl-qu-nd states}). Note that the (non-trivial) inclusion $\Stab(G)\subset\Th(G)$ holds since $\Th(G)$ contains realisations of any Hilbert space dimension. Moreover, one can show that the inclusions are strict if and only if $G$ contains an $n$-cycle of length $n\geq 5$ as an induced subgraph \cite{CabelloSeveriniWinter2014}.

To summarise, representing exclusivity relations between (a subset of the) events in a partial algebra via an orthogonality graph encodes the constraints on states of partial algebras in an economic way. Moreover, it allows to bound correlations via graph-theoretic quantities \cite{CabelloSeveriniWinter2014,RamanathanHorodecki2014,CabelloKleinmannBudroni2015}, and to use these bounds in order to construct noncontextuality inequalities \cite{Pitowsky1989_QuantumProbability,BadziagBengtssonCabelloPitowsky2009,AraujoEtAl2013}. Below, we will be concerned with one such quantity: the chromatic number. In order to compare the respective results on contextuality for orthogonality graphs in Ref.~\cite{CabelloSeveriniWinter2014,RamanathanHorodecki2014,CabelloKleinmannBudroni2015} with Thm.~\ref{thm: KSNC = n-colourability}, we first translate the notion of Kochen-Specker contextuality to orthogonality graphs with unital realisation.

\subsection{Kochen-Specker contextuality for orthogonality graphs $G=\GO$}\label{sec: consistency with graph-theoretic approach}

Our aim in this section is to translate the notion of Kochen-Specker (KS) contextuality to the graph-theoretic language. To this end, we cast Def.~\ref{def: marginal noncontextuality} to the present setting.

\begin{definition}\label{def: NCOG}
    An orthogonality graph $G=\GO$ is called \emph{noncontextual} if there exists a normalised\footnote{To be clear, we require that $\gamma$ is normalised in every maximal clique $\rC\subset G$, that is, $\sum_{v\in\rC}\gamma(v)=1$.} correlation $\gamma\in\Stab(G)$ such that $0\neq\gamma(v)\neq\gamma(v')\neq 0$ for all $v,v'\in V$, $v\neq v'$.
\end{definition}

As we point out in detail in App.~\ref{app: orthogonality graphs vs OAs}, the notions of classical correlations $\Stab(G)$ on $G$ and classical states $\Gamma_\mathrm{cl}[\CO]$ on $\cO$ are not related for general orthogonality graphs $G\subsetneq\GO$. We thus restrict (first) to orthogonality graphs $G=\GO$ associated to a finitely generated, not necessarily maximal partial algebra $\cO$. Yet, even in this case classical correlations on $\GO$ are more general than classical states on $\cO$, that is, $\Gamma_\mathrm{cl}[\CO]\subsetneq\Stab(\GO)$. Consequently, the precise relationship between Def.~\ref{def: NCOG} for $\GO$ and (marginal) KS noncontextuality for $\cO$ is not immediately clear.

In order to compare the notion of (marginal) KS noncontextuality with that in Def.~\ref{def: NCOG}, we first show that a KS noncontextual partial algebra admits an optimal classical representation with respect to the subset $\cI_\Val(\GO)\subset\cI_\mathrm{max}(\GO)$ of all maximal independent sets on $\GO$, that is, those that contain a single element in every maximal clique of $\GO$, hence, which correspond to valuations on $\cO$.

\begin{lemma}\label{lm: optimal embedding}
    Let $\cO$ be a finitely generated partial algebra with partial order of contexts $\CO$ and associated orthogonality graph $\GO$. If $\cO$ is Kochen-Specker noncontextual, then there exists a classical embedding $\epsilon:\cO\ra L_\infty(\cI_\Val(\GO))$.
\end{lemma}

\begin{proof}
    If $\cO$ is KS noncontextual, then by definition there exists an embedding $\epsilon: \cO \ra L_\infty(\Lambda)$ for some measurable space $\Lambda$, and since $\cO$ is finitely generated, we may take $\Lambda$ to be a finite set. Consequently, every element $\lambda\in\Lambda$ defines a subset $I(\lambda) = \{p\in\POm\mid\chi(\lambda\mid p)=\epsilon(p)(\lambda)=1\}$, and we have $I(\lambda)\in\cI_\Val(\GO)$ as $\lambda$ singles out a unique projection in every maximal context $C\in\COm$: indeed, in every partition $\Lambda = \dot\cup_{p\in\cP_\mathrm{min}(C)} \Lambda_p$ with $\Lambda_p=\mathrm{supp}(\epsilon(p))$ we have $p\in I(\lambda)$ if and only if $\lambda\in\Lambda_p$.

    In general, different elements $\lambda,\lambda'\in\Lambda$, $\lambda\neq\lambda'$ (under the classical embedding $\epsilon:\cO\ra L_\infty(\Lambda)$) will give rise to the same independent set $I(\lambda)=I(\lambda')$. We define a more economical map $\epsilon'$ by composing $\epsilon$ with the ``coarse-graining" map $\Lambda\ra\cI_\Val(\GO)$, that is, by linear extension of $\epsilon'(p)(I(\lambda))=\epsilon(p)(\lambda)=1$ if and only if $p\in I(\lambda)\in\cI_\Val(\GO)$. In other words, $\mathrm{supp}(\epsilon'(p))=\cI(p):=\{I(\lambda)\mid \exists\lambda\in\Lambda:\ p\in I(\lambda)\}$. Denote the image under this coarse-graining map by $\cI(\epsilon)\subset\cI_\mathrm{val}(\GO)$. Clearly, we have $\epsilon'(\One)=\dot\cup_{p\in\cP_\mathrm{min}(C)}\cI(p)=\cI(\epsilon)$ (since every independent set $I\in\cI(p)\subset\cI(\epsilon)$ contains one element in every context $C\in\COm$). Moreover, $\cI(p)\cap\cI(p')=\emptyset$ whenever $p,p'\in\POm$, $pp'=0$ since no independent set $I\in\cI(\GO)$ contains both $p$ and $p'$. $\epsilon'|_\POm$ thus defines a faithful orthomorphism which induces a classical embedding by Lm.~\ref{lm: embeddings vs orthomorphisms vs order-preserving maps}.
\end{proof}

Note that Lm.~\ref{lm: optimal embedding} is closely related with Lm.~\ref{lm: embeddings vs orthomorphisms vs order-preserving maps} which shows that the existence of a classical embedding for $\cO$ is equivalent to the existence of a faithful orthomorphism of $\POm$, and hence, that KS contextuality is a property of the latter.

Under the embedding in Lm.~\ref{lm: optimal embedding}, every classical state $\gamma\in\Gamma_\mathrm{cl}[\CO]$ thus defines a (normalised) classical correlation on $\GO$, that is, $\gamma\in\Stab(\GO)$. This brings us one step closer to comparing (marginal) KS noncontextuality of $\cO$ with the notion of noncontextuality of its associated orthogonality graph $\GO$. 
Yet, the converse does not hold, that is, $\Gamma_\mathrm{cl}[\CO]\subsetneq\Stab(\GO)$ since $\cI_\Val(\GO)\subsetneq\cI(\GO)$, in particular, maximal independent sets of $\GO$ are not necessarily valuations. Still, the following theorem shows that the respective notions of contextuality coincide if $G=G(\cO)$.

\begin{theorem}\label{thm: KSNC = NCOG}
    Let $\cO$ be a finitely generated partial algebra with associated orthogonality graph $G=\GO$. Then $\cO$ is Kochen-Specker contextual if and only if $G$ is  contextual.
\end{theorem}

\begin{proof}
    If $\cO$ is KS noncontextual, then there exists a classical embedding $\epsilon:\cO\ra L_\infty(\Lambda)$, which by Lm.~\ref{lm: optimal embedding} reduces to a classical embedding $\epsilon:\cO\ra L_\infty(\cI_\Val(\GO))$ under the coarse-graining $\Lambda\ra\cI_\Val(\GO)\subset\cI_\mathrm{max}(\GO)$. Now, since $\epsilon$ is faithful, there exists a separating state $\gamma\in\Gamma_\mathrm{cl}[\CO]\subset\Stab(\GO)$ with full support on $\GO$, that is, such that $0\neq\gamma(v)\neq\gamma(v')\neq 0$ for all $v,v'\in V$, $v\neq v'$. Hence, $G$ is noncontextual by Def.~\ref{def: NCOG}.
    
    Conversely, if $G=\GO$ is noncontextual by Def.~\ref{def: NCOG}, then there exists a state (a normalised correlation) $\gamma\in\Stab(G)$, that is, for all $v\in V=\POm$
    \begin{align}\label{eq: stab rep}
        \gamma(v)
        = \sum_{I\in\cI(G)} \rw_\gamma(I) I(v)
        = \sum_{I\in\cI(v)} \rw_\gamma(I)\; ,
    \end{align}
    where $\rw_\gamma(I)\geq 0$ for all $I\in\cI(G)$ with $\sum_{I\in\cI(G)} \rw_\gamma(I) = 1$, and where we defined $\cI(v)=\{I\in\cI(G)\mid I(v)=1\}$. Clearly, we have $\cI(v)\cap\cI(v')=\emptyset$ whenever $v,v'\in \rC$ with $v\neq v'$ for any clique $\rC\subset\GO$. Setting $\cI(\emptyset)=\emptyset$ we thus obtain orthomorphisms $\varphi_C:\PC\ra 2^{\cI(\rC)}$ defined additively on minimal projections by $\varphi_C(v)=\cI(v)$ (recall that $V=\POm$), and where $\cI(\rC) := \dot\bigcup_{v\in\rC} \cI(v) = \{I\in\cI(G)\mid\exists v\in\rC:\ I(v)=1\}$ for every context $C\in\cC(\cO)$ (equivalently, for every clique $\rC\subset G$). Now, we would like to define an orthomorphism from the collection of orthomorphisms $\{\varphi_C\}_{C\in\CO}$. For this to be the case we must have $\varphi_C(p)=\varphi_{C'}(p)$ for all $C,C'\in\CO$ with $p\in\PC,\cP(C')$, yet since $\cI(\rC)\neq\cI(\rC')$ for different (maximal) cliques, this is generally not the case.
    
    Nevertheless, since $\gamma$ is normalised (by assumption) and using Eq.~(\ref{eq: stab rep}) for every maximal clique $\rC\in G$, we find
    \begin{align*}
        1 = \sum_{v\in\rC} \gamma(v)
        = \sum_{v\in\rC} \sum_{I\in\cI(v)} \rw_\gamma(I)
        = \sum_{I\in\cI(\rC)} \rw_\gamma(I)\; .
    \end{align*}
    Since also $\sum_{I\in\cI(G)} \rw_\gamma(I)=1$ this implies that $\cI_\gamma(\rC)=\cI(G)/\cI_{\gamma,0}(G)=:\Lambda_\gamma$, where $\cI_{\gamma,0}(G)=\{I\in\cI(G)\mid \rw_\gamma(I)=0\}$ depends on the correlation $\gamma$, but is independent of the clique $\rC\subset G$. We may thus define a new collection of orthomorphisms $\{\varphi_{\gamma,C}\}_{C\in\COm}$, $\varphi_{\gamma,C}:\PC\ra 2^{\cI_\gamma(\rC)}$ by $\varphi_{\gamma,C}(v)=\cI_\gamma(v):=\cI(v)/\cI_{\gamma,0}(G)$. Clearly, we still have $\cI_\gamma(v)\cap\cI_\gamma(v')=\emptyset$ whenever $v,v'\in\rC$. Moreover, $\varphi_{\gamma,C}(\One)=\varphi_{\gamma,C'}(\One)$ for all $C,C'\in\CO$, which implies that $\varphi_\gamma:\PO\ra\cB(\Lambda_\gamma)=2^{\Lambda_\gamma}$, $\varphi_\gamma(p)=\cI_\gamma(p)$ defines an orthomorphism of event algebras. From $\cI_\gamma(\rC)=\cI_\gamma(\rC')$ for all maximal cliques $\rC,\rC'\subset G$, it follows that $\gamma(I)\neq 0$ only if $I\in\cI_\mathrm{val}(G)$, and since by Def.~\ref{def: NCOG} we may assume $0\neq\gamma(v)\neq\gamma(v')\neq 0$ whenever $v,v'\in V$, $v\neq v'$, we conclude that $\varphi_\gamma:\PO\ra\cB(\Lambda_\gamma)$ is a faithful orthomorphism of event algebras, which induces a classical embedding by Lm.~\ref{lm: embeddings vs orthomorphisms vs order-preserving maps}.
\end{proof}

Thm.~\ref{thm: KSNC = NCOG} proves that Def.~\ref{def: NCOG} is indeed a reformulation of KS noncontextuality for orthogonality graphs associated with observable algebras, that is, in cases where both notions are comparable. We highlight two important steps in the proof of Thm.~\ref{thm: KSNC = NCOG}.

First, the converse direction necessarily requires that equality holds in Eq.~(\ref{eq: qstab}), which is natural if $G=\GO$ is associated with (equivalently, has a unital realisation in) a partial algebra $\cO$. In turn, if we consider general realisations $\pi:G\ra\cP_\mathrm{min}(\cO')$, we have to add further vertices in order to obtain equality in Eq.~(\ref{eq: qstab}), that is, we need to consider (correlations on) larger graphs $G\subset G'=G(\cO')$. Yet, by Prop.~\ref{prop: no marginalisation} in App.~\ref{app: orthogonality graphs vs OAs}, $\Stab(G')|_G\neq\Stab(G)$ for $G\subset G'$, that is, the respective correlations in $G$ and $G'$ are not naturally related, in particular, they generally depend on the realisation $\pi$.

Second, note that Thm.~\ref{thm: KSNC = NCOG} is very similar to Thm.~\ref{thm: KSNC = marginal KSNC}. In particular, the second assertion in Thm.~\ref{thm: KSNC = NCOG} relies on the existence of a separating classical state with full support on an orthogonality graph. This assumption cannot be omitted without spoiling the equivalence in general: while we can construct the orthomorphism $\varphi_\gamma$ in the proof of Thm.~\ref{thm: KSNC = NCOG} for any (normalised) state $\gamma\in\Gamma_\mathrm{cl}[\CO]\subset\Stab(\GO)$, by a similar argument to the one in the proof of Lm.~\ref{lm: solutions to marginal problem do not imply KSNC} in Sec.~\ref{sec: characterising KSNC using marginal approach}, it will generally not be faithful.

Next, we show that contextual orthogonality graphs are characterised by their chromatic number. For ease of notation, we introduce the following definition.

\begin{definition}\label{def: maximal extended OG}
    Let $G$ be an orthogonality graph with unital realisation $\pi:G\ra\POm$, and let $\cO^*$ be a maximal extension of the partial algebra $\cO$ generated by $\pi(G)$. Then the \emph{maximal extension $G^*$ of $G$ with respect to $\cO^*$ (and thus $\pi$)} is defined as $G^*=\cP_1(\cO^*)$.
\end{definition}

\begin{corollary}\label{cor: NCOG=n-colourability}
    Let $\cO$ be a partial algebra with orthogonality graph $G=\GO$ and $\dim(\One)=d$. Then $G$ is contextual if and only if $\chi(G^*)>d$.
\end{corollary}

\begin{proof}
    This follows from Thm.~\ref{thm: KSNC = NCOG} together with Thm.~\ref{thm: n-colourability for non-maximal OAs}.
\end{proof}

The restriction to orthogonality graphs $G=\GO$ associated with partial algebras simplifies the characterisation of KS contextuality, yet it leaves open the problem of finding a (quantum) realisation for $G$, equivalently of characterising quantum among more general partial algebras. While this is a fascinating and likely hard problem in general, in many applications one is given a subset of quantum observables $\cO\subset\LHsa$ with $\dim(\cH)=d$ to begin with, and the problem thus reduces to determining whether $\cO$ is contextual or not. The study of state-independent contextuality (SI-C) sets which we turn to next is concerned with this more restricted problem.

\subsection{State-independent contextuality (SI-C) sets and graphs}\label{sec: SI-C sets and graphs}

The notion of (non)contextuality in Def.~\ref{def: NCOG} is not commonly used in the graph-theoretic approach. Instead, one finds a notion similar to the one in Def.~\ref{def: fully classically correlated}, which is motivated by the study of noncontextuality inequalities \cite{Pitowsky1991,BadziagBengtssonCabelloPitowsky2009,KleinmannEtAl2012}. Similar to Bell inequalities \cite{Bell1964,BrunnerEtAl2014,Horodeckisz2009,JungeEtAl2010}, the violation of a noncontextuality inequality is generally state-dependent \cite{KlyachkoEtAl2008,LiangSpekkensWiseman2011}. Yet, while the violation of a Bell inequality is necessarily state-dependent there exist noncontextuality inequalities that are violated in every state \cite{Cabello2008b,YuOh2012,BengtssonBlanchfieldCabello2012,XuChenSu2015,LeiferDuarte2020}. This has motivated the following definition.

\begin{definition}[SI-C set]\label{def: SI-C set}
    A set of projectors $\{\pi(v)\in\PH\}_{v\in V}$, given in terms of a quantum realisation $\pi:G\ra\PH$ of an orthogonality graph in dimension $\dim(\cH)=d$, is called a \emph{state-independent contextuality (SI-C) set} if $\tg_\rho\notin\Stab(G)$ for all $\rho\in\SH$, where $\tg_\rho(v)=\tr[\pi(v)\rho]$ for all $v\in V$.\footnote{We follow the notation in Ref.~\cite{CabelloKleinmannBudroni2015}, which avoids confusion with the commonly used abbreviation `SIC POVM' for symmetric, informationally complete positive operator-valued measure \cite{FuchsHoangStacey2017}.}\footnote{While $\tg$ is generally not a state, that is, it is not normalised in every clique, $\tg$ necessarily extends to a state on the completion $\overline{\pi(V)}$ with respect to $\pi$ (see Eq.~(\ref{eq: complement completion}) below).}
\end{definition}

SI-C sets play an important role in experimental tests of contextuality \cite{SimonEtAl2000,MichlerWeinfurterZukowski2000,Cabello2008a,Cabello2008b,BartosikEtAl2009,KirchmairEtAl2009,AmselemEtAl2009,GuehneEtAl2010,LapkiewiczEtAl2011}. However, characterising SI-C sets is a hard problem in general \cite{CabelloKleinmannBudroni2015}. A simplification to orthogonality graphs was proposed in Ref.~\cite{RamanathanHorodecki2014}.

\begin{definition}[SI-C graph]\label{def: SI-C graph}
    An orthogonality graph $G$ is called \emph{(state-dependently) contextual in dimension $d$} if there exists a quantum state $\rho\in\SH$ with $\mathrm{dim}(\cH)=d$ such that $\tg_\rho(v) = \tr[\pi(v)\rho]\notin\mathrm{STAB}(G)$ for some realisation $\pi:G\ra\PHone$ of $G$.
    
    Moreover, $G$ is \emph{state-independently contextual (SI-C) in dimension $d$} if it is (state-dependently) contextual in dimension $d$ for every quantum state $\rho\in\SH$.
\end{definition}

Our present concern with SI-C sets (and graphs) is that the discovery of SI-C sets has raised some concern over the true essence of Kochen-Specker (KS) contextuality. We recall that many proofs of the KS theorem construct so-called KS sets \cite{KochenSpecker1967,Mermin1990,Peres1991,Mermin1993,ZimbaPenrose1993,Kernaghan1994}, that is, sets of projectors for which it is impossible to single out one in every set of mutually orthogonal ones (see Sec.~\ref{sec: KS-colouring vs colouring}). Yet, while every KS set is a SI-C set \cite{BadziagBengtssonCabelloPitowsky2009,YuTong2014,YuGuoTong2015}, not every SI-C set is also a KS set \cite{YuOh2012,BengtssonBlanchfieldCabello2012,Cabello2012b}. In Ref.~\cite{Frembs2024a}, we argue that proofs of contextuality based on (the violation of noncontextuality inequalities corresponding to) SI-C sets are nevertheless compatible with the notion of Kochen-Specker contextuality. In the following, we will provide a more detailed comparison.

To begin with, observe that existing (partial) characterisations of SI-C sets and SI-C graphs somewhat resemble Thm.~\ref{thm: n-colourability for non-maximal OAs} which implies Cor.~\ref{cor: NCOG=n-colourability}. Specifically, Ref.~\cite{Cabello2012b,CabelloKleinmannBudroni2015} argue that $\chi(G)>d$ is a necessary condition for the realisation of an orthogonality graph $G$ to be SI-C set (note the necessary restriction on realisations discussed below and in App.~\ref{app: correction}), and Ref.~\cite{RamanathanHorodecki2014} further proves that for an orthogonality graph $G$ to be SI-C in dimension $d$ is equivalent to $\chi_f(G)>d$, where $\chi_f(G)$ denotes the fractional chromatic number of $G$. Yet, despite the similarity of these results with Thm.~\ref{thm: n-colourability for non-maximal OAs} and Cor.~\ref{cor: NCOG=n-colourability}, we point out various subtleties involved with the underlying definitions. First, unlike Def.~\ref{def: NCOG}, Def.~\ref{def: SI-C set} and Def.~\ref{def: SI-C graph} are concerned with the reduced marginal problem of quantum states with respect quantum realisations in a given dimension. Second, unlike the realisations in Cor.~\ref{cor: NCOG=n-colourability}, the realisations of SI-C graphs (in particular, SI-C sets) are generally not unital. And, third, SI-C graphs are not simply the orthogonality graphs of SI-C sets: there are SI-C graphs which do not have a realisation that is a SI-C set (see Thm.~1 in Ref.~\cite{CabelloKleinmannBudroni2015}).\footnote{Since the two notions coincide whenever the quantifiers over realisations and states in Def.~\ref{def: SI-C set} and Def.~\ref{def: SI-C graph} commute, this also shows that this is generally not the case (see Sec.~\ref{sec: SI-C sets and graphs}).} The following subsections (together with App.~\ref{app: correction}) will address these issues in turn, beginning with the case of SI-C graphs associated with partial algebras.

\subsubsection{Unital realisations}\label{sec: associated SI-C graphs}

We consider Def.~\ref{def: SI-C graph} under the restriction to orthogonality graphs $G=\GO$ associated to a partial algebra $\cO$. As discussed in Sec.~\ref{sec: basics of graph-theoretic approach}, in order to preserve the additional structure imposed by the identity $\One$ in $\cO$ one needs to require that quantum realisations map maximal cliques in $G$ to resolutions of $\One$. Consequently, for orthogonality graphs $G=\GO$, we modify Def.~\ref{def: SI-C graph} by restricting to \emph{unital} quantum realisations.

\begin{definition}\label{def: USI-C graph}
    An orthogonality graph $G$ is \emph{unital (state-dependently) contextual in dimension $d$} if there exists a state $\rho\in\SH$ with $d=\mathrm{dim}(\cH)$ such that $\gamma_\rho(v) = \tr[\pi(v)\rho]\notin\mathrm{STAB}(G)$ for some unital quantum realisation $\pi:G\ra\PH$ of $G$.
    
    Moreover, $G$ is \emph{unital state-independently contextual (USI-C) in dimension $d$} if it is unitarily (state-dependently) contextual in dimension $d$ for every state $\rho\in\SH$.
\end{definition}

By Lm.~\ref{lm: embeddings vs orthomorphisms vs order-preserving maps}, a unital quantum realisation of $\GO$ is equivalent to a constraint-preserving quantum embedding of $\cO$ (see also Tab.~\ref{tab: realisation vs orthomoprhism vs embedding}). Under this stronger assumption, we can employ the characterisation of KS contextuality in Thm.~\ref{thm: n-colourability for non-maximal OAs}.

\begin{theorem}\label{thm: n-colourable implies USI-C}
    Let $G$ be an orthogonality graph with unital quantum realisation $\pi:G\ra\PH$ in $\dim(\cH)=d$. Then $G$ is unital state-independently contextual (USI-C) in dimension $d$ only if $\chi(G^*)>d$.
\end{theorem}

\begin{proof}
    By assumption, $G=\GO$ with $\cO\subset\LHsa$ the partial algebra induced by the unital quantum realisation, that is, $\POm=\pi(G)$. Therefore, if $\chi(G^*)=d$ then $\cO$ is KS noncontextual by Thm.~\ref{thm: n-colourability for non-maximal OAs}, that is, there exists a classical embedding $\epsilon:\cO\ra L_\infty(\Lambda)$. Moreover, by Lm.~\ref{lm: min extension to max observable algebra}, there exists a classical embedding $\epsilon^*:\cO^*\ra L_\infty(\Lambda)$, and $G^*=G(\cO^*)$ has a unital realisation $\pi^*:G^*\ra\PHone$ such that $\pi^*|_G=\pi$. (Indeed, since $G$ is finite, we can always decompose projections $\pi(v)=\sum_{i=1}^{\dim(\pi(v))}\pi^*(v^*_i)$ such that no new orthogonality relations arise, that is, $\pi^*(v^*_i)\pi^*(v^*)=0$ if and only if $v^*_i,v^*\in\rC^*$ and $v_i^*\neq v^*$.) By Lm.~\ref{lm: optimal embedding}, we may further assume that $(\epsilon^*)':\cO^*\ra \cI_\mathrm{val}(G^*)$, and since $\chi(G^*)=d$, there exists a $d$-colouring $\zeta:V^*\ra[d]$. In particular, each colour picks out a unique valuation in $\cI(\zeta)=\{I_1=\zeta^{-1}(v_1),\cdots,I_d=\zeta^{-1}(v_d)\}\subset\cI_\mathrm{val}(G^*)$ for the set of vertices $\{v_1,\cdots,v_d\}=\rC^*$ in any maximal clique $\rC^*\subset G^*$. The valuations in $\cI(\zeta)$ satisfy $I_i\cap I_j=\emptyset$ whenever $i\neq j$ and $\dot\bigcup_{i=1}^d I_i=V^*$. Now, define the classical state $\gamma^*\in\Gamma_\mathrm{cl}[\cC(\cO^*)]$ by $\gamma^*(I)=\frac{1}{d}$ if $I\in\cI(\zeta)$ and $\gamma^*(I)=0$ otherwise. Clearly, we have $\gamma^*=\one/d|_{\pi^*(G^*)}$, hence, by setting $\gamma(v)=\sum_{i}\gamma(v^*_i)$ for all $v\in V$ we also have $\gamma(v)=\tr[\pi(v)\one/d]\in\Gamma_\mathrm{cl}[\CO]$. Hence, $G$ is not a USI-C graph in dimension $d$.
\end{proof}

\begin{corollary}\label{cor: unital SI-C set}
    A necessary condition for $\pi(V)$ for $\pi:G\ra\PH$ a unital quantum realisation of an orthogonality graph $G$ in dimension $\dim(\cH)=d$ to be a state-independent contextuality (SI-C) set is that $\chi(G^*)>d$.
\end{corollary}

\begin{proof}
    The orthogonality graph of a SI-C set is always a SI-C graph.
\end{proof}

The converse implications of Thm.~\ref{thm: n-colourable implies USI-C} and Cor.~\ref{cor: unital SI-C set} are generally false, since a quantum state $\rho\in\SH$ may have a classical representation of the form $\gamma(v)=\tr[\pi(v)\rho]\in\Gamma_\mathrm{cl}[\CO]$, yet one for which $\gamma$ is not separating and of full support (cf. Lm.~\ref{lm: solutions to marginal problem do not imply KSNC} in Sec.~\ref{sec: characterising KSNC using marginal approach}). The notions of USI-C graph $G=\GO$, SI-C sets with unital realisation and KS contextuality of $\cO$ therefore generally differ. What about more general realisations.

Thm.~\ref{thm: n-colourable implies USI-C} and Cor.~\ref{cor: unital SI-C set} generalise the results in Ref.~\cite{Cabello2012b,RamanathanHorodecki2014,CabelloKleinmannBudroni2015} in the case of a unital realisations. Yet, what about more general realisations? In the next section, we further extend our results by lifting a non-unital realisation $\pi:G\ra\PO$ of an orthogonality graph $G$ to the unital realisation of its completion in $\cO$.

\subsubsection{Realisations with faithful completion}\label{sec: general SI-C graphs}

We seek to generalise Thm.~\ref{thm: n-colourable implies USI-C} to orthogonality graphs with not necessarily unital realisation. Note that the realisations $\pi$ in Def.~\ref{def: SI-C set} and Def.~\ref{def: SI-C graph} are not assumed to be unital. Still, the correlations one is interested in are those generated by $\pi$ in $\PH$:
\begin{align}\label{eq: complement completion}
    \overline{\pi(V)}
    =\langle\pi(V),\mathbbm{1}\rangle\; ,
\end{align}
where the \emph{completion $\overline{\pi(V)}$} contains all projections that can be written as sums of elements in $\{\pi(V),\one\}$, and recursively from newly generated elements; in particular, $\{\op(\rC)=\mathbbm{1}-\sum_{v\in\rC}\pi(v)\mid \rC\subset G\mathrm{\ a\ maximal\ clique}\}\subset \overline{\pi(V)}$, moreover $\overline{\pi(V)}$ may contain new contexts generated by such elements. Indeed, the observables violating a noncontextuality inequality corresponding to a SI-C set necessarily include the identity element \cite{Cabello2008b,YuOh2012,BadziagBengtssonCabelloPitowsky2009}. This is most evident in noncontextuality inequalities involving dichotomic observables, whose quantum violation stems from measurements of the form $O_v=\mathbbm{1}-2\pi(v)$ for all $v\in V$. In this sense, SI-C sets correspond with the full partial algebra $\cO(\pi(V))=\cO(\overline{\pi(V)})$ induced by the event algebra $\overline{\pi(V)}=\cP(\cO(\pi(V)))$, generated by $\pi(V)$ in $\PH$. Define $\oG:=G(\cO(\pi(V)))$, then $\pi$ has an extension to a unital quantum realisation $\overline{\pi}:\oG\ra\PH$ such that $\overline{\pi}|_G=\pi$.\footnote{Note that $G$ is always a vertex-restricted subgraph of its completion $\oG$ (see App.~\ref{app: orthogonality graphs vs OAs}).} Clearly, the orthogonality graph $\oG$ depends on the realisation $\pi$ of $G$. We therefore first establish some relevant facts about completions induced by general realisations of orthogonality graphs.\\

\textbf{Faithful completions.} What is the difference between Thm.~\ref{thm: n-colourable implies USI-C} and the results in Ref.~\cite{Cabello2012b,RamanathanHorodecki2014,CabelloKleinmannBudroni2015}? The latter consider the chromatic number of the subgraph $G\subset\oG$, where $\oG$ denotes the orthogonality graph of its completion $\overline{\pi(V)}$ with respect to the realisation $\pi$ in Eq.~(\ref{eq: complement completion}). Now, it is generally not true that $\chi(G)<d$ implies $\chi(\oG)<d$. Indeed, the completion $\oG$ with respect to a realisation $\pi$ will generally introduce new constraints between the vertices in $G$. Explicitly, write $\ov_\pi(\rC)$ for the vertex in $\oG$ corresponding to the projection $\op(\rC)=\One-\sum_{v\in\rC}\pi(v)$ in the context realisation of the maximal clique $\rC\subset G$. Then we may have $\ov_\pi(\rC)=\ov_\pi(\rC')$ for two distinct maximal cliques $\rC,\rC'\subset G$, $\rC\neq\rC'$. Moreover, two vertices $\ov_\pi(\rC)$ and $\ov_\pi(\rC')$ may be orthogonal and thus pose additional constraints on a colouring of $\oG$. Consequently,
$\chi(\oG)>\chi(G)$ in general (even if $\omega(\oG)=\omega(G)$). For example, for a uniquely colourable graph with two non-maximal cliques $\rC,\rC'\subset G$ which jointly contain all colours in $\chi(G)$, together with a realisation such that $\ov_\pi(\rC)=\ov_\pi(\rC')$, one finds $\chi(\oG)>\chi(G)$. The following restriction to faithful completions avoids this issue.

\begin{definition}\label{def: faithful graph completion}
    A realisation $\pi:G\ra\PO$ of an orthogonality graph $G$ is called \emph{freely completable} and the completion $\overline{\pi(V)}$ \emph{faithful} if $\op(\rC)\neq\op(\rC')$ for any two distinct maximal cliques $\rC,\rC'\subset G$, $\rC\neq\rC'$ and $\op(\rC)\pi(v)=0$ only if $v\in\rC$.
\end{definition}

For instance, the completion of the realisation in Fig.~2 in Ref.~\cite{YuOh2012} of the Yu-Oh orthogonality graph (see also Fig.~\ref{fig: Yu-Oh graph} in App.~\ref{app: correction}) is faithful. Faithful completions of an orthogonality graph $G$ depend on the realisation only via the dimension $\dim(\One)=d$ of the target partial algebra; we therefore write $\CG$ in this case. As an aside, we mention that the existence of \emph{any} realisation $\pi:G\ra\PO$ of $G$ with $\dim(\One)=d$ implies the existence of a realisation with faithful completion $\CG$ in the same dimension.

\begin{lemma}\label{lm: existence of faithful graph completion}
    Let $G$ be an orthogonality graph with realisation $\pi':G\ra\PO$ in dimension $\dim(\One')=d$. Then $G$ has a realisation with $\dim(\One)=d$ and faithful completion $\CG$.
\end{lemma}

\begin{proof}
    Since $G$ has a realisation $\pi':G\ra\cP(\cO')$ with $\dim(\One')=d$, there exists a dimension function $\dim':\cP(\cO')\ra\mathbb{N}$ on $\cO'$ (see Def.~\ref{def: dim}) with $\dim(\One')=\dim'(\ov_{\pi'}(\rC))+\sum_{v\in\rC}\dim'(v)$ for all cliques $\rC\subset G$ and $\ov_{\pi'}(\rC)$ the unique vertex corresponding to $\op(\rC)$ in its completion in Eq.~(\ref{eq: complement completion}). Now, let $\pi:G\ra\PO$ be a realisation with faithful completion of $G$ as by Def.~\ref{def: faithful graph completion}. Then $\dim'$ induces a dimension function $\dim:\PO\ra\mathbb{N}$ on $\cO$ by setting $\dim|_V=\dim'|_V$ and $\dim(\ov_\pi(\rC))=\dim'(\ov_{\pi'}(\rC))$. In particular, we have
    \begin{align*}
        \hspace{.6cm}
        \dim(\One)
        = \dim(\ov_\pi(\rC))+\sum_{v\in\rC}\dim(v)
        = \dim'(\ov_{\pi'}(\rC))+\sum_{v\in\rC}\dim'(v)
        = \dim(\One') = d\; .
        \hspace{1.2cm} \qedhere
    \end{align*}
\end{proof}

Next, we relate the chromatic number $\chi(G)$ of an orthogonality graph $G\subset\oG$ with that of the minimal extension of its completion $\chi(\oG^*)$ with respect to some realisation $\pi:G\ra\PO$. Clearly, a $d$-colouring of $\oG^*$ induces a $d$-colouring of $G$ by restriction. In other words, $\chi(G)\leq\chi(\oG^*)$. However, the converse is generally not true. However, under the restriction to freely completable realisations with one-dimensional projections the following relation holds.

\begin{lemma}\label{lm: restricted chromatic extension}
    Let $G$ be an orthogonality graph with freely completable realisation $\pi:G\ra\POone$ in dimension $\dim(\One)=d$. Then $\chi(G)\leq d$ if and only if $\chi(G^*)\leq d$.
\end{lemma}

\begin{proof}
    Clearly, $\chi(G)\leq\chi(G^*)$ since a colouring $\zeta^*:V^*\ra[d]$ yields a colouring $\zeta^*|_V:V\ra[d]$ for $G$ by restriction, hence, $\chi(G^*)\leq d$ implies $\chi(G)\leq d$.
    
    Conversely, $G$ arises from $G^*=G(\cO^*)$ by removing in (some) maximal cliques $\rC\subset\rC^*\subset G^*$ the vertices $\ov_{\pi^*,i}(\rC)$ corresponding to the one-dimensional projections $p^*_i(\rC)$ in the maximal extension $\cO^*$ of $\cO=\cO(\pi(V))$, which (by Def.~\ref{def: max extension}) satisfy $p^*_i(\rC)p^*_j(\rC')=0$ if and only if $\rC=\rC'$ and $i\neq j$ as well as $\op(\rC)=\one-\sum_{v\in\rC}\pi(v)=\sum_{i=1}^{d-|\rC|}p^*_i(\rC)$, where $\op(\rC)$ is the projection with corresponding vertex $\ov_\pi(\rC)$ in the faithful completion of $G$. Moreover, by Def.~\ref{def: faithful graph completion}, $\ov_\pi(\rC)\neq\ov_\pi(\rC')$ whenever $\rC\neq\rC'$ and $(\ov_\pi(\rC),v)\in E^*$ if and only if $v\in\rC$. Consequently, given a $d$-colouring $\zeta:V\ra[d]$ for $G$, we obtain a valid $d$-colouring $\zeta^*:V^*\ra[d]$ for $G^*$ by assigning to $\{\ov_{\pi^*,i}(\rC)\}_i$ the colours not yet assigned to any other vertex in $\rC$, that is, such that $\zeta^*(\ov_{\pi^*,i}(\rC))\neq\zeta(v)$ for all $v\in\rC$ and $i\in\{1,\cdots,d-|\rC|\}$. It follows that $\chi(G)\leq d$ implies $\chi(G^*)\leq d$.
\end{proof}

With the notion of freely completable realisations at hand, we obtain the following partial characterisation of SI-C graphs in terms of their chromatic number.\\

\textbf{A necessary condition for SI-C graphs and SI-C sets.}

\begin{theorem}\label{thm: partial RH conjecture}
    Let $G$ be an orthogonality graph with freely completable realisation $\pi:G\ra\PH$ in $\dim(\cH)=d$ and let $\oG$ be the orthogonality graph of its faithful completion with respect to $\pi$. Then $\oG$ is a SI-C graph only if $\chi(\CG^*)>d$.

    Moreover, if $\pi:G\ra\PHone$ then $\oG$ is a SI-C graph only if $\chi(G)>d$.
\end{theorem}

\begin{proof}
    The first part follows from Thm.~\ref{thm: n-colourable implies USI-C} since $\overline{\pi}:\oG\ra\PH$ is a unital realisation and since $\oG=\CG$, by assumption. The second assertion follows from Lm.~\ref{lm: restricted chromatic extension}.
\end{proof}

Clearly, this also implies a partial characterisation of SI-C sets.

\begin{corollary}\label{cor: generalised CKB}
    A necessary condition for the faithful completion $\overline{\pi(V)}$ of a (freely completable) quantum realisation $\pi:G\ra\PH$ of an orthogonality graph $G$ in dimension $\dim(\cH)=d$ to be a state-independent contextuality (SI-C) set is that $\chi(\CG^*)>d$.

    If moreover $\pi:G\ra\PHone$, then $\overline{\pi(V)}$ is a SI-C set only if $\chi(G)>d$.
\end{corollary}

\begin{proof}
    Since the orthogonality graph of a SI-C set is a SI-C graph, the result follows immediately from Thm.~\ref{thm: partial RH conjecture}.
\end{proof}

Importantly, the statements in Thm.~\ref{thm: partial RH conjecture} and Cor.~\ref{cor: generalised CKB} do not hold without assuming (faithful) completions. Since this creates a conflict with the results in Ref.~\cite{RamanathanHorodecki2014,CabelloKleinmannBudroni2015}, we elaborate on it in more detail in App.~\ref{app: correction}. In turn, under the restriction to freely completable realisations, Thm.~\ref{thm: partial RH conjecture} and Cor.~\ref{cor: generalised CKB} generalise Thm.~4 in Ref.~\cite{CabelloKleinmannBudroni2015} which is also implied by Thm.~2 in Ref.~\cite{RamanathanHorodecki2014}. Finally, we remark that our proofs of Thm.~\ref{thm: partial RH conjecture} and Cor.~\ref{cor: generalised CKB} are independent from the ones given in Ref.~\cite{CabelloKleinmannBudroni2015} and Ref.~\cite{RamanathanHorodecki2014}, as we derived the characterisation of state-independent contextuality (in terms of the chromatic number of an orthogonality graph) from the algebraic characterisation of KS noncontextuality in Thm.~\ref{thm: CNC for OAs}, in particular, its reformulation in Thm.~\ref{thm: KSNC = n-colourability} and Thm.~\ref{thm: n-colourability for non-maximal OAs}.\\

We finish this section by briefly commenting on the distinction between SI-C graphs and SI-C sets. Note that SI-C sets and SI-C graphs are distinguished by the order of quantifiers over realisations and states in Def.~\ref{def: SI-C set} and Def.~\ref{def: SI-C graph}, which in general do not commute due to the existence of classical states: clearly, if $\SH|_{\Stab(\oG)}=\Gamma_\mathrm{qu}[\cO(\pi(V))]\cap\Stab(\oG)\neq\emptyset$ then $\overline{\pi(V)}$ is not a SI-C set. Nevertheless, there exist SI-C graphs with $\SH|_{\Stab(\oG)}\neq\emptyset$, for which therefore no realisation is a SI-C set  \cite{CabelloKleinmannBudroni2015}. Such examples arise since Def.~\ref{def: SI-C graph} allows to choose different realisations (hence, different noncontextuality inequalities) for different quantum states, which renders ``SI-C" graphs \textit{``state-independent on an operational level"} \cite{CabelloKleinmannBudroni2015}. We can see explicitly how the notion of SI-C graph applies to KS contextual partial algebras with $\SH|_{\Stab(\oG)}\neq\emptyset$. Namely, Ref.~\cite{RamanathanHorodecki2014} introduce a measure of contextuality in terms of the distance of a quantum state $\rho\in\SH$ from the set $\Stab(\oG)$ over \emph{all} realisations.$^{\ref{fn: implicitly complete}}$ In particular, given a realisation $\pi$ of $G$ this includes the unitarily rotated realisations $\pi_U=U\pi U^*$ for every unitary $U\in\UH$. Consequently, Def.~\ref{def: SI-C graph} implicitly assumes that the embedding is with respect to a simple quantum algebra $\LH$. Yet, $G$ may have a (minimal) realisation with respect to a non-simple algebra, that is, with respect to an algebra with superselection rules (for more details, see App.~\ref{app: On realisations}). In this case, allowing for arbitrary, specifically unitarily rotated realisations is indeed not justified and, consequently, the argument in Thm.~1 in Ref.~\cite{RamanathanHorodecki2014}, which reduces the problem to analysing the violation of noncontextuality inequalities by the maximally mixed state, no longer applies.

\subsubsection{A positive resolution of the conjecture in Ref.~\cite{RamanathanHorodecki2014}}\label{sec:RH conjecture}

In this final section, we address a question raised in Ref.~\cite{RamanathanHorodecki2014}: whether $\chi(G)>d$ is a necessary and sufficient condition for state-independent contextuality (in Def.~\ref{def: USI-C graph}) for the completion $\oG$ of the orthogonality graph $G$ with respect to a realisation $\pi$ as in Eq.~(\ref{eq: complement completion}). The failure of the converse of Thm.~\ref{thm: n-colourable implies USI-C} shows that this is not true even in the case of orthogonality graphs with unital realisation, for which the completion in Eq.~(\ref{eq: complement completion}) is trivial, while Thm.~\ref{thm: partial RH conjecture} refutes it for orthogonality graphs with freely completable realisation. However, by comparing with Thm.~\ref{thm: KSNC = NCOG}, we find that for the notion of contextual orthogonality graph in Def.~\ref{def: NCOG} the conjecture holds true.

\begin{theorem}\label{thm: RH conjecture}
    Let $G$ be an orthogonality graph with freely completable realisation $\pi:G\ra\PH$ in $\dim(\cH)=d$ and let $\oG$ be the orthogonality graph of its faithful completion with respect to $\pi$. Then $\oG$ is contextual if and only if $\chi(\CG^*)>d$.

    Moreover, if $\pi:G\ra\PHone$ then $\oG$ is contextual if and only if $\chi(G)>d$.
\end{theorem}

\begin{proof}
    Since a contextual graph is also a SI-C graph, the ``only if''-directions both follow from Thm.~\ref{thm: partial RH conjecture} (equivalently, from Cor.~\ref{cor: NCOG=n-colourability} and Lm.~\ref{lm: restricted chromatic extension}). The first part of the converse assertion holds by Thm.~\ref{thm: KSNC = NCOG} and Thm.~\ref{thm: n-colourability for non-maximal OAs}, and the second additionally by Lm.~\ref{lm: restricted chromatic extension}.
\end{proof}

In particular, Thm.~\ref{thm: RH conjecture} shows that under the restriction to realisations with faithful completion and one-dimensional projections the colouring problem in Thm.~\ref{thm: n-colourability for non-maximal OAs} can be reduced to the chromatic number of subgraphs $G\subset\oG$.

We re-emphasise that the assumptions in Thm.~\ref{thm: RH conjecture} cannot be omitted: (i) the restriction of SI-C to contextual graphs is necessary for the converse of Thm.~\ref{thm: partial RH conjecture}, while (ii) the restriction to (faithful) completions is required already for Thm.~\ref{thm: partial RH conjecture} (see App.~\ref{app: correction}).\\

Thm.~\ref{thm: RH conjecture} completes our comparison between SI-C sets, SI-C graphs and Kochen-Specker contextuality. To summarise, recall that the existence of SI-C sets that are not KS sets such as in Ref.~\cite{YuOh2012} indicated that the traditional notion of KS contextuality had to be superceded. In Ref.~\cite{Frembs2024a} we show that this conclusion is misguided, since the existence of KS sets is only a necessary  but not sufficient criterion to prove KS contextuality (see also Sec.~\ref{sec: KS-colouring vs colouring}). Yet, Ref.~\cite{Frembs2024a} does not fully resolve the relationship between the respective notions. Here, we provided a full comparison: first, we showed that the notions of SI-C set in Def.~\ref{def: SI-C set} and SI-C graph in Def.~\ref{def: SI-C graph} differ from that of Kochen-Specker contextuality for observable algebras in Def.~\ref{def: KSNC} - even under the restriction to unital realisations (by Thm.~\ref{thm: n-colourable implies USI-C} and Cor.~\ref{cor: unital SI-C set}); second, we independently proved that the chromatic number of an orthogonality graph still remains a necessary criterion for state-independent contextuality - provided it has a unital quantum realisation (Thm.~\ref{thm: n-colourable implies USI-C} and Cor.~\ref{cor: unital SI-C set}) or one with a faithful completion (Thm.~\ref{thm: partial RH conjecture} and Cor.~\ref{cor: generalised CKB}, see also App.~\ref{app: correction}); third, we found that for contextual orthogonality graphs (Def.~\ref{def: NCOG}) the chromatic number provides a complete invariant, thus leading to a positive version of the conjecture in Ref.~\cite{RamanathanHorodecki2014}. Together, this fully resolves the subtle differences between the respective notions of contextuality.  Tab.~\ref{tab: KSC vs SI-C} summarises these various relationships.

\begin{table}
    \centering
    \begin{tabular}{ccccc}
        \toprule[0.05cm]
        orthogonality & \multirow{2}*{with} & quantum realisation & & (partial) chromatic \\
        graph $G$ & & in $\dim(\cH)=d$ & & characterisation \\[0.1cm]
        \midrule \\[-.3cm]

        contextual $G$ & & $\pi:G\ra\PH$\ \ unital & \multirow{2}*{$\quad\stackrel{Cor.~\ref{cor: NCOG=n-colourability}}{\Longleftrightarrow}$} & $\chi(G^*)>d$ \\[0.2cm]
        (Def.~\ref{def: NCOG}) & & ( $\pi:G\ra\PHone$\ unital\ \ & & \ $\chi(G)>d$ ) \\[0.3cm]

        $\cap$ & & & & \\[.3cm]
        
        USI-C graph $G$ & & $\pi:G\ra\PH$\ \ unital & \multirow{2}*{$\quad\stackrel{Thm.~\ref{thm: n-colourable implies USI-C}}{\Longrightarrow}$} & $\chi(G^*)>d$ \\[0.2cm]
        (Def.~\ref{def: USI-C graph}) & & ( $\pi:G\ra\PHone$\ unital\ \ & & \ $\chi(G)>d$ ) \\[0.3cm]

        $\cap$ & \multicolumn{4}{l}{\hspace{.05cm} freely completable (f.c.) \rotatebox{90}{\reflectbox{$\hookrightarrow$}} $\uparrow$ unital}\\[.3cm]
        
        SI-C graph $\oG$
        & & $\pi:G\ra\PH$\ \ f. c. & \multirow{2}*{$\quad\stackrel{Thm.~\ref{thm: partial RH conjecture}}{\Longrightarrow}$} & $\chi(\CG^*)>d$\\[.2cm]
        (Def.~\ref{def: SI-C graph}) & & $\pi:G\ra\PHone$\ f. c. & & $\chi(G)>d$\\[.3cm]

        $\cup$ & & & & \\[0.3cm]

        SI-C set $\overline{\pi(V)}$ & & $\pi:G\ra\PH$\ \ f. c.
        & \multirow{2}*{$\quad\stackrel{Cor.~\ref{cor: generalised CKB}}{\Longrightarrow}$} & $\chi(\CG^*)>d$ \\[.2cm]
        (Def.~\ref{def: SI-C set}) & & $\pi:G\ra\PHone$\ f. c.
        & & $\chi(G)>d$ \\[.1cm]
        \bottomrule[0.05cm]
    \end{tabular}
    \caption{The graph-theoretic approach in Ref.~\cite{CabelloSeveriniWinter2014} aims to characterise contextuality purely in terms of a (finite) orthogonality graph $G$. At a minimum, this requires $G$ to admit a realisation with respect to some $d$-dimensional partial algebra, which is equivalent to the existence of a dimension function with $\dim(\One)=d$ (see App.~\ref{app: non-maximal OAs}). Already at this general level, the condition $\chi(G^*)>d$ of orthogonality graphs $G$ with unital realisation is a complete invariant of KS contextuality by Cor.~\ref{cor: NCOG=n-colourability} (top row), and restricts to a necessary criterion for SI-C graphs with unital quantum realisation (second row). More generally, $\chi(\CG^*)>d$ remains a necessary criterion for SI-C graphs $G$ with freely completable (f.c.) realisation (third row), which further carries over to SI-C sets (bottom row). Finally, for realisations restricted to one-dimensional projections, the respective conditions further simplify since $\chi(G)\leq d$ if and only if $\chi(\CG^*)\leq d$ by Lm.~\ref{lm: restricted chromatic extension}.}
    \label{tab: KSC vs SI-C}
\end{table}

\section{Conclusion}\label{sec: conclusion}

\textbf{Summary (see also Tab.~\ref{tab: summary}).} We proved a new and complete characterisation of Kochen-Specker (KS) contextuality \cite{KochenSpecker1967} (Thm.~\ref{thm: CNC for OAs}), based on the concept of context connections (Def.~\ref{def: context connection}) and the constraints encoded by them along context cycles (Def.~\ref{def: context cycle}), which were first introduced in Ref.~\cite{Frembs2024a}. Our result applies at the level of generality (and beyond that) of the original definition, which applied to partial algebras.

Given this generality, the comparison with various other frameworks becomes possible, and we have outlined in detail the subtle differences with various notions of ``contextuality" used in the marginal and (hyper)graph-theoretic approaches \cite{BudroniEtAl2022}. In particular, we highlighted the difference between KS contextuality and the violation of noncontextuality inequalities: while the latter is characterised in terms of acyclic partial orders of contexts by Vorob'ev's theorem \cite{Vorob1962,XuCabello2019} (Thm.~\ref{thm: Vorob'ev theorem}), that is, by the existence of a single non-trivial context cycle \cite{AraujoEtAl2013}, the former generally requires constraints across several contexts cycles (Thm.~\ref{thm: single cycles in d=3 are KS noncontextual}). Nevertheless, KS contextuality can be phrased as a colouring problem (Thm.~\ref{thm: n-colourability for non-maximal OAs}), which explains why traditional proofs of the Kochen-Specker theorem in terms of KS sets and colourings do not always capture KS contextuality: $d$-colourability (Def.~\ref{def: n-colouring}) implies KS colourability (Def.~\ref{def: KS-colouring}), but the converse is generally false \cite{YuOh2012,Frembs2024a}. Moreover, we have fully explored the relation between $d$-colourability and existing partial characterisations of contextuality in the graph-theoretic approach, based on the chromatic number of orthogonality graphs. Specifically, we gave new and refined proofs of the results in Ref.~\cite{Cabello2012b,CabelloKleinmannBudroni2015}, and resolved the conjecture raised in Ref.~\cite{RamanathanHorodecki2014}.\\

\textbf{Outlook.} The unified perspective offered by observable algebras, paired with the effectiveness of the concept of context connections - introduced first in Ref.~\cite{Frembs2024a}, and explored in more detail here - strongly suggests that this novel approach is a fruitful one also when considering applications of KS contextuality. We finish by mentioning a few.

For one thing, analysing constraints on context connections offers a complementary approach to the common study of noncontextuality inequalities, which is mainly concerned with methods from convex geometry. It may thus lead to new insights in characterising noncontextuality polytopes of quantum measurement scenarios, e.g. with respect to optimal noncontextuality inequalities for experimental tests \cite{BartosikEtAl2009,KirchmairEtAl2009,AmselemEtAl2009,GuehneEtAl2010,MoussaEtAl2010,LapkiewiczEtAl2011,Winter2014}.

Yet, our formalism is not tied to quantum theory; and one motivation to study general observable algebras is the question of how to characterise quantum observable algebras among more general ones - explicitly, when expressed in terms of the constraints these pose on context connections. Since our formalism is amenable to various related structures in quantum logic \cite{Putnam1979,Kochen2015,AbramskyBarbosa2020} and general probabilistic theories \cite{Mueller2021,Plavala2023,MuellerGarner2023}, it may thus offer a new perspective towards identifying properties that are essentially quantum, in line with the reconstruction programme of quantum theory \cite{Hardy2001,Hardy2001a,Mueller2021}.

On a more formal level, the algebraic characterisation of contextuality in Thm.~\ref{thm: CNC for OAs} may shed new light on the nature of KS contextuality as a geometric obstruction to a classical state space picture \cite{AshtekarSchilling1997,Kibble1979,Cunha2019}. In particular, it is tempting to interpret context connections as (generalisations of) connections in differentiable geometry.

Of important practical interest is to apply our new tools towards a fine-grained characterisation of contextuality as a resource in quantum computation \cite{AndersBrowne2009,Raussendorf2013,HowardEtAl2014,BravyiGossetKoenig2018,FrembsRobertsBartlett2018,FrembsRobertsCampbellBartlett2023}, quantum cryptography \cite{Ekert1991,BennettBrassard2014,ChaillouxEtAl2016,AmbainisEtAl2016,SchmidSpekkens2018}, quantum machine learning \cite{KleinmannEtAl2011,FagundesKleinmann2017,CabelloEtAl2018,GaoEtAl2022,AnschuetzEtAl2023,Bowles2EtAl023}, quantum metrology \cite{JaeEtAl2023,JaeEtAl2024} etc. In future work, we will compare constraints on context connections with tools from algebraic topology that have been successfully applied to characterise contextuality in measurement-based quantum computation \cite{Raussendorf2016,OkayRoberts2017,OkayTyhurstRaussendorf2018,OkayRaussendorf2020}, as well as sheaf-cohomological methods \cite{AbramskyEtAl2015,Caru2017,BeerOsborne2018,Aasnass2020} and the study of (quantum solutions to) linear constraint systems in \cite{Arkhipov2012,CleveLiuSlofstra2017,OkayRaussendorf2020,FrembsChungOkay2022}. To this end, we wish to promote the present techniques to quantitative tools for the analysis of contextuality necessary, e.g. as need for specific computational tasks. We believe this is a key effort towards provable quantum advantage in quantum computing and beyond.

\section*{Acknowledgments}

I thank Andreas D\"oring for many discussions, and Renato Renner for his generous hospitality during a visit to ETH Z\"urich in 2024, where part of this work was done.

\clearpage

\bibliographystyle{plain}
\bibliography{bibliography}

\appendix

\section{Proof of Thm.~\ref{thm: CNC for OAs}}\label{app: proof of CNC for OAs}

\begin{proof}
    Without loss of generality, we choose non-degenerate observables $O$ for every maximal context $C\in\COm$, that is, $C=C(O)$ with identical spectra $\spec(O)=[d]=\{0,\cdots,d-1\}$ for $d=\dim(\One)$.\footnote{Here, $[d]$ is any $d$-element set. For a $d$-dimensional quantum system, we may take the observables to correspond to (Stern-Gerlach) measurements on a spin-$\frac{d-1}{2}$ particle with $\spec(O)=\{-\frac{d-1}{2},\cdots,\frac{d-1}{2}\}$.} Importantly, note that this still does not fix $O$ uniquely, but only up to a permutation on its spectral decomposition (equivalently, pre-composition of $O:\Sigma_C\ra\R$ by a permutation on $\Sigma_C$).

    $\epsilon$ maps the projections $p\in\PO$ to measurable subsets $\Lambda_p\subset\Lambda$ such that (possibly up to negligible sets) $\Lambda_0=\emptyset$, $\Lambda_\One=\Lambda$ and $\Lambda_{p+p'}=\Lambda_p\dot\cup \Lambda_{p'}$ whenever $pp'=0$. For every observable $O\in\cO$ with spectral decomposition $O=\sum_{\sigma\in[d]}\sigma p_\sigma$, $p_\sigma\in\POone$, $\epsilon$ therefore induces a partition $\Lambda=\dot\cup_{\sigma\in[d]} \Lambda^O_\sigma$ into measurable subsets $\Lambda^O_\sigma$, which yield the outcome $\sigma\in\spec(O)$ under evaluation of $f_O:=\epsilon(O):\Lambda\ra\R$, where $\im(f_O)=\spec(O)=[d]$.\footnote{Formally, $\epsilon|_\PO:\PO\ra\cB(\Lambda)$ defines an orthomorphism into the complete Boolean $\sigma$-algebra $\cB(\Lambda)$.}
    
    Consider two non-degenerate observables $O,O'\in\cO$, and let $\phi_{O'O}:\Lambda\ra\Lambda$ be a measurable function such that $\phi_{O'O}(\Lambda^O_\sigma)\subset\Lambda^{O'}_{\sigma'=\sigma}$ for all $\sigma\in[d]$. Clearly, this definition depends on the choice of generating observables for the contexts $C=C(O),C'=C(O')$. In other words, it depends on a choice of bijection $l_{O'O}:\cP_1(C)\ra\cP_1(C')$. We may thus write $\phi_{O'O}=\phi^{l_{C'C}}_{C'C}$, and $\phi^\fl=(\phi^{l_{C'C}}_{C'C})_{C,C'\in\COm}$ for a collection of maps between maximal contexts, which depends on the context connection $\fl=(l_{C'C})_{C,C'\in\COm}$. We will derive constraints on context connections $\fl$ by concatenating the maps in $\phi^\fl$ along context cycles $(C_0,\cdots,C_{n-1})$, in which case we will further abbreviate our notation to $\phi^{l_{(i+1)i}}_{(i+1)i}:=\phi^{l_{C_{(i+1)}C_i}}_{C_{(i+1)}C_i}=\phi_{O_{(i+1)}O_i}$ for $C_i=C(O_i)$ and $i\in\zz_n$.
    
    To this end, note first that $\phi_{O_{(i+1)}O_i}$ can generally not be chosen one-to-one, since $\Lambda^{O_i}_{\sigma_i=\sigma}$ and $\Lambda^{O_{(i+1)}}_{\sigma_{i+1}=\sigma}$ will generally have different cardinalities. Nevertheless, by restricting to measurable subsets $\tilde{\Lambda}^{O_i}_{\sigma_i}\subset\Lambda^{O_i}_{\sigma_i}$ such that there exist injections $\tilde{\Lambda}^{O_i}_{\sigma_i} \hookrightarrow \Lambda^{O_j}_{\sigma_j=\sigma_i}$ for all $\sigma_i\in[d]$ and $i,j\in\zz_n$, we can always find invertible (measurable) maps $\phi_{O_{(i+1)}O_i}:\Lambda\ra\Lambda$ such that $\phi_{O_{(i+1)}O_i}(\tilde{\Lambda}^{O_i}_{\sigma_i}) = \tilde{\Lambda}^{O_{(i+1)}}_{\sigma_{(i+1)}=\sigma_i}$ for all $\sigma_i\in[d]$ and such that $\phi_{O_{(i+1)}O_i}|_{\Lambda/\tilde{\Lambda}^{O_i}}=\mathrm{id}$ for all $i\in\zz_n$, where $\tilde{\Lambda}^{O_i}=\dot\cup_{\sigma_i\in[d]} \tilde{\Lambda}^{O_i}_{\sigma_i}$.\footnote{Indeed, we may choose $\tilde{\Lambda}^O_\sigma\subset\Lambda^O_\sigma$ to be of cardinality $\min_{O'\in\cO} |\Lambda^{O'}_{\sigma'=\sigma}|$ for all $\sigma\in[d]$ and $O\in\cO$.}
    
    Denote by $f_i=\epsilon(O_i)$ the measurable function on $\Lambda$ representing the observable $O_i$ under the classical embedding $\epsilon$. Then $f_{(i+1)}|_{\tilde{\Lambda}^{O_{(i+1)}}} = f_i|_{\tilde{\Lambda}^{O_i}}\circ(\phi^{l_{(i+1)i}}_{(i+1)i})^{-1}$, hence,\footnote{What is more, since $\phi^{l_{(i+1)i}}_{(i+1)i}|_{\Lambda/\tilde{\Lambda}^{O_i}}=\mathrm{id}$ we also have $f_0\circ\left(\circ_{i=0}^{n-1} \phi^{l_{(i+1)i}}_{(i+1)i}\right)^{-1} = f_0$.}
    \begin{equation*}
        f_0|_{\tilde{\Lambda}^{O_0}} \circ \left(\circ_{i=0}^{n-1} \phi^{l_{(i+1)i}}_{(i+1)i}\right)^{-1} = f_0|_{\tilde{\Lambda}^{O_0}}
        \ \ \Llra \ \ \circ_{i=0}^{n-1} \phi^{l_{(i+1)i}}_{(i+1)i} = \id
        \ \ \Llra \ \ \circ_{i=0}^{n-1} l_{(i+1)i} = \id\; ,
    \end{equation*}
    where we used that $f_0=\epsilon(O_0)$ is non-degenerate with respect to the $\{\Lambda^{O_0}_\sigma\}_{\sigma\in[d]}$-partition ($O_0$ is non-degenerate and $\im(f_0)=\spec(O_0)=[d])$. Since we constructed $\phi^\fl$ simply by assuming the existence of a classical embedding $\epsilon$, the result follows.

    Conversely, let $\fl$ be a flat context connection on $\CO$, that is, $\fl$ satisfies the triviality constraints in Eq.~(\ref{eq: CNC}) for every context cycle in $\COm$. Then we construct a state space $\Lambda$ as follows. Let $\Sigma_{C_0} := \{\hat{\sigma}_p: \cP_1(C_0)\ra\{0,1\}\mid\forall p,q\in\cP_1(C_0):\ \hat{\sigma}_p(q) = \delta_{pq}\}$ be the ``state space" of $C_0\in\COm$.\footnote{$\hat{\sigma}_p\in\Sigma_{C_0}$ defines a state on $C_0$ by linear extension, $\hat{\sigma}_p(O)=\hat{\sigma}_p(\sum_q\sigma_q q)=\sigma_q$ for all $O=\sum_{q}\sigma_q q \in C_0$. Equivalently, we may view any observable $O\in C_0$ as a (discrete) random variable on $\Sigma_0$ (see Fig.~\ref{fig: state space decomposition}).} For any $\lambda_{C_0}\in\Sigma_{C_0}$, we define $\lambda^{\fl,\lambda_{C_0}} = (\lambda_C)_{C\in\CO}$, first, for all maximal contexts $C\in\COm$ by setting $\lambda_C=\fl_{CC_0}.\lambda_{C_0}$, (where the action of $\fl_{CC_0}$ on $\Sigma_{C_0}$ is defined by $(\fl_{CC_0}.\lambda_{C_0})(p)=\lambda_{C_0}(\fl^{-1}_{CC_0}(p))$ for all $p \in \mc{P}_1(C)$) and, second, by $\lambda_{\tC|C} := \lambda_C|_{\tC}$ for all $\tC\subset C\in\COm$. Since $\fl$ preserves the order relations in $\CO$ by definition (Eq.~(\ref{eq: noncontextual context connection})), it follows that $\lambda_{\tC|C} = \lambda_{\tC|C'} =: \lambda_{\tC}$ whenever $\tC\subset C,C'\in\COm$.\footnote{In other words, $\lambda$ defines a global section of the spectral presheaf of $\CO$ \cite{IshamButterfieldI,DoeringIsham2011,DoeringFrembs2019a}.} $\lambda^{\fl,\lambda_{C_0}}$ thus assigns a value to every observable independent of its context. Collecting all such states, and writing $\Xi_\mathrm{cl}(\CO)$ for the space of flat context connections on $\CO$, that is, those that satisfy Eq.~(\ref{eq: CNC}), we define
    \begin{equation}\label{eq: state space from context connecion}
        \Lambda:=\{\lambda^{\fl,\lambda_{C_0}} \mid \lambda_{C_0}\in\Sigma_{C_0}, \fl\in\Xi_\mathrm{cl}(\CO)\}\; .
    \end{equation}
    We equip $\Lambda$ with the product topology of the (discrete) topology on $\Sigma_{C_0}$ together with any topology on $\Xi_\mathrm{cl}(\CO)$ 
    (and its associated Borel $\sigma$-algebra). Finally, we define an embedding $\epsilon:\cO\ra L_\infty(\Lambda)$ by setting $\epsilon(O)(\lambda)=\lambda(O)$. Clearly, this defines a measurable function: since the spectrum of $O$ is discrete, any subset $S\subset\spec(O)$ is measurable, as its pre-image $\epsilon^{-1}(O)(S)=\{\lambda\in\Lambda\mid\lambda(O)\in S\}=S\times\Xi_\mathrm{cl}(\CO)$ is of product form, with both $S\subset\Sigma_{C_0}$ and $\Xi_\mathrm{cl}(\CO)$ measurable, hence, $\epsilon^{-1}(O)(S)$ is measurable in $\Lambda$.
\end{proof}

\begin{figure}
    \centering
    \scalebox{1.3}{\begin{tikzpicture}[every node/.style={scale=0.75},scale=0.55]
    \node (0) at (-0.25, 7.75) {};
    \node (1) at (2.75, 9.75) {};
    \node (2) at (8.75, 6.75) {};
    \node (3) at (11.75, 8.75) {};
    \node (4) at (-5, 4.5) {};
    \node (5) at (6.5, 4.5) {};
    \node (6) at (7, 4) {$\mathbb{R}$};
    \node (7) at (5.25, 8.35) {$\small{\bullet}$};
    \node (9) at (0, 4.5) {};
    \node (11) at (5.775, 8.425) {$\lambda$};
    \node (13) at (1.5, 5.25) {$f_O = \epsilon(O)$};
    \node (18) at (5, 7.375) {};
    \node (21) at (0.775, 8.75) {};
    \node (22) at (1.5, 9.225) {};
    \node (23) at (3.5, 7.775) {};
    \node (24) at (8.5, 8.625) {};
    \node (44) at (-5.95, 8.425) {};
    \node (45) at (-3.7, 8.575) {};
    \node (46) at (-5, 9) {};
    \node (47) at (-7.2, 8.75) {};
    \node (48) at (-3.2, 8.75) {};
    \node (51) at (1.175, 8.375) {};
    \node (52) at (7.3, 7.425) {};
    \node (53) at (3.25, 9) {};
    \node (63) at (-5.25, 8.15) {};
    \node (64) at (-1.25, 5.325) {$O$};
    \node (67) at (10.5, 7.5) {$\Lambda$};
    \node (68) at (-6.2, 8.675) {$\hat{\sigma}_1$};
    \node (69) at (-5.225, 8.675) {$\hat{\sigma}_2$};
    \node (70) at (-4.25, 8.75) {};
    \node (71) at (-4.2, 8.675) {$\hat{\sigma}_3$};
    \node (78) at (-5.25, 9.75) {};
    \node (79) at (-7.75, 8.8) {$\Sigma_O$};
    \node (80) at (0, 4) {$\epsilon(O)(\lambda) = O(\varepsilon_O(\lambda))$};
    \node (81) at (1.925, 8.5) {$\Lambda^O_{\sigma_3}$};
    \node (82) at (3.85, 8.75) {$\Lambda^O_{\sigma_2}$};
    \node (83) at (8, 7.5) {$\Lambda^O_{\sigma_1}$};
    \node (84) at (-1.25, 10.375) {$\varepsilon_O$};
    \node (85) at (-1.25, 8.675) {$\varepsilon_O$};
    \node (86) at (-1.25, 7.125) {$\varepsilon_O$};
    \draw [in=-165, out=0] (1.center) to (3.center);
    \draw [bend right=15, looseness=0.75] (3.center) to (2.center);
    \draw [in=180, out=15] (0.center) to (2.center);
    \draw [bend right=15] (1.center) to (0.center);
    \draw [->] (4.center) to (5.center);
    \draw [dashed, in=0, out=-165, looseness=1.50] (24.center) to (23.center);
    \draw [dashed, in=150, out=0, looseness=1.25] (22.center) to (23.center);
    \draw [bend left=90, looseness=0.50] (47.center) to (48.center);
    \draw [bend right=90, looseness=0.50] (47.center) to (48.center);
    \draw [->, in=90, out=-90] (63.center) to (9.center);
    \draw [->, in=90, out=-90] (7.center) to (9.center);
    \draw [->,dashed, bend left=315, looseness=0.50] (53.center) to (46.center);
    \draw [->,dashed, bend left=15, looseness=0.75] (51.center) to (45.center);
    \draw [->,dashed, bend right=330, looseness=0.75] (52.center) to (44.center);
\end{tikzpicture}}
    \caption{Under the classical embedding $\epsilon:\cO\ra L_\infty(\Lambda)$, every observable $O\in\cO$ induces a decomposition $\Lambda=\dot\cup_{\sigma\in\spec(O)} \Lambda^O_\sigma$. Viewing $O:\Sigma_C\ra\R$ with $C=C(O)$ as a random variable, we can thus define surjective maps $\varepsilon_O:\Lambda\ra\Sigma_O$ by $\Lambda^O_\sigma\ni\lambda\mapsto\hat{\sigma}\in\Sigma_O$ such that $\epsilon$ factorises as $\epsilon(O)=O\circ\varepsilon_O$ for all $O\in\cO$.}
    \label{fig: state space decomposition}
\end{figure}
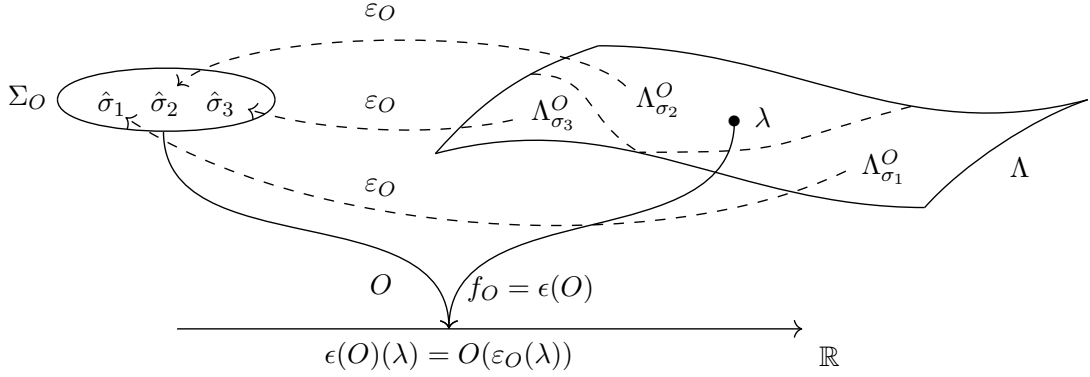
\section{Finite-dimensional observable algebras}\label{app: non-maximal OAs}

In Ref.~\cite{Frembs2024a}, context connections on the partial order of contexts of a (spin-$1$) quantum system were defined as bijective maps between rank-$1$ projections. App.~\ref{app: proof of CNC for OAs} generalises Thm.~2 in Ref.~\cite{Frembs2024a} to maximal observable algebras (see Def.~\ref{def: max OA}), where - similar to the quantum case - minimal projections all have the ``same size". However, not every observable algebra is maximal, in particular, minimal events may be of ``different size''. In this section, we will extend the result (and thus the proof of Thm.~\ref{thm: CNC for OAs}) to finite-dimensional observable algebras, for which a notion of size (see Def.~\ref{def: dim}) still exists.

We first characterise maximal observable algebras in terms of maximally mixed states. Let $\dim$ be a dimension function on $\cO$ and define the map $\gamma_{\dim}:\PO\ra\mathbb{Q}$ by $\gamma_{\dim}(p)=\frac{\dim(p)}{\dim(\One)}$ for all $p\in\PO$. Then $\gamma_{\dim}\in\Gamma[\CO]$ since the marginalisation constraints in Eq.~(\ref{eq: no-disturbance}) hold by additivity of $\dim$, and $\gamma_{\dim}(\One)=\frac{\dim(\One)}{\dim(\One)}=1$ is normalised.\footnote{For finitely generated observable algebras the converse also holds: if $\gamma\in\Gamma[\CO]$ with $\gamma(p)\in\mathbb{Q}$ for all $p\in\PO$, then we obtain a dimension function by setting $\dim(p)=L\gamma(p)$, where $L$ is the least common multiple of all denominators in $\gamma(p)$.\label{fn: fin gen and classical implies fin dim}}

\begin{lemma}\label{lm: max OA vs max-mixed states}
     An observable algebra $\cO$ is maximal if and only if it contains the maximally mixed state $\gamma_{\One/d}\in\Gamma[\CO]$ with $\gamma_{\One/d}(p)=\frac{1}{d}$ for all $p\in\POm$ and $d=\dim(\One)$.
\end{lemma}

\begin{proof}    
    If $\cO$ is maximal and of dimension $\dim(\One)=d$, then the state $\gamma_{\One/d}(p)=\frac{\dim(p)}{d}$ is maximally mixed since $\dim(p)=1$ for all $p\in\POone$. Conversely, if $\cO$ contains the maximally mixed state $\gamma_{\One/d}\in\Gamma[\CO]$, defined by $\gamma_{\One/d}(p)=\frac{1}{d}$ for all $p\in\cP_\mathrm{min}(\cO)$, then the canonical dimension function on $\cO$ is given by $\dim(p)=d\gamma_{\One/d}(p)$.
\end{proof}

While the existence of a maximally mixed state may seem a trivial assumption, there exist finitely generated observable algebras that do not admit such a state, and indeed ones that admit no states at all \cite{Greechie1971}. By Lm.~\ref{lm: max OA vs max-mixed states}, such observable algebras do not have a dimension function.\footnote{For finitely generated observable algebras, the additive constraints of a dimension function define a linear set of Diophantine equations (see Eq.~(\ref{eq: dim function constraints}) in App.~\ref{app: On realisations}). The existence of a dimension function thus coincides with the existence of a solution to this set of equations which can be checked efficiently.}\\

\textbf{Maximal extensions of observable algebras.} Next, we show that observable algebra $\cO$ may have a dimension function yet not be maximal, and may even support different dimension functions. For instance, given a set of minimal events $p_1,p_2,p_3\in\PC$ with $\dim(p_i)=1$ for $i=\{1,2,3\}$ for some context $C\in\COm$ of some maximal observable algebra $\cO$ such that $p_ip=0$ only if $p\in\PC$. Then the ``coarse-grained'' observable algebra $\tcO$, which instead of $p_1,p_2,p_3$ only contains the events $\tp_1,\tp_2$ and is otherwise identical, admits at least two dimension functions: (i) $\widetilde{\dim}(\tp_1)=1$, $\widetilde{\dim}(\tp_2)=2$ and $\widetilde{\dim}(p)=\dim(p)$ for all $p\in\PO/\{p_1,p_2,p_3\}$ and (ii) $\widetilde{\dim}'(\tp_1)=2$, $\widetilde{\dim}'(\tp_2)=1$ and $\widetilde{\dim}'(p)=\dim(p)$ for all $p\in\PO/\{p_1,p_2,p_3\}$.

Observable algebras of this type arise, for instance, from measurement scenarios (see Def.~\ref{def: measurement scenario} in App.~\ref{app: measurement scenarios vs OAs}), where the (discrete) outcomes sets of measurement labels generally do not have the same cardinality. In order to deal with such cases, we need to generalise Thm.~\ref{thm: CNC for OAs} from maximal to general (finite-dimensional) observable algebras. The basic idea is to embed finite-dimensional observable algebras within maximal ones.

\begin{definition}\label{def: max extension}
    Let $\cO$ be an observable algebra with dimension function $\dim$. Then the \emph{maximal extension $\cO^*$ of $\cO$ with respect to $\dim$} is a maximal observable algebra $\cO^*$ of dimension $\dim(\One^*)=\dim(\One)$, together with an embedding $\iota:\cO\ra\cO^*$ such that:
    \begin{itemize}
        \item[(i)] for all $C\in\COm$ there exists a unique maximal context $C^*\in\sCOm$ such that $C\subset C^*$ (and $C=C^*$ if $C$ is minimally generated by $\dim(\One)$ elements),
        \item[(ii)] for all $C'^*\in\sCO$ either $C'^*\in\CO$ or $C'^*\subset C^*\in\sCOm$, where $C^*$ is the unique maximal context of some maximal context $C\in\COm$ by (i),
        \item[(iii)] for all $C,C'\in\COm$ with $C\neq C'$ it holds $C \cap C' = C^* \cap C'^*$.
    \end{itemize}
\end{definition}

Conditions (i) and (ii) in Def.~\ref{def: max extension} define maximal extensions to be minimal under all embeddings of $\cO$ into maximal observable algebras. Furthermore, if (iii) was not true, then there would exist maximal contexts $C,C'\in\COm$ with $C\cap C' \subsetneq C^*\cap C'^*$, that is, there would exist $\cP(C^*),\cP(C'^*)\ni p^*\notin\PO$ such that $p^*p=p^*p'=0$ for all $p\in\PC$ and $p'\in\cP(C')$. Consequently, $\cO^*$ would therefore pose additional constraints (not present in $\cO$) between elements in \emph{different} maximal contexts in $\CO$ since $p^*\leq(\One-p),(\One-p')$ for all $p\in\PC$ and $p'\in\cP(C')$. (iii) thus ensures that $\CO$ already contains all relevant (noncontextuality) constraints between elements in $\cO$.

The following lemma shows that a maximal extension always exists and is unique.

\begin{lemma}\label{lm: min extension to max observable algebra}
    Let $\cO$ be an observable algebra with partial order of contexts $\CO$ and dimension function $\mathrm{dim}: \PO \ra \mathbb{N}$. Then $\cO$ has a unique maximal extension $\iota:\cO\hra\cO^*$ with maximal observable algebra $\cO^*$, corresponding partial order of contexts $\sCO$ and dimension function $\dim^*:\mc{P}(\cO^*)\ra\mathbb{N}$ such that $\dim^*|_\PO=\dim$.
\end{lemma}

\begin{proof}
    Write $|C|:=\cP_\mathrm{min}(C)$ for the cardinality of a minimal generating set of $C$. For every maximal context $C\in\COm$, define $\iota_{C^*C}:C\hra C^*$ to be an embedding of $C$ into a commutative algebra $C^*$ with $|C^*|=\dim(\One)$ and such that $\dim^*(\iota_{C^*C}(p))=\dim(p)$ for all $p\in\PC$. Then $\iota_{C^*C}$ is the unique (up to isomorphism) extension of $C$ such that for every $p\in\PC$ with $\dim(p)=k$ there exist projections $p^*_1,\cdots,p^*_k\in\cP(C^*)$ with $\dim^*(p^*_i)=1$, $\sum_{i=1}^k p^*_i=\iota_{C^*C}(p)$. Clearly, (i) and (ii) in Def.~\ref{def: max extension} are satisfied under this definition. Moreover, $p^*_ip=0$ for $p\in\PO$ if and only if $p\in\PC$ by construction such that $C\cap C'=C^*\cap C'^*$ for any two distinct maximal contexts $C,C'\in\COm$. Hence, (iii) in Def.~\ref{def: max extension} also holds and $\cO^*$ is a maximal observable, is clearly unique and $\dim^*:\cP(\cO^*)\ra\mathbb{N}$ defines a dimension function on $\cO^*$ such that $\dim^*|_\PO=\dim$.
\end{proof}

Using Lm.~\ref{lm: min extension to max observable algebra}, we can reduce the study of KS contextuality of (non-maximal) finite-dimensional observable algebras to that of their maximal extensions.

\begin{lemma}\label{lm: KSNC inherits from maximality}
    Let $\cO$ be a finite-dimensional observable algebra. Then the following are equivalent:
    \begin{itemize}
        \item[(i)] $\cO$ is Kochen-Specker noncontextual,
        \item[(ii)] there exists a maximal extension $\cO^*$ as by Lm.~\ref{lm: min extension to max observable algebra} with respect to some choice of dimension function $\dim$ on $\cO$ that is Kochen-Specker noncontextual,
        \item[(iii)] the maximal extension $\cO^*$ as by Lm.~\ref{lm: min extension to max observable algebra} with respect to any choice of dimension function $\dim$ on $\cO$ is Kochen-Specker noncontextual.
    \end{itemize}
\end{lemma}

\begin{proof}
    The implications (iii)$\Rightarrow$(ii)$\Rightarrow$(i) are obvious. Indeed, let $\cO^*$ be the unique maximal observable algebra of $\cO$ in Lm.~\ref{lm: min extension to max observable algebra} with respect to any dimension function. If $\cO^*$ is KS noncontextual then there exists a classical embedding $\epsilon:\cO^*\ra L_\infty(\Lambda)$ for some measurable space $\Lambda$ which induces an embedding $\epsilon|_\cO=\epsilon\circ\iota:\cO\ra L_\infty(\Lambda)$ for $\cO$ by restriction. Hence, $\cO$ is KS noncontextual.

    To prove (i)$\Rightarrow$(ii), assume that $\cO$ is KS noncontextual, that is, there exists an embedding $\epsilon:\cO\ra L_\infty(\Lambda)$ for some measurable space $\Lambda$. Now, let $\cO^*$ be the unique maximal observable algebra in Lm.~\ref{lm: min extension to max observable algebra} with respect to some dimension function $\dim$ on $\cO$. We need to show that every observable $\cO^*\ni O^*\notin\cO$ can also be represented as a measurable function on $\Lambda$. To do so, it is sufficient to find a decomposition of the measurable subset $\Lambda_p = \supp(\chi(\lambda\mid p))\subset\cB(\Lambda)$, corresponding to the indicator function $\chi(\lambda\mid p)=\epsilon(p)(\lambda)$ for every projection $p\in\PO$ with $\mathrm{dim}(p)=k>1$, into measurable subsets $\Lambda_{p^*_i} =: \supp(\chi(\lambda\mid p^*_i))$, $\Lambda_p=\dot\cup_{i=1}^k \Lambda_{p^*_i}$ for $p^*_1,\cdots,p^*_k$ the unique one-dimensional projections in $\cP(\cO^*)$ with $\sum_{i=1}^k p^*_i=\iota_{C^*C}(p)$ (where $\iota_{C^*C}:C\hra C^*$ is the embedding defined by the maximal extension $\iota$ in Lm.~\ref{lm: min extension to max observable algebra}). The only constraints on the subsets $\Lambda_{p^*_i}$ are that they are measurable and that $\Lambda_{p^*_i}\subset\Lambda_{p'}$ only if $p^*_i\leq p'$ for $p'\in\cP(\cO^*)$.
    
    Let $\sigma|_{\Lambda_p}=\{\widetilde{\Lambda}\cap\Lambda_p\mid\widetilde{\Lambda}\in\sigma_\Lambda\}$ be the induced $\sigma$-algebra under the restriction to $\Lambda_p$. Moreover, define $\Lambda_{\cP}=\dot\cap_{p'\in\cP} \Lambda_{p'}$ as the intersection with respect to any subset of events $\cP\subset\PO$. Note that $\Lambda_\cP\in\sigma_\Lambda$, in particular, $\Lambda_\cP\in\sigma|_{\Lambda_p}$ whenever $\cP$ is a countable set. If $\cO$ is finitely generated, $\Lambda_\cP$ is a finite intersection, hence, $\Lambda_\cP$ is itself a measurable subset for all $\cP\subset\PO$. We can therefore find a (finite) minimal generating set $\cG_{\Lambda_p}$ for $\sigma|_{\Lambda_p}$ such that $\Lambda_p=\dot\cup_{\Lambda_\cP\in\cG_{\Lambda_p}}\Lambda_\cP$ and $\Lambda_\cP\cap\Lambda_{\cP'}=\emptyset$ for all $\Lambda_\cP,\Lambda_{\cP'}\in\cG_{\Lambda_p}$.\footnote{Indeed, for finitely generated observable algebras we may assume that $\Lambda$ is a finite set (see Lm.~\ref{lm: optimal embedding}).} After possibly enlarging each generating set $\Lambda_{\cP}\in\cG_{\Lambda_p}$ such that it contains at least $k$ elements, let $\Lambda_\cP=\dot\cup_{i=1}^k \Lambda_{\cP,i}$ be a disjoint decomposition into (measurable) subsets for every $\Lambda_\cP\in\cG_{\Lambda_p}$, and define $\Lambda_{p^*_i}:=\dot\cup_{\Lambda_\cP\in\cG_{\Lambda_p}}\Lambda_{\cP,i}$. It follows that $\Lambda_p=\dot\cup_{i=1}^k\Lambda_{p^*_i}$ and $\Lambda_{p^*_i}\subset\Lambda_{p'}$ only if $p^*_i\leq p'$ for $p'\in\cP(\cO^*)$ by construction. Finally, setting $\epsilon^*(p^*_i)=\Lambda_{p^*_i}$ and $\epsilon^*|_\PO=\epsilon$ we thus obtain a classical embedding for the maximal extension $\cO^*$.

    If $\cO$ is not finitely generated, the $\sigma$-algebra $\sigma_\Lambda$ on $\Lambda$ is not finitely generated either and the above construction does not work. In this case, we construct a new space $\Lambda'$ by replacing $\Lambda_p$ in $\Lambda$ by the product space $\Lambda_p\times[k]$ for $\dim(p)=k$ where $[k]$ denotes a set of $k$ elements with $\sigma$-algebra $\sigma_{p}(k)=2^{[k]}$. We then recover $\sigma|_{\Lambda_p}$ as the projection onto the first factor in the product $\sigma$-algebra $\sigma|_{\Lambda_p}\times\sigma_p(k)$ on $\Lambda_p\times[k]$, and since every measurable set $\Lambda_\cP\in\sigma$ decomposes as $\Lambda_\cP=(\Lambda_\cP\cap\Lambda_p)\cup(\Lambda_\cP\cap\overline{\Lambda_p})$, we also recover $\sigma_\Lambda$ by defining the $\sigma$-algebra $\sigma_{\Lambda'}$ on $\Lambda'$ as the smallest $\sigma$-algebra generated by the sets $\Lambda_\cP=((\Lambda_\cP\cap\Lambda_p)\times 1_k)\cup(\Lambda_\cP\cap\overline{\Lambda_p})$ for all $\Lambda_\cP\in\sigma$. We thus obtain a new classical embedding $\widetilde{\epsilon}':\cO\ra\Lambda'$ of $\cO$. Moreover, $\widetilde{\epsilon}'$ is readily extended to the elements $p^*_i<p$ in $\cP(\cO^*)$ by setting $\epsilon'(p^*_i)=\Lambda_p\times\{i\}=:\Lambda_{p^*_i}$ for all $1\leq i\leq k$. Clearly, we have $\Lambda_{p^*_i}\neq\Lambda_{q}$ for any $q\in\cP(\cO^*)$ by construction.

    To prove (i)$\Rightarrow$(iii) it is enough to note that the construction in the previous paragraph is independent of, hence, applies to any choice of dimension function $\dim$ on $\cO$.
\end{proof}
\section{Observable algebras with separating sets of states}\label{app: OAs with separating states}

Prop.~\ref{prop: nonclassical correlations in KS noncontextual OAs} says that a KS noncontextual observable algebra $\cO$ is not necessarily fully classically correlated: the existence of a classical embedding implies $\Gamma_\mathrm{cl}[\CO]\neq\emptyset$, yet generally not $\Gamma_\mathrm{cl}[\CO]=\Gamma[\CO]$. On the other hand, this suggests that KS noncontextuality is a weaker notion than being fully classically correlated. In this section, we prove under which condition this is the case. We first show that it is not true in general.

\begin{lemma}\label{lm: single state}
    There exist observable algebras $\cO$ which are Kochen-Specker contextual, yet which are fully classically correlated, that is, for which $\emptyset\neq\Gamma_\mathrm{cl}[\CO]=\Gamma[\CO]$.
\end{lemma}

\begin{proof}
    We will use the construction of state spaces for orthomodular lattices in Ref.~\cite{Navara1994,Navara2008} (and references therein). To this end, recall first that every orthomodular lattice defines a transitive partial Boolean algebra \cite{Gudder1972}, which in turn generates a unique partial algebra (see Lm.~\ref{lm: embeddings vs orthomorphisms vs order-preserving maps}). It follows that there exist orthomodular lattices $L$ with exactly one state, which assigns every atom (and thus every minimal projection in $\PO$ in the finitely generated, maximal observable algebra $\cO(L)$ associated with $L$) the same value, hence, which is a maximally mixed state.
    
    Let $L_1$ be such an orthomodular lattice and denote by $\gamma_\One\in\Gamma[L_1]$ the unique (maximally mixed) state on $L_1$. For instance, we may take $L_1$ to be the  orthomodular lattice in Ref.~\cite{Navara1994}. Next, consider its product $L=L_1\times L_2$ with the Boolean algebra $L_2\cong\cB(2^{\Lambda})$ of measurable sets on some finite set $\Lambda$. In this case we have $\Gamma[L]\cong\Gamma[L_2]$ \cite{Navara2009}. $L$ defines a finitely generated, maximal observable algebra $\cO(L)$, in particular, note that $\gamma_\One$ gives rise to a dimension function on $\PO$ by Lm.~\ref{lm: max OA vs max-mixed states}. Hence, we can meaningfully talk about Kochen-Specker contextuality,\footnote{In turn, for this reason we cannot use orthomodular lattices with no states \cite{Greechie1971,Kalmbach1983,Navara1994}.} and clearly, $\cO(L)$ is Kochen-Specker contextual, since a classical embedding would define more than a single state on $L_1$, and thus $\Gamma[L_2]\subsetneq\Gamma[L_1\times L_2]$ contradicting $\Gamma[L_2]\cong\Gamma[L_1\times L_2]$.

    Still, every state $\gamma$ on $L$ has a solution to the marginal problem. Indeed, the respective (non-faithful) orthomorphisms in Def.~\ref{def: marginal problem} are simply given by the projection onto the second factor, that is, $\phi_\gamma:L_1\times L_2\ra \mathbf{1}\times L_2\cong L_2$, where $\mathbf{1}$ denotes the unique single-element Boolean algebra. Consequently, $\emptyset\neq\Gamma_\mathrm{cl}[\cO(L)]=\Gamma[\cO(L)]$, hence, $\cO(L)$ is fully classically correlated.
\end{proof}

Admittedly, the example in the proof of Lm.~\ref{lm: single state} is rather artificial, and relies on the vast generality of observable algebras (specifically, orthomodular lattices), as compared to observable subalgebras $\cO\subset\LHsa$ 
in quantum theory. Physically, observable algebras as in Lm.~\ref{lm: single state} violate the idea that contextuality, that is, the compatibility relations between observables in $\cO$ are \emph{learnt}: if states cannot distinguish between projections in different contexts, then the structure of $\CO$ cannot be inferred operationally alone.

In the same spirit, one may argue that the assumption of an embedding in Def.~\ref{def: KSNC} should be restricted to observable algebras admitting sufficiently many states in order to distinguish the events in $\PO$ in $\cO$. This motivates the following definition.

\begin{definition}\label{def: separating set}
    Let $\cO$ be an observable algebra with event algebra $\PO$. A set of states $\Gamma\subset\Gamma[\CO]$ is called \emph{separating} if for any $p,p'\in\PO$ there exists $\gamma\in\Gamma$ with $\gamma(p)\neq\gamma(p')$. $\cO$ is called \emph{separating} if it admits a separating set of states.
\end{definition}

We observe that if $\cO$ is finitely generated and separating, then it is also finite-dimensional. Note also that observable subalgebras $\cO\subset\LHsa$ in quantum theory are separating. Moreover, if $\cO$ admits a classical embedding $\epsilon:\cO\ra L_\infty(\Lambda)$ every measure on $\Lambda$ induces a state on $\cO$. It follows that the existence of a separating set of \emph{classical} states is a necessary condition for $\cO$ to be KS noncontextual.

In turn, under the assumptions that $\cO$ is finitely generated and separating, the implication ``fully classically correlated $\Rightarrow$ Kochen-Specker noncontextual" does indeed hold.

\begin{proposition}\label{prop: non-empty fully classical implies KS noncontextual}
    Let $\cO$ be a finitely generated, separating and fully classically correlated observable algebra. Then $\cO$ is Kochen-Specker noncontextual.
\end{proposition}

\begin{proof}
    By assumption, $\cO$ contains a separating set of states, which is classical since $\cO$ is also fully classically correlated. Hence, $\cO$ is marginally KS noncontextual by Def.~\ref{def: marginal noncontextuality}, which further implies KS noncontextuality by Thm.~\ref{thm: KSNC = marginal KSNC}.
\end{proof}

Together with Prop.~\ref{prop: nonclassical correlations in KS noncontextual OAs}, we thus find that for finitely generated, separating observable algebras, KS noncontextuality is a weaker notion of classicality than Def.~\ref{def: fully classically correlated} (see Tab.~\ref{tab: comparison}).
\section{Comparison of approaches and unified perspective}\label{app: unified perspective}

Unfortunately, the various approaches to (the marginal approach to) contextuality use slightly different mathematical structures, notations and conventions. The purpose of this section is threefold: first, to give a minimal account of the necessary terminology in order to collect all these approaches under one umbrella,\footnote{A similar effort was made in Ref.~\cite{BudroniEtAl2022}, where the relevant notion was identified as a ``marginal scenario" \cite{ChavesFritz2012,FritzChaves2013}. Here, we substantially expand this view: we add the complementary algebraic perspective on KS contextuality, provide a detailed comparison between the respective approaches and prove various new results including the minimal mathematical structure necessary to study KS contextuality.} second, to expose unnecessary redundancies inherent in various descriptions, and third, to prove that the formalism of observable algebras and their underlying partial orders of contexts subsumes the minimal account in the most economical fashion. The latter may thus be seen as a unified framework for Kochen-Specker contextuality, and we endorse it as such.

\subsection{Measurement scenarios vs observable algebras}\label{app: measurement scenarios vs OAs}

We begin by reviewing the notion of a measurement scenario in the sheaf-theoretic approach \cite{AbramskyBrandenburger2011}. For a brief review of the sheaf approach, see Ref.~\cite{AbramskyBrandenburger2011,AbramskyEtAl2015,AbramskyEtAl2017}.

\begin{definition}\label{def: measurement scenario}
    Let $X$ be a finite set of \emph{measurement labels}. Each label $x\in X$ is assigned a discrete \emph{outcome space $o_x\subset\R$}.\footnote{For a generalisation to random variables with continuous spectra, see Ref.~\cite{BarbosaEtAl2022}.} A \emph{measurement cover $\cM \subset 2^X$ of $X$} is a set of subsets of $X$ such that (i) $\cM$ covers $X$, that is, $\bigcup \cM = X$, and (ii) $\cM$ is an anti-chain, that is, whenever $\mC'\subset\mC$ for $\mC,\mC'\in\cM$ then $\mC'=\mC$. A pair $M = (X,\cM)$ is called a \emph{measurement scenario}.

    Any subset $\tilde{\mC}\subset X$ that can be written as $\tilde{\mC}=\bigcap_{i=1}^n \mC_i$ for $\mC_i\in\cM$ is called a \emph{measurement context}. In particular, the elements $\mC\in\cM$ are called \emph{maximal measurement contexts}. The set of all measurement contexts is denoted $\cC(\cM)$ with $\cC_\mathrm{max}(\cM)=\cM$.
\end{definition}

We will distinguish measurement contexts in the sheaf-theoretic approach, that is, subsets of $X$ from contexts (commutative algebras in observable algebras) by using the non-italic font for the former.

It is usually not explicitly assumed that $X$ contains a measurement label $I\in X$ corresponding to the trivial observable which has outcome set $o_I=\{1\}$\footnote{Measuring $I$ always results in the outcome $1$ (``true"). Its interpretation lies in defining the physical system with respect to which all other measurements $x\in X$ are performed.} and is contained in every context, hence, defines a least element $\mC_I=\{I\}$ of $\cC(\cM)$. Nevertheless, this assumption is made implicitly when considering states (``empirical models") (see Def.~\ref{def: state}), where it appears in the form of the normalisation condition on probability distributions across contexts. Indeed, normalisation follows from $\mu_\mC(1) = \mu_\mC|_{\mC_I}(1)=1$ if $\mC_I\subset\mC$ for all $\mC\in\cC(\cM)$. For this reason, we will henceforth explicitly include the measurement label $I\in X$ for the trivial observable, equivalently, a least element $\mC_I$ in $\cC(\cM)$.

The collection of all (not necessarily maximal) contexts is sometimes also called a \emph{marginal scenario $\cC(\cM)\subset 2^X$} \cite{ChavesFritz2012,FritzChaves2013,BudroniEtAl2022} which satisfies (i) in Def.~\ref{def: measurement scenario} and instead of (ii) is closed under subsets, that is, if $\mC,\mC'\in\cC$ with $\mC'\subset\mC$ then $\mC'\in\cC$. A measurement cover encodes which measurement labels in $X$ are compatible. This information can further equivalently be encoded in terms of the \emph{compatibility graph $\cG(X)=(X,\circ\subset X\times X)$ of $X$} where $x\circ x'$ if and only if $x,x'\in\rC\in\cM$. Indeed, given $\cG(X)$ we recover $\cM$ from $X\supset\rC\in\cM$ if and only if $x\circ x'$ for all $x,x'\in\rC$ and $\rC$ is maximal.\footnote{For two self-adjoint operators $x,x'\in\LHsa$ in quantum theory, we have $x\circ x'\Leftrightarrow [x,x']=0$.} In essence, measurement covers, marginal scenarios and compatibility graphs all define a partial order, represented in terms of subsets of a set (measurement labels $X$ of $M$ or vertices $V$ of $\cG(X)$). Clearly, this partial order is the analogue of the partial order of contexts $\CO$ for observable algebras $\cO$. However, unlike observable algebras, measurement scenarios and compatibility graphs are concerned with a particular choice of measurements. In general, this adds a level of redundancy to the presentation, which as we now show relates to a choice of (generally non-minimal) generating set $\cX(\cO)$ of $\cO$.\\

\textbf{Measurement scenarios as generating sets of observable algebras.} At first sight, by removing the full algebraic structure in commutative algebras $C\in\CO$ of an observable algebra $\cO$, the sheaf-theoretic approach seems to achieve a significant reduction and thus a more economical representation of a ``contextuality scenario". However, this intuition turns out to be misleading as we will show in this section.

To see this, we first show that a (finitely generated) observable algebra $\cO$ specifies a measurement scenario through (subsets of observables that constitute) a generating set.

\begin{definition}\label{def: generating set}
    Let $\cO$ be a finitely generated observable algebra with partial order of contexts $\CO$. Then $\cX(\cO) = \{O_1,\cdots,O_N\}$ is called a \emph{generating set for $\cO$} if $\cO\ni O=\sum_{k=1}^{|S_O|} c_kh_k(O_k)$ for some $c_k\in\R$, $h_k:\R\ra\R$ and $O_k\in\cX(\cO)$. A generating set $\cX(\cO)$ is called \emph{minimal} if $\cX'(\cO)\subset\cX(\cO)$ is a generating set for $\cO$ only if $\cX'(\cO)=\cX(\cO)$.
\end{definition}

\begin{definition}\label{def: associated measurement scenario}
    Let $\cO$ be a finitely generated observable algebra with maximal set of contexts $\COm$ and generating set $\cX(\cO)$. Then $M(\cX(\cO)) = (X,\cM)$ with $X=\cX(\cO)$ and $\cM=\COm$ is called the \emph{measurement scenario associated with $\cX(\cO)$}.
\end{definition}

Clearly, the measurement scenario associated with an observable algebra $\cO$ depends on the choice of generating set $\cX(\cO)$. At the same time, by the definition of a generating set, each choice defines the same partial order of contexts. Since the latter is the only relevant structure (from the perspective of KS contextuality), we see that specifying a partial order of contexts in terms of measurement scenarios is generally redundant.

Nevertheless, there are certain canonical choices of generating sets, which we have encountered before: (i) for every maximal context $C\in\COm$, let $O \in C$ be a non-degenerate observable such that $C = C(O) = \{h(O) \mid h:\R\ra\R\}$. We have been using such generating sets in the arguments leading up to Thm.~\ref{thm: CNC for OAs}. Recall that KS noncontextuality may be seen as the impossibility to choose a generating observable in every maximal context. (ii) for every minimal context $C\in\COmin$, let $O\in C$ be a non-degenerate observable such that $C=C(O)=\{h(O) \mid h:\R\ra\R\}$. Both $\cX(\cO)=\{O_C\}_{C\in\COmin}$ and $\cX'(\cO)=\{O_C\}_{C\in\COm}$ are minimal generating sets for $\cO$. In particular, (the set of elementary observables corresponding with) $\POone$ forms a minimal generating set for maximal observable algebras. More generally, $\POm$ is a minimal generating set for coarse-grainged observable algebras (see Lm.~\ref{lm: reduction to coarse-grained OA} below). Generating sets of this form (for a canonical choice of outcome sets $\{0,1\}$) justify the (hyper)graph-theoretic approach \cite{CabelloSeveriniWinter2014,ChavesFritz2012,FritzChaves2013,AcinFritzLeverrierSainz2015,AmaralCunha2018} (see App.~\ref{app: orthogonality graphs vs OAs} below).

However, the above are not the only minimal generating sets. Indeed, there are even infinitely many generating sets already of the form (i) and (ii), whose observables merely differ by a choice of invertible function $\{h_C:\R\ra\R\}_{C\in\CO}$, that is, by a choice of re-labelling of their outcome sets. Such re-labellings play no role in the sheaf-theoretic formalism to contextuality, and may thus be considered as ``gauge". This again emphasises that functional relations between observables may be seen as part of the side-processing of an experiment. Removing this ambiguity therefore naturally leads back to the perspective of observable algebras and their partial order of contexts defined in terms of commutative subalgebras as presented here.\\

\textbf{Observable algebras associated with measurement scenarios.} We complete our comparison between our formalism and the sheaf-theoretic one by showing that observable algebras subsume measurement scenarios, that is, given a measurement scenario $M = (X,\cM)$, there always exists an observable algebra $\cO$ and a generating set $\cX(\cO)$ such that $M=M(\cX(\cO))$. One way to see this is by noting that any partial order with least element corresponds to the partial order of contexts of some observable algebra. An explicit construction is given in the following lemma.

\begin{lemma}\label{lm: OA-representation of measurement scenarios}
    Let $M=(X,\cM)$ be a measurement scenario. Then there exists a unique observable algebra $\cO$ with generating set $\cX(\cO)$ such that $M = M(\cX(\cO))$.
\end{lemma}

\begin{proof}
    We need to represent the measurement labels $x\in X$ by discrete random variables with event spaces ordered according to the partial order of measurement contexts, defined by the marginal scenario $\cC(\cM)$ of $M$. More precisely, for every measurement context $\cC(\cM) \ni \mC \subset X$, we need to define a commutative algebra $C$, whose generating projections define a Boolean algebra (of events) $\PC$ such that the measurement label $x\in X$ is represented as a random variable $O:\PC\ra o_x\subset\R$. Moreover, the inclusion relations between these commutative algebras must reflect the partial order of measurement contexts, that is, $\tC \subset C$ if and only if $\tilde{\mC} \subset \mC$ for all $\tilde{\mC},\mC \in \cM$.
    
    To do so, we proceed by induction in the cardinality of measurement contexts $\mC\in\cC(\cM)$ in the measurement cover $\cM$. Since we assume that $I\in\mC$ for all $\mC\in\cM$, the initial step is to represent the trivial measurement context $\mC_I=\{I\}$ (with $|\mC|=1$) by the Boolean algebra $\cP(C_I)\cong \mathbf{2}$, and to define the commutative algebra $C_\One = \{\lambda \One \mid \lambda\in\R\}$ which will be included in every other commutative algebra containing the identity $\One\in\cO$.
    
    In the inductive step, we may assume to have represented all measurement contexts $\mC\in\cC(\cM)$ of cardinality smaller or equal to $n\in\mathbb{N}$ by commutative algebras $C$ such that $\widetilde{C}\subset C$ if and only if $\widetilde{\mC}\subset\mC$. Let $\mC\in\cC(\cM)$ be a measurement context of cardinality $|\mC|=n+1$ and let $\{\widetilde{\mC}_k\}_{k=1}^S$ with $|\widetilde{\mC}_k|=m_k\leq n$ for all $1\leq k\leq S$ be the set of measurement subcontexts $\widetilde{\mC}_k\subset\mC$ in $\cC(\cM)$. We represent $\mC$ by the minimal commutative algebra $C$ such that there exist inclusion relations $i_{C\tC_k}:\tC_k\hra C$, subject to the conditions $i_{C\tC_k}(\tC_k)\neq i_{C\tC_l}(\tC_l)$ whenever $k \neq l$ for $1\leq k,l \leq S$ and
    \begin{align*}
        i_{C\tC_k}|_{\tC_l}(\tC_k\cap\tC_l) = i_{C\tC_l}|_{\tC_k}(\tC_k\cap\tC_l) = i_{C\tC_k}(\tC_k) \cap i_{C\tC_l}(\tC_l)\; .
    \end{align*}
    Clearly, in this way we obtain commutative algebras $C$ for every measurement context $\mC\subset\cM$ whose inclusion relations define a partial order $\CO$ which faithfully represents the measurement scenario generated by $\cM$, that is, $\tC\subset C$ if and only if $\widetilde{\mC}\subset\mC$. Since we assume the commutative algebras $C$ to be minimal, we thus obtain a unique observable algebra $\cO=\bigcup_{C\in\CO} C$, where every measurement label $x\in X$ is represented by a random variable $O(x):\PC\ra o_x\subset\R$; and it is straightforward to see that $\{O(x)\}_{x\in X}$ forms a generating set for $\cO$ with $M=M(\{O(x)\}_{x\in X})$.
\end{proof}

Note that the generating set $\{O(x)\}_{x\in X}$ for $\cO$ in the proof of Lm.~\ref{lm: OA-representation of measurement scenarios} is generally not minimal, as it may contain observables $O(x)$ that can be written as linear combinations of functions of other observables in the set. In this case, we may remove the respective measurement labels from $X$ without changing the construction in Lm.~\ref{lm: OA-representation of measurement scenarios}. Consequently, such measurement labels are redundant for the purposes of studying contextuality. Moreover, we reiterate from above that even specifying a minimal generating set is generally ambiguous given re-labellings of outcome sets.\\

In summary, up to a choice of generating set, every measurement scenario is represented by an observable algebra such that the respectively defined partial orders of contexts coincide. Hence, our formalism subsumes, and generally provides a more economical representation than, the sheaf-theoretic approach \cite{AbramskyBrandenburger2011,AbramskyEtAl2015,AbramskyEtAl2017} (see also Sec.~\ref{app: reduced context categories}).

\subsection{Orthogonality graphs vs partial algebras}\label{app: orthogonality graphs vs OAs}

The generating sets of type (ii) in App.~\ref{app: measurement scenarios vs OAs}, for fixed choice of outcomes sets $o=o_x=\{0,1\}$, are given by $\POm$, that is, they correspond with the event algebra $\PO$ of the observable algebra $\cO$ defined in Lm~\ref{lm: OA-representation of measurement scenarios}. Since, by Lm.~\ref{lm: embeddings vs orthomorphisms vs order-preserving maps}, $\PO$ encodes the same context structure as $\cO$, it constitutes an economical representation of the contextuality scenario described by $\cO$. For general observable algebras, the information contained in $\PO$ is readily encoded as a hypergraph \cite{ChavesFritz2012,FritzChaves2013,AcinFritzLeverrierSainz2015}.$^{\ref{fn: identity element in hypergraphs}}$ What is more, under the restriction to partial algebras (which also satisfy Specker's principle), we can further restrict to simple (``orthogonality") graphs (see Def.~\ref{def: orthogonality graph}), that is, to the orthogonal relations between minimal elements in $\POm$. For the ensuing comparison between our formalism and that of orthogonality graphs, we will therefore (mostly) restrict to partial algebras.

Now, every partial algebra defines a unique orthogonality graph (see Def.~\ref{def: associated orthogonality graph}). However, the converse is generally false, since orthogonality graphs lack a notion of identity element. In the following, we first show that this implies that (classical) correlations in Eq.~(\ref{eq: cl-qu-nd correlations}) are generally incompatible with ``coarse-graining'', and then provide various ways to encode this additional structure in order to uniquely identify an orthogonality graph with an event algebra. The issue of sub-normalisation has been pointed out in Ref.~\cite{AcinFritzLeverrierSainz2015}, and further distinguishes the graph- from the hyper-graph approach. Here, we highlight those aspects necessary for our treatment of SI-C sets and graphs in Sec.~\ref{sec: SI-C sets and graphs}.\\

\textbf{Embedding and coarse-graining of states and correlations.} We compare states on partial algebras in Eq.~(\ref{eq: cl-qu-nd states}) with correlations defined on orthogonality graphs in Eq.~(\ref{eq: cl-qu-nd correlations}), with respect to their compatibility under embedding and coarse-graining. In particular, it is natural to require that the state spaces $\Gamma_X[\tCO]$ for $X=\mathrm{cl},\mathrm{qu},\mathrm{nd}$ do not depend on whether we consider $\tcO$ as describing a closed system or whether it arises as part of a larger system which is described by a partial algebra $\cO\supset\tcO$, that is, whether we first embed $\tcO$ within a larger partial algebra, followed by coarse-graining. (Similarly, for correlations on orthogonality graphs $\tG\subset G$.)

To see whether this is the case, note first that for a general embeddings $\epsilon:\tcO\ra\cO$, we have $\Gamma_X[\tCO]\not\subset\Gamma_X[\CO]$ for $X=\mathrm{cl},\mathrm{qu},\mathrm{nd}$, since $\cO$ will generally impose more constraints on states than $\tcO$ does, that is, the marginal problem posed by $\cO$ is generally more restrictive than that posed by any observable subalgebra $\tcO\subset\cO$.\footnote{Of course, we also have $\Gamma_X[\CO]\not\subset\Gamma_X[\cC(\tcO)]$ for $X=\mathrm{cl},\mathrm{qu},\mathrm{nd}$ since states on $\cO$ generally contain more information than states $\tcO$.} For instance, consider a classical observable algebra $\tcO = L_\infty(\Lambda)$ with $\Lambda$ a finite set of cardinality $|\Lambda|=d$. Clearly, in this case $\Gamma_\mathrm{cl}[\CO]=\Gamma[\CO]$. Now, consider the embedding $\tcO\ra\cO$ with $\cO=\LHsa$ for a Hilbert space of dimension $\dim(\cH)=d$, which maps the elements in $\Lambda$ to an orthogonal basis of $\cH$. By the Kochen-Specker theorem, $\Gamma_\mathrm{cl}[\CO]=\emptyset$ whenever $d\geq 3$. Consequently, we find $\Gamma_\mathrm{cl}[\tCO]\not\subset\Gamma_\mathrm{cl}[\CO)]$ in this case. A similar argument will show that $\Gamma_\mathrm{qu}[\cC(\tcO)]\not\subset\Gamma_\mathrm{qu}[\CO]$.

We call an embedding $\epsilon:\tcO\ra\cO$ \emph{constraint-preserving} if $\epsilon(p)\epsilon(p')=0$ only if $pp'=0$ for all $p,p'\in\cP(\tcO)$,\footnote{Equivalently, an embedding $\epsilon$ is constraint-preserving if $\eta$ in Lm.~\ref{lm: embeddings vs orthomorphisms vs order-preserving maps} is a fully faithful functor.} and write $\tcO\subset_\mathrm{c-p}\cO$ in this case.

\begin{lemma}\label{lm: fully faithful embedding of states on OAs}
    Let $\tcO\subset_\mathrm{c-p}\cO$ be two observable algebras. Then $\Gamma_X[\tCO]\subset\Gamma_X[\CO]$, where $X=\mathrm{cl},\mathrm{qu},\mathrm{nd}$ denotes classical, quantum and non-disturbing states, respectively.
\end{lemma}

\begin{proof}
    By definition, two contexts $\tC,\tC'\in\tCO$ are related $\tC'\subset\tC$ in $\tCO$ if and only if $\epsilon(\tC')\subset\epsilon(\tC)$ are related in $\CO$. It follows that contexts $\CO\ni C\notin\tCO$ pose no additional marginalisation constraints on states in Eq.~(\ref{eq: no-disturbance}).
\end{proof}

Clearly, we also have $\Gamma_X[\CO]|_\tcO\subset\Gamma_X[\CO]$ for $X\in\{\mathrm{cl,qu,nd}\}$, where $|_{\tcO}$ denotes marginalisation. In other words, (classical, quantum) states are compatible with marginalisation. Consequently, $\Gamma_\mathrm{X}[\cC(\tcO)]=\Gamma_\mathrm{X}[\CO]|_\tcO$ whenever $\tcO\subset_\mathrm{c-p}\cO$.

Finally, we observe that compatibility with marginalisation already implies a partial consistency condition with respect to general subsystem embeddings.

\begin{lemma}\label{lm: embedding of correlations on G OA}
    Let $\epsilon:\tcO\ra\cO$ be an embedding. Then $\Gamma_\mathrm{X}[\CO]|_{\tcO}\subset\Gamma_\mathrm{X}[\cC(\tcO)]$, where $X=\mathrm{cl},\mathrm{qu},\mathrm{nd}$ denotes classical, quantum and non-disturbing states, respectively.
\end{lemma}

\begin{proof}
    Clearly, $\Gamma_\mathrm{nd}[\CO]|_{\tcO}\subset\Gamma_\mathrm{nd}[\cC(\tcO)]$. Moreover, every classical, respectively quantum embedding $\epsilon':\cO\ra\cO'$ induces a classical, respectively quantum embedding $\epsilon'|_\tcO$ of $\tcO$ by restriction. It therefore follows from the definition of the state spaces $\Gamma_X[\CO]$ for $X=\mathrm{cl},\mathrm{qu}$ via such embeddings in Eq.~(\ref{eq: cl-qu-nd states}) that marginalisation of classical, respectively quantum states $\gamma\in\Gamma_X[\CO]$ defines classical, respectively quantum states on $\tcO$, that is, $\gamma|_\tcO\in\Gamma_X[\tCO]$.
\end{proof}

Differences between $\Gamma_\mathrm{X}[\cC(\tcO)]$ and $\Gamma_\mathrm{X}[\CO]|_\tcO$ in a physical theory reveal a mismatch with respect to subsystem embeddings. In turn, by demanding full consistency with subsystem embeddings, that is, $\Gamma_\mathrm{X}[\CO]|_{\tcO}=\Gamma_\mathrm{X}[\cC(\tcO)]$ for $X=\mathrm{cl},\mathrm{qu},\mathrm{nd}$, one may sometimes rule out states on certain subsystems, for instance, one may rule out non-quantum states \cite{PopescuRohrlich1994} on two-dimensional systems under consistency with quantum embeddings.\\

Next, we turn to correlations on orthogonality graphs in Eq.~(\ref{eq: cl-qu-nd correlations}). Specifically, given a subgraph $\tG\subset G$ we seek a notion of ``coarse-graining'' for correlations that corresponds with the restriction to subgraphs. Note first that in analogy with the case of states on observable algebras the marginal problems defined by $G$ and $\tG$ also differ in two opposing ways. On the one hand, $G$ generally admits more correlations than $\tG$ since it contains more vertices and thus correlations on $G$ encode more information than those on $\tG$. On the other hand, $G$ generally poses more constraints on the respective marginal problem than $\tG$ does. Consequently, in general $\QStab(G)\not\subset\QStab(\tG)$ and $\QStab(\tG)\not\subset\QStab(G)$, and similarly for the restriction to quantum and classical correlations.

Nevertheless, the following relations hold for vertex-restricted subgraphs $\tG=(\tV,\tE)\subset_V G=(V,E)$ defined by $\tV\subset V$ and $\tE=E|_\tV=\{(\tv,\tv')\in E\mid (\tv,\tv')\in\tV\}$, analogous to constraint-preserving embeddings between observable algebras in Lm.~\ref{lm: fully faithful embedding of states on OAs}.

\begin{lemma}\label{lm: embedding of correlations on G}
    Let $\tG,G$ be two orthogonality graphs with $\tG\subset_V G$. Then
    \begin{itemize}
        \item[(i)] $\QStab(\tG)\subset\QStab(G)$,
        \item[(ii)] $\Stab(\tG)\subset\Stab(G)$,
        \item[(iii)] $\Th(\tG)\subset\Th(G)$.\footnote{Here, we allow for realisations $\pi:G\ra\PH$ of $G=(V,E)$ with $\pi(v)\pi(v')=0$ even if $(v,v')\notin E$.}
    \end{itemize}
\end{lemma}

\begin{proof}
    (i) Since cliques in $\tG$ are cliques in $G$, we find $\QStab(\tG)\subset\QStab(G)$.

    (ii) Let $\tg\in\Stab(\tG)$,\footnote{Unlike before, here we write $\gamma\in\QStab(G)$, $\tg\in\QStab(\tG)$ for correlations on $G$, $\tG\subset G$, respectively.} that is, there exist weights $\rw_\tg(\tI)\in[0,1]$ over independent sets $\tI\in\cI(\tG)$ with $\sum_{\tI\in\cI(\tG)}\rw_\tg(\tI)=1$, $\tg(\tv)=\sum_{\tI\in\cI(\tG)} \rw_\tg(I)\tI(\tv)$ for all $\tv\in\tV$ and such that $\sum_{\tv\in\widetilde{\mC}}\sum_{I\in\cI(\tv)}\rw_\tg(\tI)\leq 1$ for all (maximal) cliques $\widetilde{\mC}\subset\tG$. But since $\cI(\tG)\subset\cI(G)$, we also have $\tg\in\Stab(G)$, hence, $\Stab(\tG)\subset\Stab(G)$.

    (ii) Let $\tg\in\Th(\tG)$, that is, there exists a realisation $\pi:\tG\ra\cP(\tH)$ and a quantum state $\trho\in\tSH$ such that $\tg(\tv)=\tr[\trho\tpi(\tv)]$ for all $\tv\in\tV$. Since $\tG\subset_V G$ with $\tE=E|_\tV$ we can extend $\tpi$ to a realisation $\pi:G\ra\PH$ with $\dim(\cH)=\dim(\tH)+|V|/|\tV|$, by mapping the vertices in $V/\tV$ onto mutually orthogonal rank-$1$ projections in the subspace orthogonal to $\tpi(\tV)$. Clearly, in this case $\tpi=\pi|_\tH$ as well as $\trho=\rho|_\tH$ under the canonical embedding $\tH\hra\cH$, hence, $\tg\in\Th(G)$.
\end{proof}

What about the (partial) consistency condition with respect to subsystem embeddings in Lm.~\ref{lm: embedding of correlations on G OA}? In particular, do correlations on $G$ marginalise to correlations on $\tG$ for subgraphs $\tG\subset G$, specifically for vertex-restricted subgraphs $\tG\subset_V G$?
 
Note first that there is no natural notion of marginalisation given merely the structure of an orthogonality graph.\footnote{Indeed, this requires a notion of coarse-graining of events, such as given by the event algebra of an observable algebra or by hyperedges in a hypergraph \cite{AcinFritzLeverrierSainz2015}.} Nevertheless, in the absence of a canonical marginalisation map, we may call any map $|_\tG:\QStab(G)\ra\QStab(\tG)$ a \emph{valid marginalisation map} if (i) it preserves convex decompositions, that is, $\gamma|_\tG=p\gamma_1|_\tG+(1-p)\gamma_2|_\tG$ for every $\gamma_1,\gamma_2,\gamma=p\gamma_1+(1-p)\gamma_2\in\QStab(G)$ and $p\in[0,1]$, and (ii) it acts trivially on $\QStab(\tG)$, that is, $(|_\tG)|_{\Stab(\tG)}=\mathrm{id}$. For non-disturbing and quantum states, such maps generally exist: let $\tG\subset_V G$ be a vertex-restricted subgraph of $G$, e.g. with $V=\tV\cup\{v\}$. Given a unital (quantum) realisation $\pi:G\ra\POone$, we obtain a unital (quantum) realisation $\tpi:\tG\ra\PO$ by setting $\tpi(v_\trC)=\pi(v)+\pi(v_\trC)$ for some choice of vertex $v_\trC\in\trC$ in every maximal clique $C=\trC\cup\{v\}\subset G$ (such that if $v_\trC\in\rC'=\trC\cup\{v\}$, then $v_{\trC'}=v_\trC$) and $\tpi=\pi$ otherwise. This is easily seen to define a valid marginalisation map.\footnote{Here, we assume that correlations are normalised, that is, we assume equality in Eq.~(\ref{eq: qstab}).} In particular, this construction applies to graph completions $G\subset_V \oG$.

However, there is generally no marginalisation map for classical correlations. The reason is that upon deleting a vertex $v\in V$ from $G$ there is generally no way to ``re-distribute" the weight $w_\gamma(I_v)$ with $I_v(v')=\delta_{vv'}$ of a correlation $\gamma\in\QStab(G)$ to a coarse-grained correlation in $\QStab(\tG)$. More precisely, note that $I\in\Stab(G)$ for every $I\in\cI(G)$. Hence, for independent sets $I\in\cI(G)$ with $I(\tv)=1$ for some $\tv\in\tV$, a natural way to define marginalisation is by restriction, that is, $I|_\tG\in\Stab(\tG)$ and thus $w_{\gamma|_\tG}(I|_\tG)=w_\gamma(I)$, where $I|_\tG(\tv)=I(\tv)$ for all $\tv\in\tG$. However, for $I_v$ with $V\ni v\notin\tV$ there is no natural way to assign a probability weight and we may thus assign it any element $I_v|_\tG\in\Stab(\tG)$. Yet, if $\gamma\in\Stab(G)$ is a classical correlation such that $w_\gamma(I_v)>0$, then $\sum_{I|_\tG\in\cI(\tG)}w_\gamma|_\tG(I|_\tG)=\sum_{I_v\neq I|_\tG\in\cI(\tG)}w_\gamma|_\tG(I|_\tG)+w_\gamma(I_v)=1+w_\gamma(I_v)>1$.

\begin{proposition}\label{prop: no marginalisation}
    There exists no marginalisation map $|_\tG:\Stab(G)\ra\Stab(\tG)$ for $\tG\subset G$. Consequently, $\Stab(G)$ is incompatible with subsystem embeddings, that is, $\Stab(G)|_\tG\neq\Stab(\tG)$, in particular, $\Gamma_\mathrm{cl}[\cC(\cO(\oG)]|_G\neq\Stab(G)$ for $\pi:G\ra\PO$.
\end{proposition}

\def\boxit#1{%
  \smash{\color{red}\fboxrule=1pt\relax\fboxsep=2pt\relax%
  \llap{\rlap{\fbox{\vphantom{0}\makebox[#1]{}}}~}}\ignorespaces
}

\begin{table}[htb!]
    \centering
    \begin{tabular}{c|cccc}
        \toprule
        states & $\bullet(\tO)\subset\bullet(\cO)$ & $\bullet(\cO)|_\tO\subset\bullet(\cO)$ & $\bullet(\cO)|_\tO\subset\bullet(\tO)$ & $\bullet(\tO)\subset\bullet(\cO)|_\tO$ \\
        \midrule
        $\bullet=\Gamma[\cC(\cdot)]$ & \checkmark & \checkmark & \checkmark & \checkmark \\
        $\bullet=\Gamma_\mathrm{qu}[\cC(\cdot)]$ & \checkmark & \checkmark & \checkmark & \checkmark \\
        $\bullet=\Gamma_\mathrm{cl}[\cC(\cdot)]$ & \checkmark & \boxit{6.2cm} \checkmark & \checkmark & \checkmark \\
        \bottomrule
        \toprule
        correlations & $\bullet(\tG)\subset\bullet(G)$ & $\bullet(G)|_\tG\subset\bullet(G)$& $\bullet(G)|_\tG\subset\bullet(\tG)$ & $\bullet(\tG)\subset\bullet(G)|_\tG$ \\
        \midrule
        $\bullet=\QStab(\cdot)$ & \checkmark & ? (\checkmark) & ? (\checkmark) & ? (\checkmark) \\
        $\bullet=\Th(\cdot)$ & \checkmark & ? (\checkmark) & ? (\checkmark) & ? (\checkmark) \\
        $\bullet=\Stab(\cdot)$ & \checkmark & 
        \boxit{6.2cm} \xmark & \xmark & \xmark \\
        \bottomrule
    \end{tabular}
    \caption{Compatibility of (top) states on observable (specifically, partial) algebras as opposed to (bottom) correlations on orthogonality graphs, with respect to constraint-preserving embeddings $\tO\subset_\mathrm{cp}\cO$, respectively vertex-restricted subgraphs $\tG\subset_V G$, as well as coarse-graining and subsystem embeddings. Contrary to classical states in Lm.~\ref{lm: embedding of correlations on G OA}, by Prop.~\ref{prop: no marginalisation} there is no notion of marginalisation and thus of subsystem embedding for classical correlations.}
        \label{tab: states under subsystem embeddings}
\end{table}

Prop.~\ref{prop: no marginalisation} is a consequence of the different normalisation condition imposed on classical correlations (in Eq.~(\ref{eq: stab})) as compared to the normalisation condition imposed on states on observable algebras. This shows that the respective notions are generally incomparable (see also App.~\ref{app: correction}). If we nevertheless want to compare correlations on orthogonality graphs $G$ with states on partial algebras $G\subset\POm$, we thus need to add additional information to the graph $G$ in the form of an identity. The obvious way to enforce normalisation is by restricting realisations of $G$ to be unital (see Def.~\ref{def: unital realisation}). Yet, while every orthogonality graph admits a realisation in some (quantum) partial algebra \cite{HeunenFritzReyes2014}, to determine whether it admits a unital realisation remains an open problem (see App.~\ref{app: On realisations}).

In summary, orthogonality graphs correspond with the event algebras of partial algebras if and only if they admit a unital realisation (see Tab.~\ref{tab: realisation vs orthomoprhism vs embedding}). While the existence problem of such realisations is notably harder than that of generic realisations (see App.~\ref{app: On realisations}), in Thm.~\ref{thm: n-colourable implies USI-C} we find that under the restriction to unital realisations the chromatic number of (the maximal extensions of) an orthogonality graph (see Def.~\ref{def: maximal extended OG}) is a complete invariant of KS contextuality. What is more, in Thm.~\ref{thm: RH conjecture} we prove that the chromatic number of (the extension of) an orthogonality graph with freely completable realisation (see Def.~\ref{def: faithful graph completion}) is a complete invariant of state-independent contextuality.

\subsection{Reduced context categories}\label{app: reduced context categories}

We finish this section by pointing out further redundancies in the presentation of observable algebras. More precisely, we will prove a number of reduction arguments that allow us to reduce a given observable algebra to a simpler one, which is Kochen-Specker contextual if and only if the original one is. By Thm.~\ref{thm: CNC for OAs}, this is the case as long as we do not change the constraints in Eq.~(\ref{eq: CNC}). By comparison with measurement scenarios and (hyper)graphs, and the context categories (the partial order of contexts) they generate as shown in App.~\ref{app: measurement scenarios vs OAs} and App.~\ref{app: orthogonality graphs vs OAs}, we will find that these frameworks generally contain redundant information from the perspective of KS contextuality.\\

\textbf{Downward generated context categories.}

\begin{definition}\label{def: max gen OAs}
    Let $\cO$ be an observable algebra with context category $\CO$. We call $\CO\supset\COdown:=\{\tC\in\CO\mid\exists C,C'\in\COm:\tC=C\cap C'\}$ the \emph{downward generated context category}, and $\cOdown \subset \cO$ the \emph{downward generated observable algebra of $\cO$.} 
\end{definition}

Note that $\COm\subset\COdown$ (pick $C'=C$ in Def.~\ref{def: max gen OAs}). Since a classical observable algebra with $\cO=L_\infty(\Lambda)$ for a measurable space $\Lambda$ has a single maximal context, this immediately implies that $\COdown=\CO$ for classical systems.

\begin{lemma}\label{lm: sparse equals full in QM}
    Let $\cO=\LHsa$ with $\dim(\One)=\dim(\cH)=d$. Then $\COdown=\CO$.
\end{lemma}

\begin{proof}
    Let $\tC\in\CO$ be a non-maximal context in $\CH$ with generating projections $\{\tp_1,\cdots,\tp_m\}$, that is, $\mathbbm{1}=\sum_{k=1}^m\tp_k$, $\sum_{k=1}^m d_k=d$ and $d_k>1$ for some $k$. Consequently, there exist maximal contexts $C,C'\in\COm$ generated by one-dimensional projections
    \begin{align*}
        \cP_1(C) &= \{p_{11},\cdots,p_{1d_1},p_{21},\cdots,p_{2d_2},p_{31},\cdots,p_{md_m}\}\; ,\\
        \cP_1(C') &= \{p'_{11},\cdots,p'_{1d_1},p'_{21},\cdots,p'_{2d_2},p'_{31},\cdots,p'_{md_m}\}\; ,
    \end{align*}
    with $\tp_k=\sum_{j=1}^{d_k} p_{kj} = \sum_{j=1}^{d_k} p'_{kj}$ $p_{ik} \neq p'_{jl}$ for all $k,l\in\{1,\cdots,m\}$ and such that the contexts $C_k=C(p_{k1},\cdots,p_{kd_k},\mathbbm{1})$, $C'_k=C(p'_{k1},\cdots,p'_{kd_k},\mathbbm{1})$ generated by the respective resolutions of $\tp_k$ satisfy $C_k\cap C'_k=C(\tp_k,\mathbbm{1})$ for all $k\in\{1,\cdots,m\}$. It follows that $C\cap C'=\tC$, hence, $\tC\in\COdown$. Since this holds for all $\tC\in\CO$, the result follows.
\end{proof}

Since every finite-dimensional quantum system is a direct sum of matrix algebras, the result holds for all finite-dimensional quantum systems. However, for general observable algebras, including more general subalgebras $\cO\subset\LHsa$ we have $\cOdown\subsetneq\cO$. Namely, if $\tC\in\CO$ is such that $\tC\subsetneq C\cap C'$ for all $C,C'\in\COm$ then $\tC\notin\COdown$.

\begin{lemma}\label{lm: reduction to max gen OAs}
    Let $\cO$ be an observable algebra with context category $\CO$. Then $\cO$ is Kochen-Specker contextual if and only if $\cOdown$ is Kochen-Specker contextual.
\end{lemma}

\begin{proof}
    By Lm.~\ref{lm: KSNC inherits from maximality}, it is sufficient to establish the claim for $\cO$ maximal.\footnote{This follows from condition (iii) on maximal extensions in Def.~\ref{def: max extension} (see App.~\ref{app: non-maximal OAs}).\label{fn: constraint-preserving}}
    
    ``$\Rightarrow$'': Since $\COdown \subset \CO$, this direction is obvious.

    ``$\Leftarrow$'': Let $\gamma = (C_0,\cdots,C_{n-1})$ be a context cycle in $\CO$. Since, $C_{i+1i} = C_{i+1}\cap C_i\in\COdown$ for all $C_i,C_{i+1}\in\COm$ by definition, $\gamma$ is also a context cycle in $\COdown$. Hence, $\cO$ poses no additional constraints and so inherits KS noncontextuality from $\cOdown$.
\end{proof}

Lm.~\ref{lm: reduction to max gen OAs} highlights redundant structure (from the perspective of KS contextuality) in general observable algebras, e.g. those generated from measurement scenarios (see App.~\ref{app: measurement scenarios vs OAs}). A similar result holds for the other direction.\\

\textbf{Upward generated observable algebras.}

\begin{definition}\label{def: min gen OAs}
    Let $\cO$ be an observable algebra with context category $\CO$. We call $\CO\supset\COup:=\{\tC\in\CO \mid \exists C,C'\in\COmin:\tC=C\cup C'\}$ the \emph{upward generated context category}, and $\cOup \subset \cO$ the \emph{upward generated observable algebra of $\cO$.} 
\end{definition}

Minimal contexts are defined in the same way as maximal ones. However, for Def.~\ref{def: min gen OAs} not to be trivial, we need to exclude the least element in $C_\One$ in $\CO$.

In the extreme case, minimal contexts are generated by a single projection. Finitely generated classical and quantum observable algebras are of this form.

\begin{lemma}\label{lm: min gen = full in QM}
    Let $\cO=\cAsa\subset\LHsa$ with $\dim(\cH)=d$. Then $\COup=\CO$.
\end{lemma}

\begin{proof}
    Clearly, every maximal context $C\in\CHm$ is generated by one-dimensional projections $\cP_1(C) = \{p_1,\cdots,p_d\}$ for $p_i\in\cP_1(\cH)$, $\sum_{i=1}^dp_i=\mathbbm{1}$ and $p_ip_j=\delta_{ij}$.
\end{proof}

Since (finitely generated) classical observable algebras arise as a special case with $\cAsa=\oplus_{k=1}^N M_1(\C)$, Lm.~\ref{lm: min gen = full in QM} applies to both the classical and quantum case.

For general observable algebras, minimal contexts are not generated by a single projection, and we have $\cOup\subsetneq\cO$. Namely, if $\tC\in\CO$ is such that $C\cup C'\subsetneq\tC$ for all $C,C'\in\COmin$ then $\tC\notin\COup$. Nevertheless, for what regards Kochen-Specker contextuality, we may always reduce to upward generated observable algebras.

\begin{lemma}\label{lm: reduction to min gen OAs}
    Let $\cO$ be an observable algebra with context category $\CO$. Then $\cO$ is Kochen-Specker contextual if and only if $\cOup$ is Kochen-Specker contextual.
\end{lemma}

\addtocounter{footnote}{-1}

\begin{proof}
    By Lm.~\ref{lm: KSNC inherits from maximality}, it is again sufficient to establish the claim for $\cO$ maximal.$^{\ref{fn: constraint-preserving}}$

    ``$\Rightarrow$'': Since $\COup\subset\CO$, this direction is obvious.

    ``$\Leftarrow$'': Let $\gamma=(C_0,\cdots,C_{n-1})$ be a context cycle in $\CO$ with subcontexts $C_{i+1i}=C_{i+1}\cap C_i$ for all $i\in\zz_n$. By Def.~\ref{def: min gen OAs}, there exist contexts $C_{i+1i}\supset\tC_{i+1i}\in\COup$ and $C_i\supset\tC_i=\tC_{i+1i}\cup\tC_{ii-1}\in\COup$ such that
    $\tg=(\tC_0,\cdots,\tC_{n-1})$ is a context cycle in $\COup$. Moreover, Def.~\ref{def: min gen OAs} implies that the contexts $C_i\in\CO$ and $\tC_i\in\COup$ (respectively, $C_{i+1i}\in\CO$ and $\tC_{i+1i}\in\COup$) are distinguished by additional projections which are not related to projections other than in $\tC_i$ (respectively, $\tC_{i+1i}$). Consequently, $\cO$ is a (non-maximal) extension of $\cOup$ (similar to Def.~\ref{def: max extension}), and the result thus follows by analogous arguments to those in the proof of Lm.~\ref{lm: KSNC inherits from maximality}.
\end{proof}

Of course, we can combine the reductions in Def.~\ref{def: max gen OAs} and Def.~\ref{def: min gen OAs}.

\begin{lemma}\label{lm: min max reduction}
    Let $\cO$ be a finitely generated observable algebra with context category $\CO$. Then $((\cOup)_\cap)_\cup = (\cOup)_\cap = (\cOdown)_\cup = ((\cOdown)_\cup)_\cap$.
\end{lemma}

\begin{proof}
    Clearly, $((\cOup)_\cap)_\cup\subset(\cOup)_\cap$. Conversely, let $C\in\cC((\cOup)_\cap)$, then there exist two maximal contexts $\tC,\tC'\in\cC_{\max}((\cOup)_\cap)$ such that $C=\tC\cap\tC'$. We will show that there also exist $\ttC,\ttC'\in((\cOup)_\cap)_\cup$ such that $C=\ttC\cap\ttC'$. Let $\tC\supset\ttC\in ((\cOup)_\cap)_\cup$ and $\tC'\supset\ttC'\in ((\cOup)_\cap)_\cup$ be maximal. Since $\ttC\cap C=\tC\cap C$ (similarly, $\ttC'\cap C=\tC'\cap C$) for all $C\in((\cOup)_\cap)_\cup$ (unless $\tC\supset C \supset\ttC$, but then $C\notin((\cOup)_\cap)_\cup$) this implies $\ttC\cap\ttC'=\ttC\cap\tC'(=\tC\cap\ttC')=\tC\cap\tC'=C$. Consequently, $C\in((\cOup)_\cap)_\cup$ which implies $(\cOup)_\cap\subset((\cOup)_\cap)_\cup$ and thus $(\cOup)_\cap = ((\cOup)_\cap)_\cup$. Similarly, we find $((\cOdown)_\cup)_\cap=(\cOdown)_\cup$.
    
    Finally, for $\tC\in\cO$ and $\tC\supset\ttC\in\cOup$ maximal such that $\ttC=D\cup E$ for $D,E\in\cO$, we also have $\ttC\cap C=\tC\cap C$ for all $C\in\cOup$ (unless $\tC\supset C\supset\ttC$, but then $C\notin\cOup$). It follows that the two operations commute, hence, $(\cOdown)_\cup=(\cOup)_\cap$.
\end{proof}

After having reduced our observable algebra by means of Lm.~\ref{lm: reduction to max gen OAs} and Lm.~\ref{lm: reduction to min gen OAs}, we may further coarse-grain minimal generating sets as follows.

\begin{lemma}\label{lm: reduction to coarse-grained OA}
    Let $\cO$ be an observable algebra, $\PO$ its event algebra and  $\COmin$ the set of minimal contexts. Then $\CO$ is isomorphic to the context category generated by $\cP=\{\mathbbm{1}\}\cup \{p_C\}_{C\in\COmin}$ with $p_C=\sum_i p^C_i$ and $\{\mathbbm{1},p^C_i\}_i$ a minimal generating set of $C$,
    \begin{align}\label{eq: isomorphic context categories under coarse-graining}
        \CO=\cC(\PO)\cong\cC(\cP)\; .    
    \end{align}
    In particular, $\cO$ is KS noncontextual if and only if $\cOr=\cO(\cP)$ is KS noncontextual.
\end{lemma}

\begin{proof}
    Recall that $\POm$ is a (not necessarily minimal) generating set for $\cO$ (see App.~\ref{app: orthogonality graphs vs OAs}). Hence, Eq.~(\ref{eq: isomorphic context categories under coarse-graining}) holds by minimality of $\COmin$, that is, $\cC(\cO^\circ)\ni C\subset C(\{p^C_i\})$ implies $C=C(\{p^C_i\})$. But then the second statement follows since, by Thm.~\ref{thm: CNC for OAs}, KS noncontextuality is a property of the partial order of contexts only.
\end{proof}

Combining the arguments in this section, we obtain a more economical representation of ``contextuality scenarios" \cite{FritzChaves2013} represented by generic context categories.

\begin{theorem}\label{thm: optimal representation}
    Let $\cO$ be an observable algebra with partial order of contexts $\CO$. Then $\cO$ is KS noncontextual if and only if $((\cO^\circ_\cap)_\cup)_\mathrm{cg}=((\cO^\circ_\cup)_\cap)_\mathrm{cg}$ is KS noncontextual.
\end{theorem}

\begin{proof}
    This follows from Thm.~\ref{thm: KSNC of truncated OAs}, together with Lm.~\ref{lm: reduction to max gen OAs}, Lm.~\ref{lm: reduction to min gen OAs}, Lm.~\ref{lm: min max reduction} and Lm.~\ref{lm: reduction to coarse-grained OA}.
\end{proof}

Thm.~\ref{thm: CNC for OAs} shows that the essence of Kochen-Specker contextuality lies in the partial order of contexts $\CO$ \cite{DoeringFrembs2019a,Frembs2024a}, and extracts the relevant information in terms of constraints imposed on context connections, and expressed along context cycles (see Eq.~(\ref{eq: CNC})). Thm.~\ref{thm: optimal representation} reduces the representation of $\CO$ with respect to these constraints, and thus highlights redundancies in the representation of partial orders of contexts generated from generic sets of observables or events as in Def.~\ref{def: measurement scenario}. Moreover, the reduction in Thm.~\ref{thm: optimal representation} is optimal in that removing further contexts will generally reduce the number of constraints imposed by KS noncontextuality.
\section{On realisations of orthogonality graphs}\label{app: On realisations}

Thm.~\ref{thm: n-colourability for non-maximal OAs} characterises Kochen-Specker contextuality of an observable algebra in terms of the chromatic number of the orthogonality graph defined by its minimal projections (see also Thm.~\ref{thm: n-colourable implies USI-C}, Thm.~\ref{thm: partial RH conjecture} and Thm.~\ref{thm: RH conjecture}). In turn, this suggests a way to search for contextual arrangements in quantum theory, namely in terms of (critical) $k$-colourable graphs.\footnote{An (undirected) graph is called critical if removing any vertex or edge reduces its chromatic number. For instance, the 5-cycle in Ref.~\cite{YuOh2012} is $3$-critical, whereas the orthogonality graph corresponding to the provably minimal KS set in dimension $d\geq 4$ in Ref.~\cite{CabelloEstebanzGarcia-Alcaine1997} is not $4$-critical.} More generally, the (hyper-)graph approach \cite{ChavesFritz2012,FritzChaves2013,CabelloSeveriniWinter2014,AcinFritzLeverrierSainz2015,AmaralCunha2018} identifies additional graph-theoretic invariants for the study of contextuality.

An application of interest is to use such graph-theoretic invariants for the construction of novel contextual arrangements. In turn, this poses the following question:

\begin{problem*}
    Which simple graphs are realisable in a finite-dimensional observable algebra?
\end{problem*}

Of most interest are orthogonality graphs that are quantum realisabile.

\begin{problem*}
    Which simple graphs admit a finite-dimensional quantum realisation?
\end{problem*}

In this section, we discuss these problems and pose a number of (to our knowledge) open questions for refinements of the type of (quantum) realisation considered. In App.~\ref{app: orthogonality graphs vs OAs}, we have seen that an orthogonality graph is the event algebra of an observable algebra if and only if it has a unital realisation. This first raises the question when a unital realisation exists. We will consider both the general and the quantum case in turn.\\

\textbf{Unital realisations in general observable algebras.} The existence problem for general observable algebras is tied to its (minimal) dimension. Indeed, we obtain an embedding of $G$ in a (finitely generated) finite-dimensional observable algebra if there exists a dimension function $\dim:V\ra\mathbb{N}$ for $G=\GO$ (cf. App.~\ref{app: non-maximal OAs}) such that,
\begin{align}\label{eq: dim function constraints}
    \sum_{v\in\rC}\dim(v)=\dim(\One)\; ,
\end{align}
for all maximal cliques $\rC\subset G$. We can write this as a homogeneous system of linear Diophantine equations in $|V|+1$ variables of the form $Ax=0$, where $A\in M_{|C|\times (|V|+1)}(\mathbb{N})$ is the matrix given for all maximal cliques $\rC\subset G$ by $A_{\rC,0}=d$, $A_{\rC,v}=-1$ if $v\in\rC$ and $A_{\rC,v}=0$ otherwise. The existence of solutions to Eq.~(\ref{eq: dim function constraints}) can thus be checked efficiently by computing the Smith normal form of $A$. If a solution exists, $G$ has a realisation $\pi:G\ra\PO$ with $\cO$ a finite-dimensional observable algebra.  Moreover, a minimal realisation simply corresponds with solutions for which $d$ is minimal.\\

\textbf{Unital quantum realisations.} Given a quantum realisation $\pi:G\ra\PH$ of an orthogonality graph $G$ we obtain another quantum realisation by unitary conjugation $\pi_U=U\pi U^*$ for any $U\in\UH$. Both realisations contain the same quantum states and we may thus identify such realisations. However, $G$ may also admit quantum realisations that are not identifiable in this way.

\begin{problem}\label{prob: inequivalent realisations}
    Given a simple graph $G$, characterise all inequivalent unital realisations $\pi:G\ra\PH$ in dimension $\dim(\cH)=d$.
\end{problem}

Furthermore, one may ask how realisations in different dimensions are related. For instance, given a realisation $\pi:G\ra\PH$ with $\dim(\cH)=d$ we obtain another realisation in $\dim(\cH')=2d$ by doubling the dimension of every subspace. However, beyond such simple constructions, the relationship is not clear and a given orthogonality graph may thus lead to genuinely different realisations. Fortunately, the distinction between inequivalent unital realisations bears no relevance for the study of KS contextuality.

\begin{proposition}\label{prop: different unital realisations are KS equivalent}
    Let $G$ be an orthogonality graph with two unital realisations $\pi:G\ra\PH$ and $\pi':G\ra\cP(\cH')$. Then $\cO=\cO(\pi(G))$ is Kochen-Specker contextual if and only if $\cO'=\cO(\pi'(G))$ is Kochen-Specker contextual.
\end{proposition}

\begin{proof}
    Since $\CO\cong\cC(\cO')$ this follows immediately from Thm.~\ref{thm: CNC for OAs}.
\end{proof}

Prop.~\ref{prop: different unital realisations are KS equivalent} holds for all inequivalent realisations (in all dimensions). Note also that we do not need to restrict to rank-$1$ projections as if often done. Using Prop.~\ref{prop: different unital realisations are KS equivalent}, we can thus reduce Problem \ref{prob: inequivalent realisations} to the mere existence problem (cf. Ref.~\cite{HeunenFritzReyes2014}).

\begin{problem}\label{prob: unital realisation}
    Which simple graphs admit a unital quantum realisation?
\end{problem}

Problem~\ref{prob: unital realisation} remains generally open. Even if $G$ has a unital realisation in a (finite-dimensional) observable algebra, this does not imply that $G$ has a unital quantum realisation; indeed, Lm.~\ref{lm: single state} in App.~\ref{app: OAs with separating states} provides an explicit counterexample.\\

\textbf{Freely completable quantum realisations.} For realisations with faithful completion Problem~\ref{prob: unital realisation} can be further simplified. Specifically, in order to apply the characterisation of SI-C sets in terms of the chromatic number of the orthogonality graphs of their faithful completion $\CG^*$, one is faced with the following:

\begin{problem}\label{prob: freely completable realisation}
    Which simple graphs admit a freely completable quantum realisation?
\end{problem}

In Ref.~\cite{HeunenFritzReyes2014} it was shown that every finite simple graph $G$ admits a (faithful) quantum realisation $\pi:G\ra\PH$ in a Hilbert space of dimension $\dim(\cH)<|V|$. Clearly, any realisation $\pi:G\ra\PO$ must also satisfy the lower bound $\omega(G)\leq\dim(\One)$, where $\omega(G)$ is the size of the maximal clique in $G$. Moreover, both bounds are tight in general since they coincide for complete graphs, that is, maximally connected graphs $G=(V,V\times V)$. It would therefore be interesting to see if the construction in Ref.~\cite{HeunenFritzReyes2014} can be lifted to provide an answer to Problem~\ref{prob: freely completable realisation}.\\

\textbf{Minimal quantum realisations and reduction to simple algebras.} Since Thm.~\ref{thm: n-colourability for non-maximal OAs}, Thm.~\ref{thm: partial RH conjecture} and Thm.~\ref{thm: RH conjecture} all depend on the dimension (of the realisation of an orthogonality graph), it is natural to ask for the minimal quantum realisation of a simple graph.\footnote{Still, the respective results hold irrespective of whether the realisation is minimal or not.} To this end, note that a simple graph $G$ may have a realisation with respect to a non-simple quantum algebra of the form $\pi_m:G\ra\bigoplus_{n=1}^{N(\pi)}\bigoplus_{i_n=1}^{I_n(\pi)} M_{ni_n}(\C)$ with $\dim(\cH)=\sum_{n=1}^{N(\pi)} nI_n(\pi)$ and where $I_n(\pi)\in\mathbb{N}$ counts multiplicities. This motivates to consider a fine-grained (partial) ordering on the space of matrix algebras,
\begin{align}\label{eq: algebra ordering}
    \cA\preceq\cB\quad :\Longleftrightarrow\quad\exists\phi:\cA\hookrightarrow\cB\; ,
\end{align}
where $\phi$ is an embedding, that is, an injective ($C^*$-)algebra homomorphism.\footnote{The order in Eq.~(\ref{eq: algebra ordering}) can be visualised using Bratelli diagrams for $\cA$ and $\cB$.} Clearly, Eq.~(\ref{eq: algebra ordering}) includes the natural embedding of simple algebras $M_m(\C)=\cA\preceq\cB=M_n(\C)$ if and only if $m\leq n$. Moreover, it implies an ordering between non-simple algebras acting on the same Hilbert space, e.g. $M_1(\C)\oplus M_2(\C)\prec M_3(\C)$. Finally, the ordering is partial, for instance, let $\cA=M_2(\C)\oplus M_2(\C)$ and $\cB=M_1(\C)\oplus M_3(\C)$, then neither $\cA\prec\cB$ nor $\cB\prec\cA$. We now define a minimal realisation in the following way.

\begin{definition}\label{def: minimal realisation}
    Let $G$ be a simple graph. Then a (quantum) realisation $\pi:G\ra\cP(\cA)$ with respect to a finite-dimensional matrix algebra $\cA$ is called minimal if for any other realisation $\pi':G\ra\cP(\cB)$ there exists an injective homomorphism $\phi:\cA\ra\cB$.
\end{definition}

Note that the notion of minimal realisation in Def.~\ref{def: minimal realisation} does not distinguish between inequivalent realisations in the same (Hilbert space) dimension as in Problem \ref{prob: inequivalent realisations}. Since $\prec$ is a partial order, Def.~\ref{def: minimal realisation} a priori allows for minimal realisations with respect to incomparable algebras under the ordering in Eq.~(\ref{eq: algebra ordering}). Nevertheless, this is not the case.

\begin{lemma}\label{lm: restriction to simple algebras}
    Let $G$ be a simple graph with minimal quantum realisations $\pi:G\ra\cP(\cA)$ and $\pi':V\ra\cP(\cB)$. Then $\cA\cong\cB$.
\end{lemma}

\begin{proof}
    For every summand $\cA_k=M_{n_k}(\C)$ in $\cA=\oplus_k\cA_k$, there exists a subgraph $G_k\subset G$ such that $\pi_k=\pi|_{G_k}:G_k\ra\cP(\cH_k)$ with $\dim(\cH_k)=n_k$ is a realisation of $G_k$. Moreover, we have $G=G_k\times G_{\overline{k}}$ if $\pi$ is faithful and constraint-preserving (and $G\subset G_k\times G_{\overline{k}}$ otherwise), and where $\times$ denotes the join of two graphs (defined from the union of the respective graphs, where every vertex in one factor is additionally adjacent to every vertex in the other). By minimality of $\pi'$, we must then have $\pi'(G_k)=\cB_k$ for some summand $\cB_k$ in $\cB=\oplus_k\cB_k$. Consequently, there exists a unitary $u\in\UH$ such that $u\cA u^*=\oplus_k u\cA_ku^*=\oplus_k \cB_k=\cB$.
\end{proof}

By decomposing a graph into parts $G_k$ according to the decomposition of a minimal realisation, Lm.~\ref{lm: restriction to simple algebras} therefore allows us to reduce the search for quantum contextuality scenarios to simple algebras. In particular, we arrive at the following graph-theoretic refinement of the ``Kochen-Specker problem" (defined in Eq.~($*$) in Ref.~\cite{Frembs2024a}), which asks to characterise all observable subalgebras $\cO\subset\LHsa$ that are KS contextual.

\begin{problem}\label{prob: graph-reduced KS problem}
    Characterise all simple graphs $G$ with minimal quantum realisation $\pi:G\ra\PH$ in a given dimension $3<\dim(\cH)$.
\end{problem}

Finally, minimal realisation in Prob.~\ref{prob: graph-reduced KS problem} will generally not be composed of rank-$1$ projections. Consequently, quantum realisations restricted to rank-$1$ projections will generally have larger dimension than minimal quantum realisations without this restriction. Moreover, even if a quantum realisation exists it is not clear whether $\dim^*(\One)=\dim^*(\mathbbm{1})$.
\section{Critical comparison with results in Ref.~\cite{Cabello2012b,RamanathanHorodecki2014,CabelloKleinmannBudroni2015}}\label{app: correction}

In this section, we show that Thm.~1 in Ref.~\cite{Cabello2012b} and Thm.~4 in Ref.~\cite{CabelloKleinmannBudroni2015} which also follow from Thm.~1 in Ref.~\cite{RamanathanHorodecki2014} are false without restricting to faithful graph completions. In short, $\chi(\oG^*)>d$ does not imply $\chi(G)>d$ for general realisations $\pi:G\ra\PHone$.

We first remark that Ref.~\cite{RamanathanHorodecki2014,CabelloKleinmannBudroni2015} consider the completion of orthogonality graphs and their (not necessarily unital) realisations implicitly.\footnote{We point out that Thm.~1 (specifically, the measure of contextuality defined in Eq.~(1)) in Ref.~\cite{RamanathanHorodecki2014} already (implicitly) considers the completion of an orthogonality graph $G$. In particular, it compares quantum and classical correlations with respect to $\oG$ - yet, note that, by Prop.~\ref{prop: no marginalisation}, classical correlations on $\oG$ are generally unrelated with classical correlations on $G$ (see also App.~\ref{app: correction}).\label{fn: implicitly complete}} In particular, the distance defined in Eq.~(1) in Ref.~\cite{RamanathanHorodecki2014} is with respect to classical states on the completion of the respective orthogonality graph. However, classical correlations on $\oG$ generally do not restrict to classical states on $G$ in a straightforward manner (see App.~\ref{app: orthogonality graphs vs OAs}), indicating that Thm.~1 in Ref.~\cite{RamanathanHorodecki2014} does not apply without considering completions.

An explicit counterexample is shown in Fig.~\ref{fig: Yu-Oh graph}, and is based on the SI-C set in Ref.~\cite{YuOh2012}, which is the completion of the following $13$ vectors (see Fig.~\ref{fig: Yu-Oh graph} (a))
\begin{align}\label{eq: 13 vectors}
    z_1&=(1,0,0) & z_2&=(0,1,0) & z_3&=(0,0,1) \nonumber\\
    y^+_1&=(0,1,1) & y^+_2&=(1,0,1) & y^+_3&=(1,1,0) \\
    y^-_1&=(0,1,\bar{1}) & y^-_2&=(1,0,\bar{1}) & y^-_3&=(1,\bar{1},0) \nonumber\\
    h_0&=(1,1,1) & h_1&=(\bar{1},1,1) & h_2&=(1,\bar{1},1) & h_3&=(1,1,\bar{1})\; ,\nonumber
\end{align}
where $\bar{1}=-1$. These vectors give rise to the following $16$ contexts:
\begin{align*}
    C(z_1,z_2,z_3) & &C(z_1,y^+_1,y^-_1) & &C(z_2,y^+_2,y^-_2) & &C(z_3,y^+_3,y^-_3) \\
    C(h_0,y^-_1,x_{01}) & &C(h_1,y^-_1,x_{11}) & &C(h_2,y^+_1,x_{21}) & &C(h_3,y^+_1,x_{31}) \\
    C(h_0,y^-_2,x_{02}) & &C(h_1,y^+_2,x_{12}) & &C(h_2,y^-_2,x_{22}) & &C(h_3,y^+_2,x_{32}) \\
    C(h_0,y^-_3,x_{03}) & &C(h_1,y^+_3,x_{13}) & &C(h_2,y^+_3,x_{23}) & &C(h_3,y^-_3,x_{33})
\end{align*}
where $C(v_1,v_2,v_3)$ denotes the context spanned by the projections onto the rays defined by the vectors $v_1$, $v_2$ and $v_3$. Note that $v_3$ is thus uniquely defined by $v_1,v_2$. It follows that the same SI-C set arises from the completion of the $15$ vectors in Fig.~\ref{fig: Yu-Oh graph} (b), where $h_0$ is replaced by the vectors $x_{01}=(-2,1,1)$, $x_{02}=(1,-2,1)$ and $x_{03}=(1,1,-2)$. Yet, for the respective orthogonality graphs (see Fig.~\ref{fig: Yu-Oh graph}) we find $\chi(G_\mathrm{YO})=4$ while $\chi(G'_\mathrm{YO})=3$. Consequently, $\chi(G)>d$ is not a necessary criterion for the graph completion $\overline{\pi(V)}=\langle\pi(V),\one\rangle$ of $G$ with respect to a generic realisation $\pi$ to be a SI-C set.

\begin{figure}[htb!]
    \centering
    \begin{subfigure}{0.49\textwidth}
        \scalebox{0.9}{\begin{tikzpicture}
    \node (0) at (0, 3.75) {\color{red}$\bullet$};

    \node(0L) [above=-0.1cm of 0]   {$z_1$};
    \node (1) at (-3, -1) {\color{blue}$\bullet$};
    \node(1L) [left=-0.1cm of 1]   {$z_2$};
    \node (2) at (3, -1) {\color{green}$\bullet$};
    \node(2L) [right=-0.1cm of 2]   {$z_3$};
    \node (3) at (0, -1) {\color{red}$\bullet$};
    \node(3L) [below=-0.1cm of 3]   {$h_2$};
    \node (4) at (0, 0.75) {$\bullet$};
    \node(4L) [below=-0.1cm of 4]   {$h_0$};
    \node (5) at (1.475, 1.425) {\color{blue}$\bullet$};
    \node(5L) [right=-0.1cm of 5]   {$h_3$};
    \node (6) at (-1.475, 1.4) {\color{green}$\bullet$};
    \node(6L) [left=-0.1cm of 6]   {$h_1$};
    \node (7) at (-1.425, -1) {\color{green}$\bullet$};
    \node(7L) [below=-0.2cm of 7]   {$y^-_2$};
    \node (8) at (-2.225, 0.225) {\color{red}$\bullet$};
    \node(8L) [left=-0.1cm of 8]   {$y^+_2$};
    \node (9) at (2.25, 0.225) {\color{red}$\bullet$};
    \node(9L) [right=-0.1cm of 9]   {$y^-_3$};
    \node (10) at (1.55, -1) {\color{blue}$\bullet$};
    \node(10L) [below=-0.2cm of 10]   {$y^+_3$};
    \node (11) at (0.8, 2.475) {\color{green}$\bullet$};
    \node(11L) [right=-0.1cm of 11]   {$y^+_1$};
    \node (12) at (-0.775, 2.475) {\color{blue}$\bullet$};
    \node(12L) [left=-0.1cm of 12]   {$y^-_1$};

    \draw [bend left=60] (1.center) to (0.center);
    \draw [bend left=60] (0.center) to (2.center);
    \draw [bend right=60] (1.center) to (2.center);
    \draw (8.center) to (7.center);
    \draw (12.center) to (11.center);
    \draw (9.center) to (10.center);
    \draw (4.center) to (9.center);
    \draw (12.center) to (4.center);
    \draw (7.center) to (4.center);
    \draw [bend left=330] (5.center) to (8.center);
    \draw [bend left, looseness=0.75] (11.center) to (3.center);
    \draw [bend right, looseness=0.75] (6.center) to (10.center);
    \draw (0.center) to (1.center);
    \draw (0.center) to (2.center);
    \draw (1.center) to (2.center);
\end{tikzpicture}}
        \subcaption{}
    \end{subfigure}
    \begin{subfigure}{0.49\textwidth}
        \scalebox{0.9}{\begin{tikzpicture}
    \node (0) at (0, 3.75) {\color{red}$\bullet$};

    \node(0L) [above=-0.1cm of 0]   {$z_1$};
    \node (1) at (-3, -1) {\color{blue}$\bullet$};
    \node(1L) [left=-0.1cm of 1]   {$z_2$};
    \node (2) at (3, -1) {\color{green}$\bullet$};
    \node(2L) [right=-0.1cm of 2]   {$z_3$};
    \node (3) at (0, -1) {\color{red}$\bullet$};
    \node(3L) [below=-0.1cm of 3]   {$h_2$};
    
    \node (41) at (-0.25, 1) {\color{green}$\bullet$};
    \node(41L) [left=-0.2cm of 41,yshift=-0.2cm]   {$x_{01}$};
    \node (42) at (0.25, 1) {\color{blue}$\bullet$};
    \node(42L) [above=-0.2cm of 42,xshift=0.2cm]   {$x_{03}$};
    \node (43) at (0, 0.5) {\color{red}$\bullet$};
    \node(43L) [below=-0.1cm of 43,xshift=0.2cm]   {$x_{02}$};
    
    \node (5) at (1.475, 1.425) {\color{blue}$\bullet$};
    \node(5L) [right=-0.1cm of 5]   {$h_3$};
    \node (6) at (-1.475, 1.4) {\color{green}$\bullet$};
    \node(6L) [left=-0.1cm of 6]   {$h_1$};
    \node (7) at (-1.425, -1) {\color{green}$\bullet$};
    \node(7L) [below=-0.2cm of 7]   {$y^-_2$};
    \node (8) at (-2.225, 0.225) {\color{red}$\bullet$};
    \node(8L) [left=-0.1cm of 8]   {$y^+_2$};
    \node (9) at (2.25, 0.225) {\color{red}$\bullet$};
    \node(9L) [right=-0.1cm of 9]   {$y^-_3$};
    \node (10) at (1.55, -1) {\color{blue}$\bullet$};
    \node(10L) [below=-0.2cm of 10]   {$y^+_3$};
    \node (11) at (0.8, 2.475) {\color{green}$\bullet$};
    \node(11L) [right=-0.1cm of 11]   {$y^+_1$};
    \node (12) at (-0.775, 2.475) {\color{blue}$\bullet$};
    \node(12L) [left=-0.1cm of 12]   {$y^-_1$};

    \draw [bend left=60] (1.center) to (0.center);
    \draw [bend left=60] (0.center) to (2.center);
    \draw [bend right=60] (1.center) to (2.center);
    \draw (8.center) to (7.center);
    \draw (12.center) to (11.center);
    \draw (9.center) to (10.center);

    \draw (12.center) to (41.center);
    \draw (9.center) to (42.center);
    \draw (7.center) to (43.center);
    
    \draw [bend left=330] (5.center) to (8.center);
    \draw [bend left, looseness=0.75] (11.center) to (3.center);
    \draw [bend right, looseness=0.75] (6.center) to (10.center);
    \draw (0.center) to (1.center);
    \draw (0.center) to (2.center);
    \draw (1.center) to (2.center);
\end{tikzpicture}}
        \subcaption{}
    \end{subfigure}
    \caption{Orthogonality graphs (a) $G_\mathrm{YO}$ of 13 three-dimensional vectors in Eq.~(\ref{eq: 13 vectors}) (taken from Ref.~\cite{YuOh2012}), and (b) $G'_\mathrm{YO}$ of the 15 vectors obtained by replacing $h_0$ with $x_{10}$, $x_{20}$ and $x_{30}$. Both graphs share the same completion $\overline{\pi(V_\mathrm{YO})}=\langle\pi(V_\mathrm{YO}),\one\rangle=\langle\pi(V'_\mathrm{YO}),\one\rangle=\overline{\pi(V'_\mathrm{YO})}$, yet $\chi(G_\mathrm{YO})=4$ while $\chi(G'_\mathrm{YO})=3$.}
    \label{fig: Yu-Oh graph}
\end{figure}
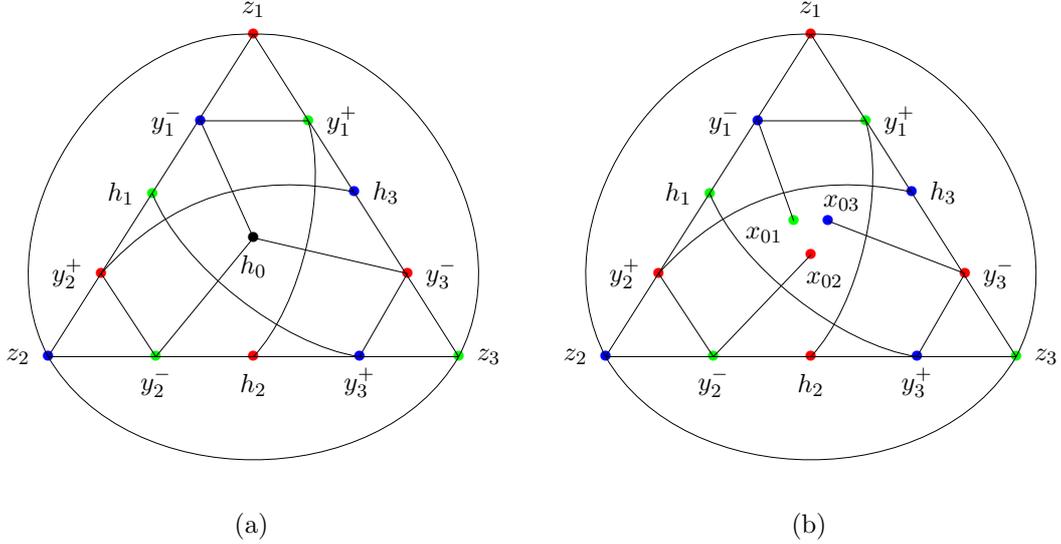

\end{document}